\documentclass[journal]{IEEEtran}
\input{psfig.sty}
\input{epsf.sty}
\usepackage{hyperref} 
\usepackage{fancybox}   
\usepackage{psfig}      
\usepackage{graphicx}
\usepackage{amsmath}
\usepackage{color}
\usepackage{epsfig}
\usepackage{amssymb}
\usepackage{verbatim}
\usepackage{enumitem}

 \usepackage[T1]{fontenc}
 \usepackage[english]{babel}
 \usepackage[font=small,labelfont=bf]{caption}

\IEEEoverridecommandlockouts

\newcommand{\remove}[1]{}

\newtheorem{thm}{Theorem}
\newtheorem{cor}{Corollary}
\newtheorem{lem}{Lemma}
\newtheorem{prop}{Proposition}

\newtheorem{rem}{Remark}

\newtheorem{claim}{Claim}

\newtheorem{condition}{Condition}

\newcommand{\qed}{\hfill \ensuremath{\Box}}

\usepackage{setspace}
\usepackage{subfigure}

\begin{document}

\title{Optimal Capacity Relay Node Placement in a Multi-hop Wireless Network on a Line\thanks{This paper is a substantial 
extension of the workshop paper \cite{chattopadhyay-etal12optimal-capacity-relay-placement-line}. 
The work was supported by the Department of Science and Technology (DST), India, through the J.C. Bose 
Fellowship and a project funded  by the Department of Electronics and Information Technology, India, and 
the National Science Foundation, 
USA, titled ``Wireless Sensor Networks for Protecting Wildlife and Humans in Forests."}\thanks{This work was done during the 
period when M. Coupechoux was a Visiting Scientist in the ECE Deparment, IISc, Bangalore.}}

\newcounter{one}
\setcounter{one}{1}
\newcounter{two}
\setcounter{two}{2}
\newcounter{three}
\setcounter{three}{3}

\author{
Arpan~Chattopadhyay$^\fnsymbol{one}$, Abhishek~Sinha$^\fnsymbol{three}$, Marceau~Coupechoux$^\fnsymbol{two}$, and Anurag~Kumar$^\fnsymbol{one}$\\
\parbox{0.49\textwidth}{\centering $^\fnsymbol{one}$Dept. of ECE, Indian Institute of Science\\
Bangalore 560012, India\\
arpanc.ju@gmail.com, anurag@ece.iisc.ernet.in}
\hfill
\parbox{0.49\textwidth}{\centering $^\fnsymbol{two}$Telecom ParisTech and CNRS LTCI \\
Dept. Informatique et R\'eseaux\\
23, avenue d'Italie, 75013 Paris, France\\
marceau.coupechoux@telecom-paristech.fr}
\hfill
\parbox{0.59\textwidth}{\centering $^\fnsymbol{three}$Laboratory for Information and Decision Systems (LIDS)\\
Massachusetts Institute of Technology, Cambridge, MA 02139\\
sinhaa@mit.edu}
}

\maketitle
\thispagestyle{empty}

\begin{abstract}
We use information theoretic achievable rate formulas for the multi-relay channel to study the 
problem of optimal placement of relay nodes along 
the straight line joining a source node and a sink node. The achievable rate formulas that 
we use are for full-duplex radios at 
the relays and decode-and-forward relaying. For the single relay case, and individual power constraints at 
the source node and the relay node, we provide explicit formulas for the optimal relay location and the 
optimal power allocation to the source-relay 
channel, for the exponential and the power-law path-loss channel models. 
For the multiple relay case, we consider exponential path-loss and a total power constraint over the source and the relays, 
and derive an optimization problem, the solution of which provides the optimal relay locations. Numerical results 
suggest that at low attenuation the relays are mostly clustered close to the source in order to be able 
to cooperate among themselves, whereas at high attenuation they are uniformly placed and work as repeaters. 

The structure of the optimal power allocation for a given placement of the nodes, then motivates us to formulate 
the problem of impromptu (``as-you-go") placement of relays along a line of exponentially distributed length, with exponential 
path-loss, so as to minimize a cost function that is additive over hops. The hop cost trades off a capacity limiting term, 
motivated from the optimal power allocation solution, against the cost of adding a relay node. We formulate the problem 
as a total cost Markov decision process, for which we prove results for the value function, and provide insights into 
the placement policy via numerical exploration.  
\end{abstract}

\vspace{-3.0mm}
\section{Introduction}
\label{Introduction}
\vspace{-1.0mm}

Wireless interconnection of mobile user devices (such as smart 
phones or mobile computers) or wireless sensors to the wireline 
communication infrastructure is an important requirement. 
These are battery operated, resource constrained devices. Hence, 
due to the physical placement of these devices, or due to the channel 
conditions, a direct one-hop link to the infrastructure ``base-station'' 
might not be feasible. In such situations, other nodes could 
serve as \emph{relays} in order to realize a multi-hop path 
between the source device and the infrastructure. In the cellular 
context, these relays might themselves be other users' devices. 
In the wireless sensor network context, the relays could be other 
wireless sensors or battery operated radio routers deployed 
specifically as relays. In either case, the relays are also 
resource constrained and a cost might be involved in engaging 
or placing them. Hence, there arises the problem of \emph{optimal relay placement}. 
Such an optimal relay placement problem involves the joint optimization 
of node placement and of the operation of the resulting network, 
where by ``operation" we mean activities such as transmission 
scheduling, power allocation, and channel coding.

\begin{figure}[t!]
\centering
\includegraphics[scale=0.32]{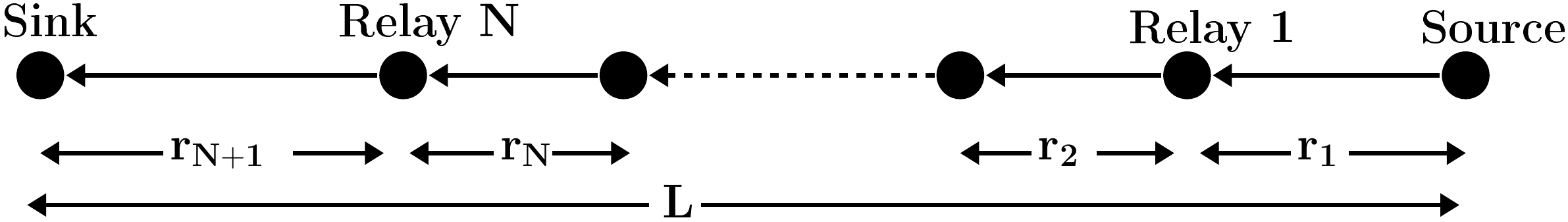}
\caption{A source and a sink connected by a multi-hop path comprising $N$ relay nodes along a line.}
\label{fig:general_line_network}
\end{figure}

 In this paper we consider the  
problem of maximizing the data rate between a source node and a sink 
by means of optimally placing relay nodes on the line segment joining the source and the sink; 
see Figure \ref{fig:general_line_network}. More specifically, we consider the two scenarios where the length $L$  of the line 
in Figure \ref{fig:general_line_network} is known, and where $L$ is an unknown random variable whose distribution is known, 
the latter being motivated by recent interest in problems of ``as-you-go" 
deployment of wireless networks in emergency situations (see Section \ref{subsection:related_work}, Related Work). 
In order to understand the fundamental trade-offs involved in such problems, we consider an 
information theoretic model. For a placement of the relay nodes along the line and allocation of transmission powers to these relays, 
we model the ``quality" of communication between the source and the 
sink by the information theoretic achievable rate of the relay 
channel\footnote{The term ``relay channel" will include the term ``multi-relay channel."}. The relays are equipped 
with full-duplex radios\footnote{See \cite{choi-etal10single-channel-full-duplex}, 
\cite{jain-etal11real-time-full-duplex} for recent efforts to realize practical full-duplex radios.}, and carry out 
decode-and-forward relaying. We consider scalar, memoryless, time-invariant, additive white Gaussian noise (AWGN) channels. 
Since we work in the information theoretic framework, we also assume synchronous operation across all transmitters and receivers. 
A path-loss model is 
important for our study, and we consider both power-law and exponential path-loss models.

\subsection{Related Work}
\label{subsection:related_work}

A formulation of the problem of relay placement requires a 
model of the wireless network at the physical (PHY) and medium access control (MAC) layers. Most researchers have adopted 
the link scheduling and interference model, i.e., 
a scheduling algorithm determines radio resource allocation (channel and power) and interference is treated 
as noise (see~\cite{georgiadis-etal06resource-allocation-cross-layer-control}). But node placement for throughput
maximization  with this model is intractable because the optimal throughput is 
obtained by first solving for the optimum schedule assuming fixed node locations, followed by an optimization over those locations. 
Hence, with 
such a model, there appears to be little work on the problem of jointly optimizing 
the relay node placement and the 
transmission schedule. In \cite{firouzabadi-martins08optimal-node-placement}, 
the authors considered placing a set of nodes in an existing network such that certain network utility (e.g., total transmit 
power) is optimized subject to a set of linear constraints on link rates. They posed the problem as one of geometric 
programming assuming exponential path-loss, and 
showed that it can be solved in a 
distributed fashion. To the best of our knowledge, there appears to be no other work which considers joint optimization of link scheduling and 
node placement using the link scheduling model.


On the other hand, an information theoretic model for a wireless network often provides a closed-form expression for the 
channel capacity, or at least an achievable rate region. These results are asymptotic, and make idealized assumptions such 
as full-duplex radios, perfect interference cancellation, etc., but provide algebraic expressions that can be used to formulate 
tractable optimization problems. The results from these formulations can provide useful insights. In the context of  
optimal relay placement, some researchers have already exploited this approach. For example, 
Thakur et al. in \cite{thakur-etal10optimal-relay-location-power-allocation} report on the problem of 
placing a single  relay node to maximize the capacity of a broadcast relay channel in a wideband regime. The linear deterministic 
channel model (\cite{avestimehr-etal11wireless-network-deterministic}) is used in \cite{appuswamy-etal10relay-placement-deterministic-line} to study the problem of placing two or more relay nodes along a line so as to maximize the end-to-end 
data rate. Our present paper is in a similar spirit; however, we use the achievable rate formulas for the $N$-relay channel 
(with decode and forward relays) to study the problem of placing relays on a line under individual node power constraints 
as well as with sum power constraints over the source and  the relays. 

Another major difference of our paper with 
\cite{thakur-etal10optimal-relay-location-power-allocation} and \cite{appuswamy-etal10relay-placement-deterministic-line} is that 
we also address the problem of sequential placement of the relay nodes along a line having unknown random length. 
Such a problem is motivated by the problem faced by ``first-responders" to emergencies (large building fires, 
buildings with terrorists and hostages), in which these personnel might need to deploy wireless sensor networks ``as-they-go." 
Howard et al., in \cite{howard-etal02incremental-self-deployment-algorithm-mobile-sensor-network}, 
provide heuristic algorithms for incremental
deployment of sensors (such as surveillance cameras) with the objective of covering the deployment area.
Souryal et al., in \cite{souryal-etal07real-time-deployment-range-extension}, address 
the problem of impromptu deployment of static wireless networks with an extensive study of indoor RF
link quality variation. More recently, Sinha et al. (\cite{sinha-etal12optimal-sequential-relay-placement-random-lattice-path}) 
have provided a Markov Decision Process based formulation of a problem  to establish a multi-hop 
network between a sink and an unknown source location by placing relay nodes along a random lattice path. 
The formulation in \cite{sinha-etal12optimal-sequential-relay-placement-random-lattice-path} 
is based on the so-called ``lone packet traffic model" under which, at any time instant, 
there can be no more than one packet traversing the network, thereby eliminating contention between wireless links. 
In our present paper, we consider the problem of as-you-go relay node placement so as to maximize a 
certain information theoretic capacity limiting term which is derived from the results on the fixed length line segment.

\subsection{Our Contribution}

\begin{itemize}
 \item In Section \ref{sec:individual_power_constraint}, we consider the problem of 
placing a single relay with individual power constraints at the source and the relay. In this context, 
we provide explicit formulas for the optimal relay location and the optimal source power split 
(between providing new information to the relay and cooperating with the relay to assist the sink),
for the exponential path-loss model (Theorem \ref{theorem:exponential_individual_power_constraint}) and for the power 
law path-loss model (Theorem \ref{theorem:power_law}). We find that at low attenuation it is better to place the relay near the source, whereas at very high 
attenuation the relay should be placed at half-distance between the source and the sink. 

\item In Section \ref{sec:total_power_constraint}, we focus on the $N$ relay placement problem with 
exponential path-loss model and a sum power constraint among the source and the relays.
For given relay locations, the optimal power split among the nodes and the achievable rate are given in Theorem 
\ref{theorem:multirelay_capacity} in terms of the channel gains. We explicitly solve the single relay placement problem in this 
context (Theorem~\ref{theorem:single_relay_total_power}). A numerical study shows that, the relay nodes are clustered 
near the source at low attenuation and are placed uniformly between the source and 
the sink at high attenuation. We have also studied the asymptotic behaviour of the achievable rate $R_N$ when $N$ relay nodes are placed uniformly on a 
line of fixed length, and  
show that for a total power constraint $P_T$ among the the source and the relays, 
$\liminf_{N \to \infty} R_N \geq C(\frac{P_T}{2 \sigma^2})$, where $C (\cdot)$ is the AWGN capacity formula and $\sigma^2$ is the power 
of the additive white Gaussian noise at each node.

\item In Section \ref{sec:mdp_total_power}, we consider the problem of placing relay nodes along a line 
of random length. Specifically, the problem is to start from a source, and walk along a line, placing relay nodes as we go, 
until the line ends, at which point the sink is placed. We are given that the distance between the source and the 
sink is exponentially distributed. With a sum power constraint over the source and the deployed relays, the aim is to 
maximize a certain capacity limiting term that is derived from the problem over a fixed length line segment, while 
constraining the expected number of relays that are used. Since the objective can be expressed as a sum of certain 
terms over the inter-relay links, we  ``relax" the expected number of relays constraint via a Lagrange multiplier, and thus 
formulate the problem as a total cost Markov Decision Process (MDP), with an uncountable state space and non-compact 
action sets. We establish the existence of an optimal policy, and convergence of value iteration, and also provide some 
properties of the value function. A numerical study supports certain intuitive structural properties of the optimal policy.
\end{itemize}

The rest of the paper is organized as follows. In Section \ref{sec:system_model_and_notation}, 
we describe our system model and notation. 
In Section \ref{sec:individual_power_constraint}, node placement with per-node power constraint 
is discussed. Node placement for total 
power constraint is discussed in Section \ref{sec:total_power_constraint}. Sequential placement of relay nodes for 
sum power constraint is discussed in Section \ref{sec:mdp_total_power}. Conclusions are 
drawn in Section \ref{conclusion}. The proofs of all the theorems are given in the appendices.

\vspace{-0.2cm}
\section{The Multirelay Channel: Notation and Review}
\label{sec:system_model_and_notation}

\subsection{Network and Propagation Models}\label{subsec:network_propagation_model}

The multi-relay channel was studied in \cite{xie-kumar04network-information-theory-scaling-law} and 
\cite{reznik-etal04degraded-gaussian-multirelay-channel} and is an extension of the single relay model presented in 
\cite{cover-gamal79capacity-relay-channel}. We 
consider a network deployed on a line with a source node, a sink node at the end of the line, 
and $N$ full-duplex relay nodes as shown in Figure 
\ref{fig:general_line_network}. The relay nodes are numbered as 
$1, 2,\cdots,N$. 
The source and sink are indexed by $0$ and $N+1$, respectively. The distance of the $k$-th node from the source is denoted by 
$y_{k}:=r_{1}+r_{2}+\cdots+r_{k}$. Thus, $y_{N+1}=L$. As in \cite{xie-kumar04network-information-theory-scaling-law} 
and \cite{reznik-etal04degraded-gaussian-multirelay-channel}, we consider the scalar, 
time-invariant, memoryless, additive white Gaussian noise setting. We do not consider propagation delays, 
hence, we use the model that a symbol transmitted by node~$i$ is received at node~$j$ after 
multiplication by the (positive, real valued) channel gain $h_{i,j}$; such a model has been used by previous researchers, see 
\cite{xie-kumar04network-information-theory-scaling-law}, \cite{gupta-kumar03capacity-wireless-networks}. 
The Gaussian additive noise at any receiver is independent 
and identically distributed from symbol to symbol and has variance $\sigma^2$. The \emph{power gain} from Node $i$ to Node $j$ is 
denoted by $g_{i,j} = h_{i,j}^2$.  We  model the power gain via two alternative path-loss models: exponential path-loss and 
power-law path-loss. The power 
gain at a distance $r$ is $e^{-\rho r}$ for the exponential path-loss model and $r^{-\eta}$ for the power-law path-loss model, where 
$\rho > 0$, $\eta >1$. These path-
loss models have their roots in the physics of wireless channels (\cite{franceschetti-etal04random-walk-model-wave-propagation}). 
For the exponential path-loss model, we will denote $\lambda:=\rho L$; $\lambda$ can be treated as a measure of attenuation 
in the line network. Under the exponential 
path-loss model, the channel gains and power gains in the line network become multiplicative, e.g., $h_{i,i+2}=h_{i,i+1}h_{i+1,i+2}$ 
and $g_{i,i+2}=g_{i,i+1}g_{i+1,i+2}$ for $i \in \{0,1,\cdots,N-1\}$. In this case, we define $g_{i,i}=1$ and $h_{i,i}=1$.
The power-law path-loss expression 
fails to characterize near-field transmission, since it goes to $\infty$ as $r \rightarrow 0$. One alternative is the 
``modified power-law path-loss'' model where the path-loss is $\min \{r^{-\eta}, b^{-\eta}\}$ with $b>0$ a reference distance. 
In this paper we consider both power-law and modified power-law path loss models, apart from the exponential path-loss model.

\subsection{An Inner Bound to the Capacity of the Multi-Relay Channel}\label{subsec:achievable_rate_multirelay_xie_and_kumar}

For the multi-relay channel, we denote the symbol transmitted by the $i$-th node at time $t$ ($t$ is discrete) by $X_{i}(t)$
 for $i=0,1,\cdots,N$. $Z_{k}(t) \sim \mathcal{N}(0,\sigma^{2})$ is the additive white Gaussian noise\footnote{We consider real 
symbols in this paper. But similar analysis will carry through if the symbols are complex-valued, so long as the channel 
gains, i.e., $h_{i,j}$-s are real positive numbers.} at 
node $k$ and time $t$, and is assumed to be independent and identically distributed across $k$ and $t$. 
Thus, at symbol time $t$, node~$k, 1 \leq k \leq N+1$ receives:
\begin{equation}
Y_{k}(t)= \sum_{j\in \{0,1,\cdots,N\}, j \neq k} h_{j,k}X_{j}(t)+Z_{k}(t) \label{eqn:network_equation}
\end{equation}
In \cite{xie-kumar04network-information-theory-scaling-law}, 
the authors showed that an inner bound to the capacity of this network is given by (defining $C(x):=\frac{1}{2×}\log_{2}(1+x)$):~\footnote{$\log (\cdot)$ in this paper will mean the natural logarithm unless the base is specified.}
\begin{equation}
 R=\min_{1 \leq k \leq N+1} C \bigg(     \frac{1}{\sigma^{2}×} \sum_{j=1}^{k} ( \sum_{i=0}^{j-1} h_{i,k} \sqrt{P_{i,j}}  )^{2}      \bigg)    \label{eqn:achievable_rate_multirelay}
\end{equation}
where $P_{i,j}$ denotes the power at which node $i$ transmits to node $j$.


In order to provide insight into the expression in (\ref{eqn:achievable_rate_multirelay}) and 
the relay placement results in this paper, we provide a descriptive overview of the coding and decoding scheme 
described in \cite{xie-kumar04network-information-theory-scaling-law}. 
Transmissions take place via block codes of $T$ symbols each. The transmission blocks at the source and the $N$ 
relays are synchronized. The coding and decoding scheme is such that a message generated at the source at the beginning 
of block $b, b \geq 1,$ is decoded by the sink at the end of block $b + N$, i.e., $N+1$ block durations after 
the message was generated (with probability tending to 1, as $T \to \infty)$. Thus, at the end of $B$ blocks, $B \geq N+1$, 
the sink is able to decode $B-N$ messages. It follows, by taking $B \to \infty$, that, if the code rate 
is $R$ bits per symbol, then an information rate of $R$ bits per symbol can be achieved from the source to the 
sink.

As mentioned earlier, we index the source by $0$, the relays 
by $k, 1 \leq k \leq N$, and the sink by $N+1$. There are  $(N+1)^2$ independent Gaussian random codebooks, each containing $2^{TR}$ codes, 
each code being of length $T$; these codebooks are available to all nodes. At the beginning of block $b$, the source 
generates a new message $w_b$, and, at this stage, we assume that 
each node $k, 1 \leq k \leq N+1,$ has a reliable estimate of all the 
messages $w_{b-j}, j \geq k$.  In block $b$, the source uses a new codebook to encode $w_b$. In addition, 
relay $k, 1 \leq k \leq N,$ and {\em all} of its previous transmitters (indexed $0 \leq j \leq k-1$), use {\em another} 
codebook to encode $w_{b-k}$ (or their estimate of it). Thus, if the relays $1,2,\cdots,k$ have a perfect estimate of $w_{b-k}$ 
at the beginning of block $b$, they will transmit the same codeword for $w_{b-k}$. Therefore, in block $b$, 
the source and relays $1, 2, \cdots, k$ \emph{coherently transmit} the codeword for $w_{b-k}$. 
In this manner, in block $b$, transmitter $k, 0 \leq k \leq N,$ 
generates $N+1 - k$ codewords, corresponding to $w_{b-k}, w_{b-k-1}, \cdots, w_{b-N}$, which are transmitted with powers 
$P_{k,k+1}, P_{k,k+2}, \cdots, P_{k,N+1}$. In block $b$, node $k, 1 \leq k \leq N+1,$ receives a superposition of transmissions 
from all other nodes. Assuming that node $k$ knows all the powers, and all the channel gains, and recalling that it has a reliable 
estimate of  all the messages $w_{b-j}, j \geq k$, it can subtract the interference from transmitters $k+1, k+2, \cdots, N$. 
At the end of block $b$, after subtracting the signals it knows, node $k$ is left with the $k$ received signals from 
nodes $0, 1, \cdots, (k-1)$ (received in blocks $b, b-1, \cdots, b-k+1$), which all carry an encoding of the message $w_{b-k+1}$. These $k$ signals are then jointly used 
to decode  $w_{b-k+1},$ using joint typicality decoding. The codebooks are cycled through in a manner so that in any block all 
nodes encoding a message (or their estimate of it) use the same codebook, but different (thus, independent) codebooks are used 
for different messages. Under this encoding and decoding scheme, relatively simple arguments lead to the conclusion 
that any rate strictly less than $R$ displayed in (\ref{eqn:achievable_rate_multirelay}) is achievable.

From the above description we see that a node receives information 
about a message in two ways (i) by the message being directed to it cooperatively by all the previous nodes, and (ii) 
by overhearing previous transmissions of the message to the previous nodes. 
Thus node $k$ receives codes corresponding to a message $k$ times before it attempts to decode the message. 
Note that, $C \bigg( \frac{1}{\sigma^{2}×} \sum_{j=1}^{k} ( \sum_{i=0}^{j-1} h_{i,k} \sqrt{P_{i,j}}  )^{2} \bigg)$ 
in (\ref{eqn:achievable_rate_multirelay}), for any $k$, denotes a possible rate that can be achieved by node $k$ from 
the transmissions from nodes $0,1,\cdots,k-1$. The smallest of these terms become the bottleneck, which 
has been reflected in (\ref{eqn:achievable_rate_multirelay}).


For the single relay channel, $N=1$. Thus, by (\ref{eqn:achievable_rate_multirelay}), 
an achievable rate is given by (see also \cite{cover-gamal79capacity-relay-channel}):

\footnotesize
 \begin{eqnarray}
R=\min & \bigg\{ & C \left(\frac{g_{0,1}P_{0,1}}{\sigma^{2}×}\right),\nonumber\\
 & & C  \left( \frac{g_{0,2}P_{0,1}+(h_{0,2}\sqrt{P_{0,2}}+h_{1,2}\sqrt{P_{1,2}})^{2}}{\sigma^{2}×}  \right)  \bigg\}\label{eqn:single_relay_genaral_capacity_formula}
\end{eqnarray}
\normalsize

Here, the first term in the $\min\{\cdot,\cdot\}$ of (\ref{eqn:single_relay_genaral_capacity_formula}) is the achievable rate at 
node $1$ (i.e., the relay node) due to the transmission from the source. The second term in the $\min\{\cdot,\cdot\}$ 
corresponds to the possible achievable rate at the sink node due to direct coherent transmission to itself from the 
source and the relay and due to the overheard transmission from the source to the relay.

The higher the channel attenuation, the less will be the contribution of 
farther nodes, ``overheard" transmissions become less relevant, 
and coherent transmission reduces to a simple transmission from the 
previous relay. The system is then closer to simple store-and-forward 
relaying.

The authors of \cite{xie-kumar04network-information-theory-scaling-law} have shown that any rate strictly less
than $R$ is achievable through the above-mentioned coding and
decoding scheme which involves coherent multi-stage relaying
and interference subtraction. This achievable rate formula can 
also be obtained from the capacity formula of a physically
degraded multi-relay channel (see \cite{reznik-etal04degraded-gaussian-multirelay-channel} for the channel model),
since the capacity of the degraded relay channel is a lower
bound to the actual channel capacity. In this paper, we will
seek to optimize $R$ in (\ref{eqn:achievable_rate_multirelay}) over power allocations to the nodes and
the node locations, keeping in mind that $R$ 
is a lower bound to the actual capacity. We denote the value of $R$ optimized over
power allocation and relay locations by $R^*$.

The following result justifies our aim to seek optimal relay placement on the line 
joining the source and the sink (rather than anywhere else on the plane). 
\begin{thm}\label{theorem:why_on_a_line}
 For given source and sink locations, the
relays should always be placed
on the line segment joining the source and the 
sink in order to maximize the end-to-end data rate between the source and the sink.
\end{thm}
\begin{proof}
 See Appendix \ref{appendix:system_model_and_notation}.
\end{proof}

\begin{figure*}[!t]
 \footnotesize
\begin{equation}
 \alpha^{*}=\frac{e^{-\lambda}}{\left(2e^{-\lambda}+e^{-\frac{\lambda}{2×}}\right)×} \left( \sqrt{1-\frac{e^{-\lambda}}{(2e^{-\lambda}+e^{-\frac{\lambda}{2×}})^{2}×}} + \sqrt{\frac{1}{(2e^{-\lambda}+e^{-\frac{\lambda}{2×}})×} \left( 1-\frac{e^{-\lambda}}{(2e^{-\lambda}+e^{-\frac{\lambda}{2×}})×} \right)}   \right)^{2}\label{eqn:optimum_alpha}
\end{equation}
\begin{equation}
 R^{*}=C \left(\frac{P}{\sigma^{2}×}e^{-\lambda} \left( \sqrt{1-\frac{e^{-\lambda}}{(2e^{-\lambda}+e^{-\frac{\lambda}{2×}})^{2}×}} + \sqrt{\frac{1}{(2e^{-\lambda}+e^{-\frac{\lambda}{2×}})×} \left( 1-\frac{e^{-\lambda}}{(2e^{-\lambda}+e^{-\frac{\lambda}{2×}})×} \right)}   \right)^{2} \right)\label{eqn:theorem_capacity_high_attenuation}
\end{equation}
\normalsize
\end{figure*}

\vspace{3mm}
\section{Single Relay Node Placement: Node Power Constraints}
\label{sec:individual_power_constraint}
\vspace{2mm}

In this section, we aim at placing a single relay node between the source and the sink in order to maximize the achievable 
rate $R$ given by (\ref{eqn:single_relay_genaral_capacity_formula}). Let the distance between the source and the 
relay be $r$, i.e., $r_{1}=r$. Let $\alpha :=\frac{P_{0,1}}{P_{0}×}$. We assume that the source 
and the relay use the same transmit power $P$. Hence, $P_{0,1}=\alpha P$,
$P_{0,2}=(1-\alpha)P$ and $P_{1,2}=P$. Thus, for a given placement of the relay node, we obtain:

\footnotesize
\begin{eqnarray}
R= \min & \bigg\{& C  \left(\frac{\alpha g_{0,1}P}{\sigma^{2}×} \right),\nonumber\\
& & C  \left( \frac{P}{\sigma^{2}×}\left(g_{0,2}+g_{1,2}+2 \sqrt{(1-\alpha)g_{0,2}g_{1,2}}\right) \right)\bigg\} \label{eqn:capacity_single_relay_equal_power}
\end{eqnarray}
\normalsize

Maximizing over $\alpha$ and $r$, and exploiting the monotonicity of $C(\cdot)$ yields the following problem: 
\footnotesize
{
\begin{equation}
 \max_{r \in [0,L]} \, \max_{\alpha \in [0,1]} \min \bigg\{\alpha g_{0,1},\,g_{0,2}+g_{1,2}+2 \sqrt{(1-\alpha)g_{0,2}g_{1,2}}  \bigg\} \label{eqn:optimize_fd_relay_df}
\end{equation}
}
\normalsize
Note that $g_{0,1}$ and $g_{1,2}$ in the above equation depend on $r$, but $g_{0,2}$ does not depend on $r$.

\subsection{Exponential Path Loss Model}
\subsubsection{Optimum Relay Location}
Here $g_{0,1}=e^{-\rho r}$, $g_{1,2}=e^{-\rho (L-r)}$ and $g_{0,2}=e^{-\rho L}$. Let $\lambda:=\rho L$, 
and the optimum relay location be $r^{*}$. Let $x^{*}:=\frac{r^{*}}{L×}$ be the normalized optimal relay location. Let $\alpha^{*}$ be 
the optimum value of $\alpha$ when relay is optimally placed.

\begin{thm} For $N=1$, under the exponential path-loss model, if the source and relay have same 
power constraint $P$, then there is a unique optimum relay location $r^*$ and the following hold:
 \begin{enumerate}[label=(\roman{*})]

\item $x^*=\max \{x_{+},0\}$ where $x_{+}:=-\frac{1}{\lambda×} \log \left(2e^{-\lambda}+e^{-\frac{\lambda}{2}}\right)$.
  \item For $0 \leq \lambda \leq \log 2$, $x^{*}=0$, $\alpha^{*}=1$ and $R^{*}=C(\frac{P}{\sigma^{2}×})$.
   \item For $\log 2 \leq \lambda \leq \log 4$, $x^{*}=0$, $\alpha^{*}=4e^{-\lambda}(1-e^{-\lambda})$ and 
$R^{*}=C(\frac{P}{\sigma^{2}×}4e^{-\lambda}(1-e^{-\lambda}))$.
\item For $\lambda \geq \log 4$, $x^{*}=x_{+}$. Moreover, $\alpha^{*}$ and $R^{*}$ are given by (\ref{eqn:optimum_alpha}) 
and (\ref{eqn:theorem_capacity_high_attenuation}).
 \end{enumerate}
\label{theorem:exponential_individual_power_constraint}
\end{thm}
\begin{proof}
 See Appendix \ref{appendix:exponential_individual}.
\end{proof}

\begin{figure}[!t]
\centering
\includegraphics[scale=0.3]{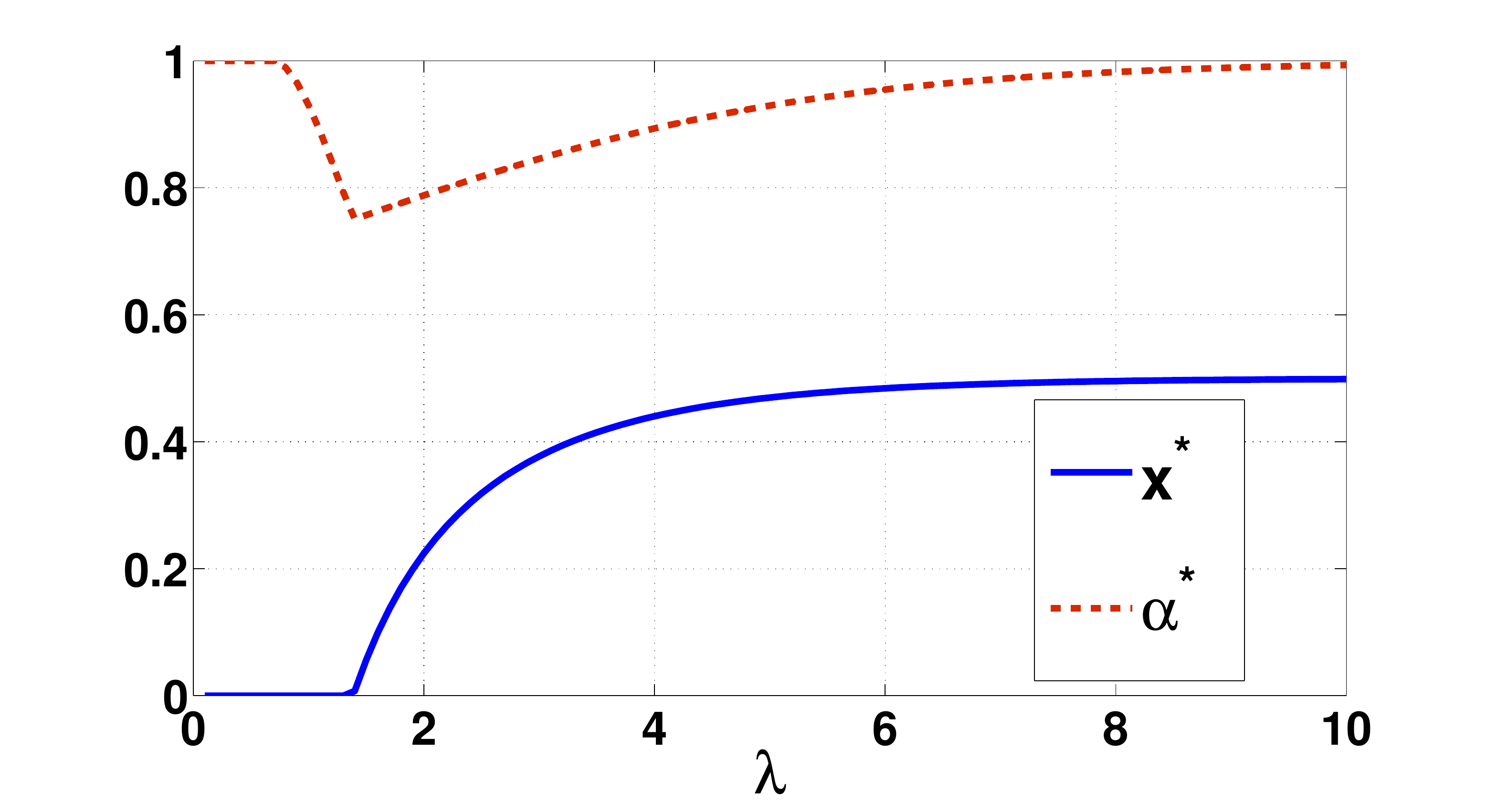}
\caption{Single Relay, exponential path-loss, node power constraint: $x^{*}$ and $\alpha^{*}$ versus $\lambda$.}
\label{fig:exponential_full_duplex_df_optimum_position}
\vspace{-5mm}
\end{figure}

\subsubsection{Numerical Work} In Figure 
\ref{fig:exponential_full_duplex_df_optimum_position}, we plot $x^*$ and $\alpha^*$ provided by Theorem 
\ref{theorem:exponential_individual_power_constraint}, versus $\lambda$. Recalling the discussion of the coding/decoding 
scheme in Section \ref{sec:system_model_and_notation}, we note that the relay provides 
two benefits to the source: (i) being nearer to the source than the sink, the power 
required to transmit a message to the relay is less than that to the sink,  and (ii) having received the message, the relay 
coherently assists the source in resolving ambiguity at the sink. These two benefits correspond to the  two terms 
in (\ref{eqn:single_relay_genaral_capacity_formula}). Hence, Theorem \ref{theorem:exponential_individual_power_constraint} and 
Figure \ref{fig:exponential_full_duplex_df_optimum_position} provide the following insights:
 
\begin{itemize}
 \item At very low attenuation ($\lambda \leq \log 2$), $x^{*}=0$ (since $x_{+} \leq 0$). Since $g_{1,2}\geq \frac{1}{2×}$ and 
$g_{0,2} \geq \frac{1}{2×}$, we will always have $g_{0,2}+g_{1,2}+2 \sqrt{(1-\alpha)g_{0,2}g_{1,2}} \geq \alpha g_{0,1}$ for all 
$\alpha \in [0,1]$ and for all $r \in [0,L]$. The minimum of the two terms in (\ref{eqn:optimize_fd_relay_df}) is maximized by 
$\alpha=1$ and $r=0$.

Note that the maximum possible achievable rate from the source to the sink 
for any value of $\lambda$ is $C \bigg( \frac{P}{\sigma^{2}×} \bigg)$, since the source 
has a transmit power $P$. In the low attenuation regime, we can dedicate the entire source power to the relay and place 
the relay at the source, so that the relay can receive at a rate $C \bigg( \frac{P}{\sigma^{2}×} \bigg)$ 
from the source. The attenuation 
is so small that even at this position of the relay, the sink can receive data 
at a rate $C \bigg( \frac{P}{\sigma^{2}×} \bigg)$ by receiving data from the relay and overhearing the transmission from 
the source to the relay; since $\alpha^{*}=1$, there is no coherent transmission involved in this situation.

\item For $\log 2 \leq \lambda \leq \log 4$, we again have $x^* = 0$. $\alpha^{*}$ decreases from $1$ to $\frac{3}{4×}$ 
as $\lambda$ increases. 
In this case since attenuation is higher, the source needs to direct some power 
to the sink for transmitting coherently with the relay to the sink to balance the source to relay data rate. 
$\lambda= \log 4$ is the critical value of $\lambda$ 
at which source to relay channel ceases to be the bottleneck. Hence, 
for $\lambda \in [\log 2, \log 4]$, we still place the relay at the 
source but the source reserves some of its power to transmit to the sink coherently 
with the relay's transmission to the sink. As attenuation increases, the source has to direct more 
power to the sink, and so $\alpha^*$ decreases with $\lambda$.

\item For $\lambda \geq \log 4$, $x^* \geq 0$. Since attenuation is high, the sink overhears less the source to relay transmissions 
and the relay to sink transmissions become more important. Hence, 
it is no longer optimal to keep the relay close to the source. Thus, 
$x^{*}$ increases with $\lambda$. Since the link between the source and the sink 
becomes less and less useful as $\lambda$ increases, we dedicate less and less power for the direct transmission 
from the source to the sink. Hence, $\alpha^{*}$ increases with $\lambda$. 
As $\lambda \rightarrow \infty$, $x^{*} \rightarrow \frac{1}{2×}$ and $\alpha^{*}\rightarrow 1$. 
We observe that the ratio of powers received by the sink from the source and the relay 
is less than $e^{-\frac{\lambda}{2×}}$, since $x^{*}\leq\frac{1}{2×}$. This ratio tends to zero as $\lambda \rightarrow \infty$. 
Thus, at high attenuation, the source transmits at full power to the relay and the relay acts just as a repeater.

\item We observe that $x^{*} \leq \frac{1}{2×}$. The two data rates in 
(\ref{eqn:capacity_single_relay_equal_power}) can be equal only if $r \leq \frac{L}{2×}$. Otherwise we will readily have 
$g_{0,1} \leq g_{0,2}+g_{1,2}$, which means that the two rates will not be equal.
\end{itemize}

\vspace{-2mm}
\subsection{Power Law Path Loss Model}
\subsubsection{Optimum Relay Location}

For the power-law path-loss model, $g_{0,1}=r^{-\eta}$, $g_{0,2}=L^{-\eta}$ and $g_{1,2}=(L-r)^{-\eta}$. Let $x^{*}$ 
be the normalized optimum relay location which maximizes $R$. Then the following theorem states how to compute $x^{*}$.

\begin{thm}\label{theorem:power_law}

\begin{enumerate}[label=(\roman{*})]
 \item For the single relay channel and power-law path-loss model (with $\eta >1$), there is a unique 
optimal placement point for the relay on the 
line joining
the source and the sink, maximizing $R$. The normalized distance of the point from source node is $x^{*}$,
 where $x^{*}$ is 
precisely the unique real root $p$ of the Equation 
$(x^{-\eta+1}-1)^{2}(1-(\frac{1}{x×}-1)^{-\eta})=(1-x)^{-\eta}-(\frac{1}{x×}-1)^{-\eta}$ in the 
interval $(0,\frac{1}{2×})$. 
\item For the ``modified power-law path-loss'' model with $2b<L$, the normalized optimum relay location is $x^{*}=\max \{p,\frac{b}{L×} \}$.
\end{enumerate}
\end{thm}
\begin{proof}
 See Appendix \ref{appendix:exponential_individual}.
\end{proof}

\begin{figure}[t!]
\centering
\includegraphics[scale=0.3]{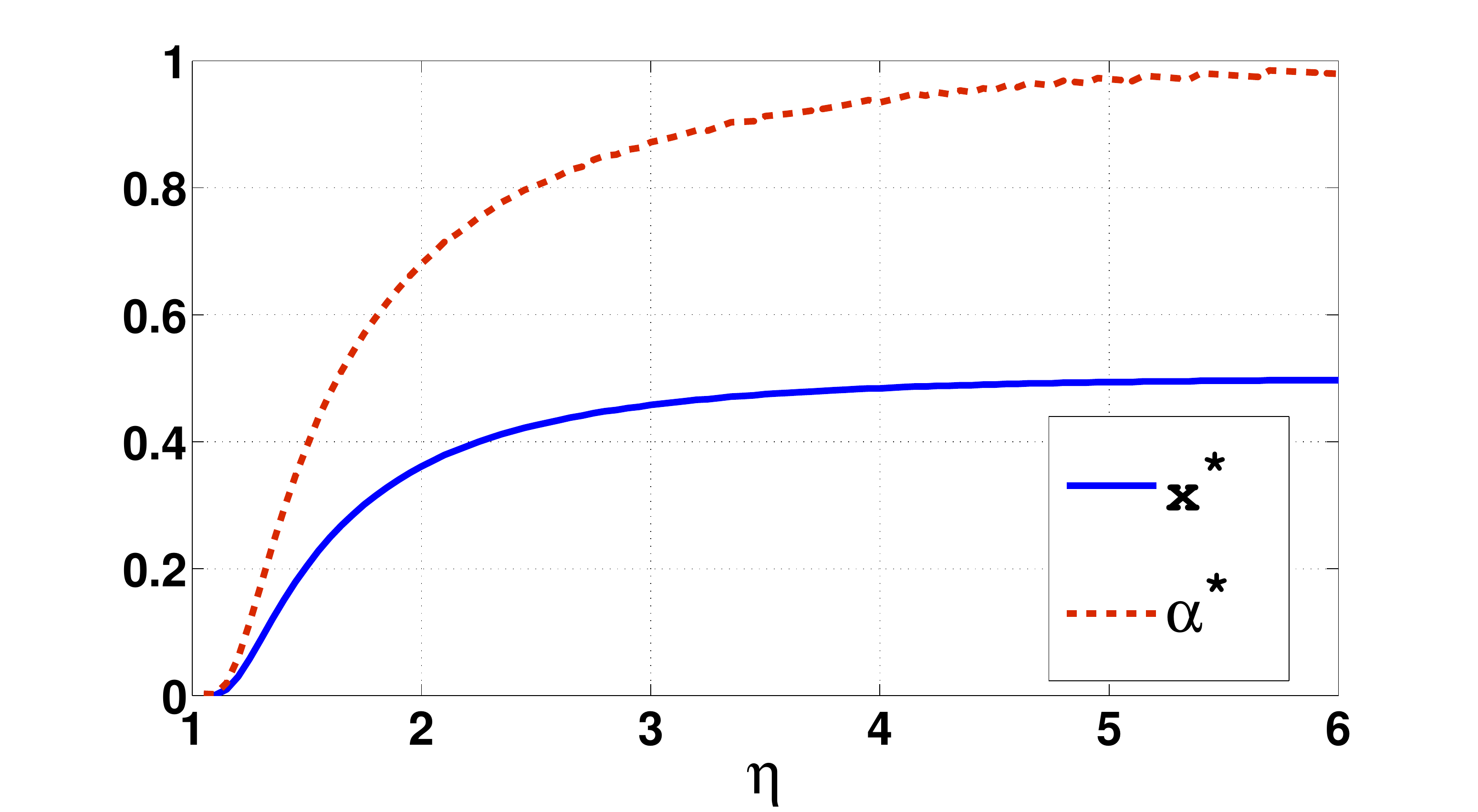}
\caption{Power law path-loss, single relay node, individual power constraint: $x^{*}$ and $\alpha^{*}$ versus $\eta$.}
\label{fig:power_law_full_duplex_df_position}
\vspace{-3mm}
\end{figure}

\subsubsection{Numerical Work} The variation of $x^{*}$ and $\alpha^{*}$ as a function of $\eta$ 
for the power-law path loss model are shown in Figure \ref{fig:power_law_full_duplex_df_position}. 
As $\eta$ increases, both $x^{*}$ and $\alpha^{*}$ increase. For large $\eta$, they are close to $0.5$ and $1$ respectively, which 
means that the relay works just as a repeater. At low attenuation the 
behaviour is different from exponential path-loss because in that case 
the two rates in (\ref{eqn:capacity_single_relay_equal_power}) can be equalized since channels gains are unbounded. For 
small $\eta$, relay is placed close to the source, $g_{0,1}$ is high and hence a small $\alpha$ suffices to equalize the two rates. 
Thus, we are in a situation similar to the $\lambda \geq \log 4$ case for exponential path-loss 
model. Hence, $\alpha^*$ and $x^*$ increase with $\lambda$.

On the other hand, the variation of $x^{*}$ and $\alpha^{*}$ with $\eta$ under the modified power-law path-loss model are shown 
in Figure~\ref{fig:modified_power_law_full_duplex_df_position}, with $\frac{b}{L×}=0.1$. 
Note that, for a fixed $\eta$ and $L$, $g_{0,1}$ remains constant at $b^{-\eta}$ 
if $r \in [0,b]$, but $g_{1,2}$ achieves its maximum over this set of relay locations 
at $r=b$. 
Hence, for any $\alpha$, the achievable rate is maximized at $r=b$ over the interval $r \in [0,b]$. 
For small values of $\eta$, 
$p \leq 0.1$ and hence $x^{*}=0.1$. 
Beyond the point where $0.1$ is the solution of 
$(x^{-\eta+1}-1)^{2}(1-(\frac{1}{x×}-1)^{-\eta})=(1-x)^{-\eta}-(\frac{1}{x×}-1)^{-\eta}$, 
the behaviour is similar to the power-law path-loss model. However, for those values of $\eta$ which result in 
$p \leq \frac{b}{L×}$, the value of $\alpha^{*}$ decreases 
with $\eta$. This happens because in this region, if $\eta$ increases, $g_{0,1}=b^{-\eta}$ decreases at a slower rate 
compared to $g_{1,2}=(L-b)^{-\eta}$ (since $b < \frac{L}{2×}$) and $g_{1,2}$ decreases at a slower rate compared to 
$g_{0,2}=L^{-\eta}$, and, hence, more source power should be dedicated for direct transmission to the sink as $\eta$ increases. 
The variation of $\alpha^{*}$ is not similar to that in Figure~\ref{fig:exponential_full_duplex_df_optimum_position} because here, 
unlike the exponential path-loss model, $g_{0,1}$ remains constant over the region $r \in [0,b]$ and also because $\eta>1$.

\begin{figure}[t!]
\centering
\includegraphics[scale=0.3]{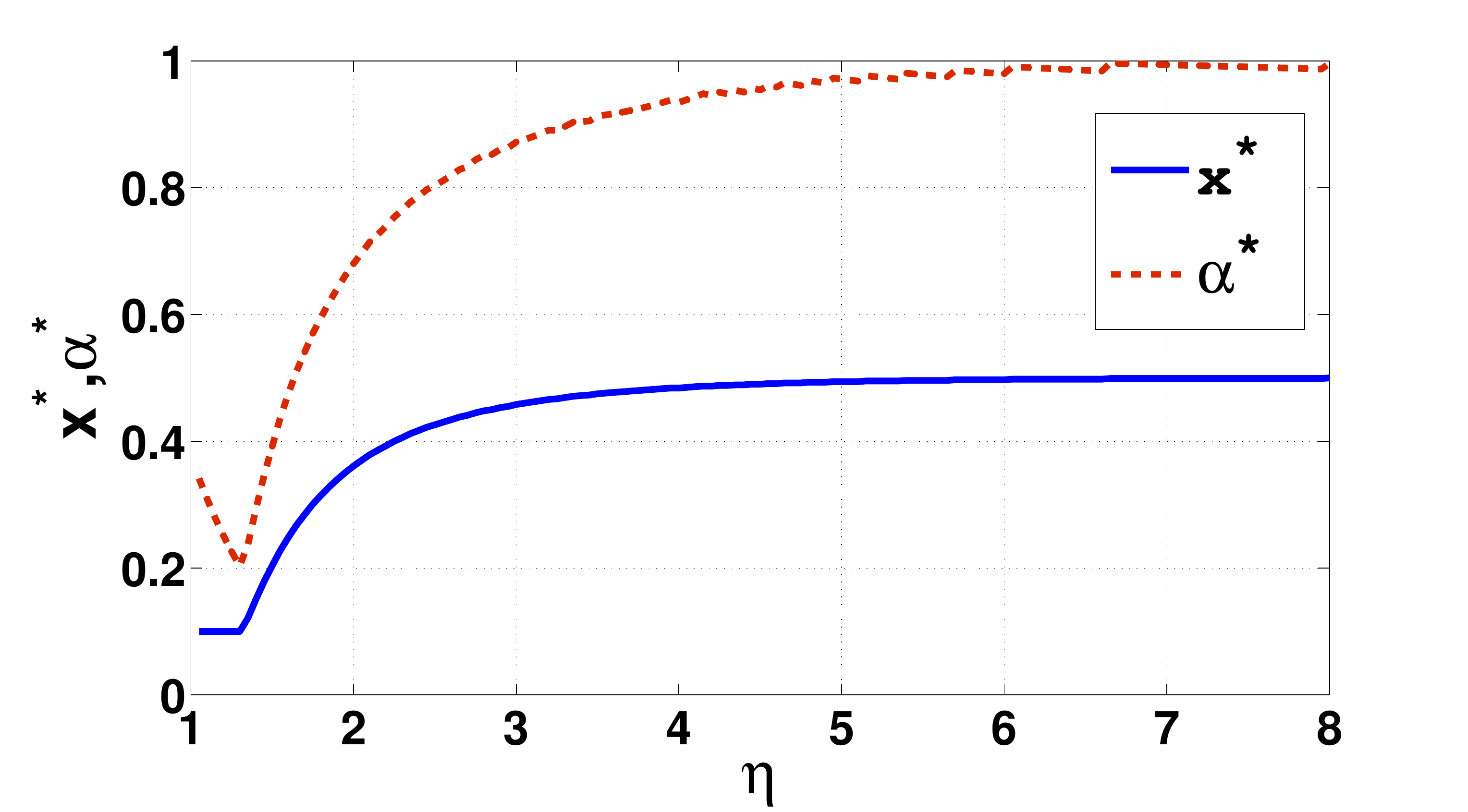}
\caption{Modified power-law path-loss (with $\frac{b}{L×}=0.1$), single relay node, individual power 
constraint: $x^{*}$ and $\alpha^{*}$ versus $\eta$.}
\label{fig:modified_power_law_full_duplex_df_position}
\vspace{-5mm}
\end{figure}

\vspace{-2mm}
\section{Multiple Relay Placement : Sum Power Constraint}
\label{sec:total_power_constraint}

In this section, we consider the optimal placement of relay nodes to maximize $R$ (see (\ref{eqn:achievable_rate_multirelay})), 
subject to a total power constraint on the source and relay nodes given by $\sum_{i=0}^{N}P_{i}=P_{T}$. We 
consider only the exponential 
path-loss model. We will first maximize 
$R$ in (\ref{eqn:achievable_rate_multirelay}) over $P_{i,j}, 0 \leq i < j \leq (N+1)$ for any given placement 
of nodes (i.e., given $y_{1}, y_{2},\cdots,y_{N}$). This will provide an expression of achievable rate in terms of channel gains, which has to 
be maximized over $y_{1}, y_{2},\cdots,y_{N}$. Let $\gamma_{k}:=\sum_{i=0}^{k-1}P_{i,k}$ for $k \in \{1,2,\cdots,N+1\}$. 
Hence, the sum power constraint becomes 
$\sum_{k=1}^{N+1}\gamma_{k}=P_{T}$. 

\begin{thm}\label{theorem:multirelay_capacity}
 \begin{enumerate}[label=(\roman{*})]
  \item For fixed location of relay nodes, the optimal power allocation that maximizes the achievable rate for the sum power constraint is given by: 
\begin{equation}
 P_{i,j}=
\begin{cases}
\frac{g_{i,j}}{\sum_{l=0}^{j-1}g_{l,j}×}\gamma_{j}\,\, & \forall 0 \leq i <j \leq (N+1) \\
0, \,\,  &\text{if}\,\,  j \leq i 
\end{cases}\label{eqn:power_gamma_relation}
\end{equation}
where
\begin{eqnarray}
\gamma_{1}&=&\frac{P_{T}×}{1+g_{0,1}\sum_{k=2}^{N+1} \frac{(g_{0,k-1}-g_{0,k})}{g_{0,k}g_{0,k-1}\sum_{l=0}^{k-1}\frac{1}{g_{0,l}×}×}  ×} \label{eqn:gamma_one}\\
 \gamma_{j}&=&\frac{g_{0,1}\frac{(g_{0,j-1}-g_{0,j})}{g_{0,j}g_{0,j-1}\sum_{l=0}^{j-1}\frac{1}{g_{0,l}×}×}×}{1+g_{0,1}\sum_{k=2}^{N+1}
  \frac{(g_{0,k-1}-g_{0,k})}{g_{0,k}g_{0,k-1}\sum_{l=0}^{k-1}\frac{1}{g_{0,l}×}×}  ×} P_{T} \,\,\,\, \forall \, j \geq 2 \label{eqn:gamma_k}\nonumber
\end{eqnarray}
\item The achievable rate optimized over the power allocation for a given placement of nodes is given by:
\footnotesize
\begin{equation}
 R^{opt}_{P_T}(y_1,y_2,\cdots,y_N)=C \bigg( \frac{\frac{P_{T}}{\sigma^{2}×}}
{\frac{1}{g_{0,1}×}+\sum_{k=2}^{N+1}
 \frac{(g_{0,k-1}-g_{0,k})}{g_{0,k}g_{0,k-1}\sum_{l=0}^{k-1}\frac{1}{g_{0,l}×}×}  ×} \bigg)\label{eqn:capacity_multirelay}
\end{equation}
 \end{enumerate}
\normalsize
\end{thm}

\begin{proof}
 The basic idea of the proof is to choose the power levels (i.e., $P_{i,j},\, 0 \leq i<j \leq N+1$) in 
(\ref{eqn:achievable_rate_multirelay}) so that all the terms in the $\min\{\cdot\}$ 
in (\ref{eqn:achievable_rate_multirelay}) become equal. This has been proved in 
\cite{reznik-etal04degraded-gaussian-multirelay-channel} via an inductive argument based on the 
coding scheme used in \cite{reznik-etal04degraded-gaussian-multirelay-channel}, for a degraded Gaussian multi-relay channel. 
However, we provide explicit expressions for $P_{i,j},\,0 \leq i<j \leq N+1$ and the 
achievable rate (optimized over power allocation) in terms of the power gains, and the proof depends on 
LP duality theory. 

See Appendix \ref{appendix:total_power_constraint} for the detailed proof.
\end{proof}

{\em Remarks and Discussion:}
\begin{itemize}
 \item We find that in order to maximize $R^{opt}_{P_T}(y_1,y_2,\cdots,y_N)$, we need to place the relay nodes such that 
$\frac{1}{g_{0,1}×}+\sum_{k=2}^{N+1}\frac{(g_{0,k-1}-g_{0,k})}{g_{0,k}g_{0,k-1}\sum_{l=0}^{k-1}\frac{1}{g_{0,l}×}×}$ is minimized. 
This quantity can be seen as the net attenuation the power $P_T$ faces.

\item When no relay is placed, the effective attenuation is $e^{\rho L}$. 
The ratio of the attenuation with no relay and the attenuation 
with relays is called the ``relaying gain'' $G$, and is defined as follows:
\begin{equation}
 G:=\frac{e^{\rho L}}{\frac{1}{g_{0,1}×}+\sum_{k=2}^{N+1}\frac{(g_{0,k-1}-g_{0,k})}{g_{0,k}g_{0,k-1}\sum_{l=0}^{k-1}\frac{1}{g_{0,l}×}×}×}
\label{eqn:gain_definition}
\end{equation}
The relaying gain $G$ captures the effect of placing the relay nodes on the SNR $\frac{P_T}{\sigma^2}$. Since, 
by placing relays we can only 
improve the rate, and we cannot increase the rate beyond $C(\frac{P_T}{\sigma^2×})$, $G$ can take values only between $1$ and 
$e^{\rho L}$. 

\item  From (\ref{eqn:power_gamma_relation}), it follows that 
any node $k<j$ will transmit at a higher power to node $j$, compared to any node preceding node $k$.

\item Note that we have derived Theorem \ref{theorem:multirelay_capacity} using the fact that $g_{0,k}$ is nonincreasing in $k$. 
If there exists some $k \geq 1$ 
such that $g_{0,k}=g_{0,k+1}$, i.e, if $k$-th and $(k+1)$-st nodes are placed at the same position, then $\gamma_{k+1} = 0$, i.e., the nodes $i < k$ do not direct any power specifically to relay $k+1$. However, 
relay $k+1$ can decode the symbols received at relay $k$, and those transmitted by relay $k$. Then
relay $(k+1)$ can transmit coherently with the nodes $l \leq k$ to improve effective received power in the nodes $j > k+1$. 

\end{itemize}

\begin{figure}[t!]
\centering
\includegraphics[scale=0.3]{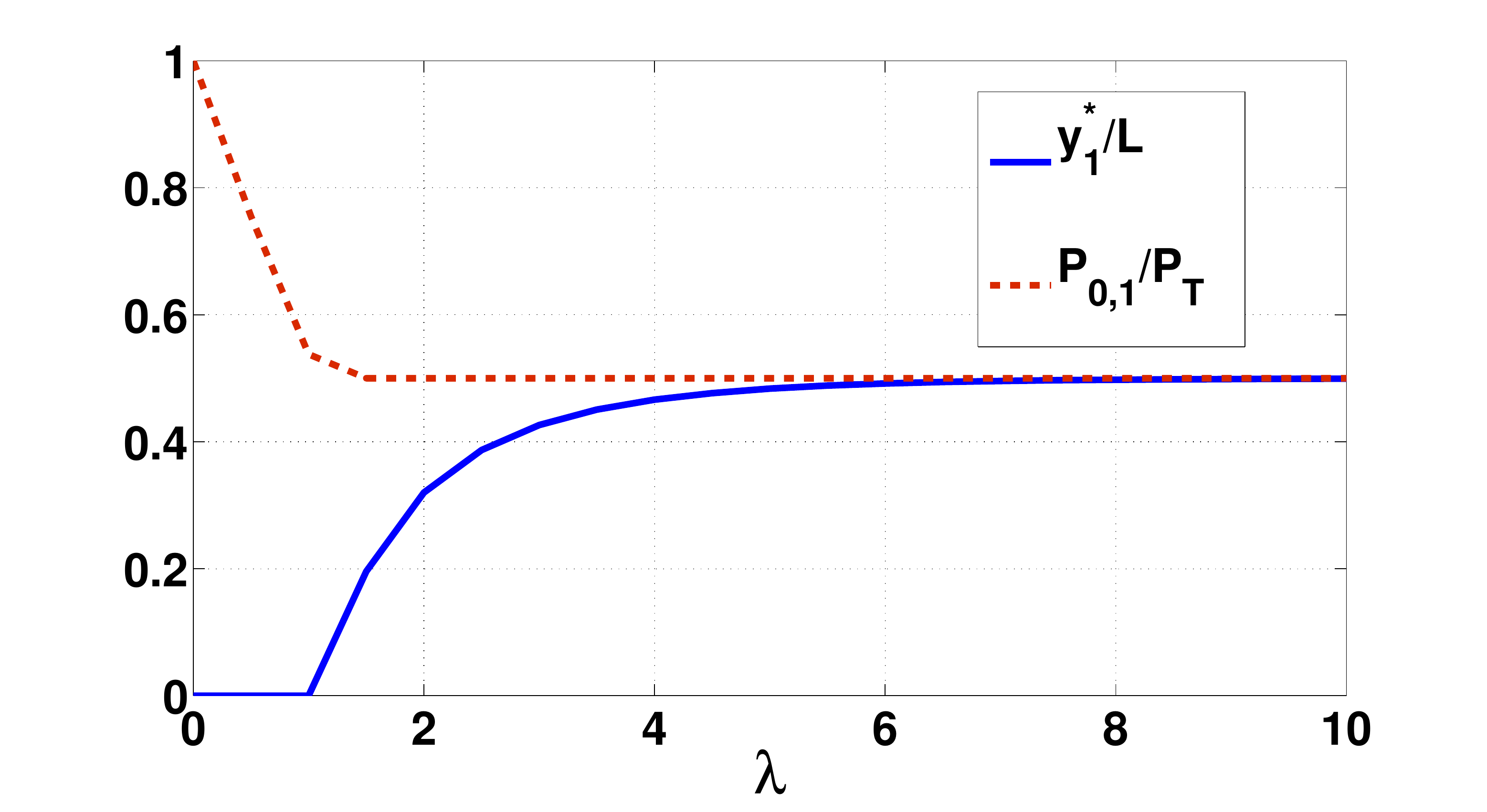}
\caption{Single relay placement, total power constraint, exponential path-loss: $\frac{y_{1}^{*}}{L×}$ and optimum $\frac{P_{0,1}}{P_{T}×}$ versus $\lambda$.}
\label{fig:single_relay_total_power}
\vspace{-3mm}
\end{figure}

\begin{figure*}[t]
\begin{minipage}[r]{0.32\linewidth}
\subfigure[$N=2$]{
\includegraphics[width=\linewidth]{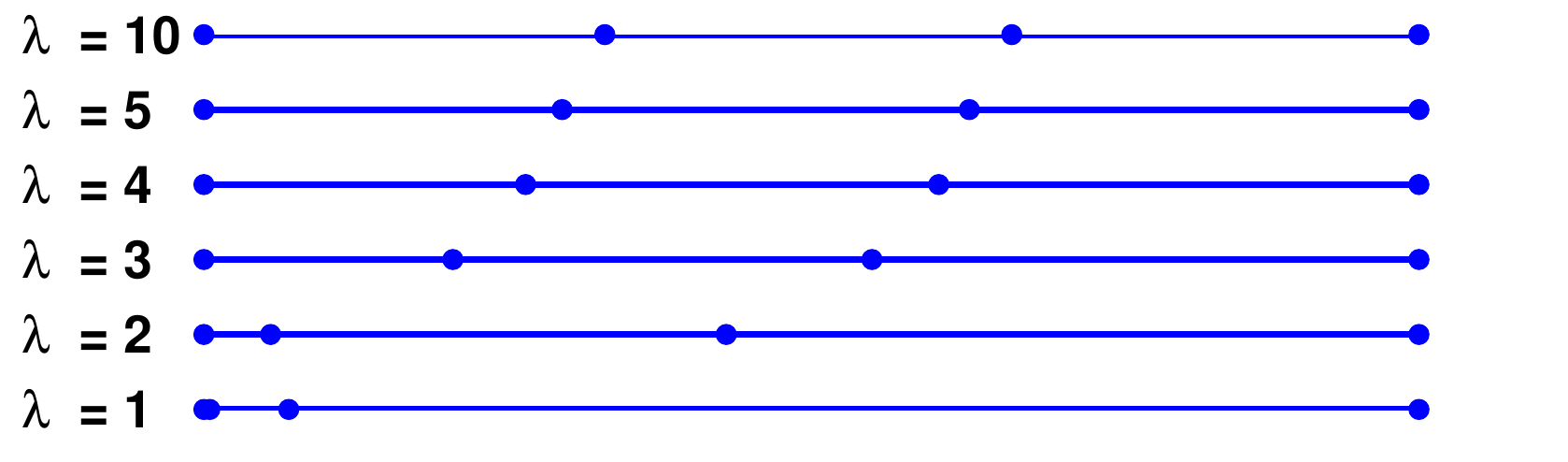}
\label{fig:2relays}}
\end{minipage}  \hfill
\begin{minipage}[c]{0.32\linewidth}
\subfigure[$N=3$]{
\includegraphics[width=\linewidth]{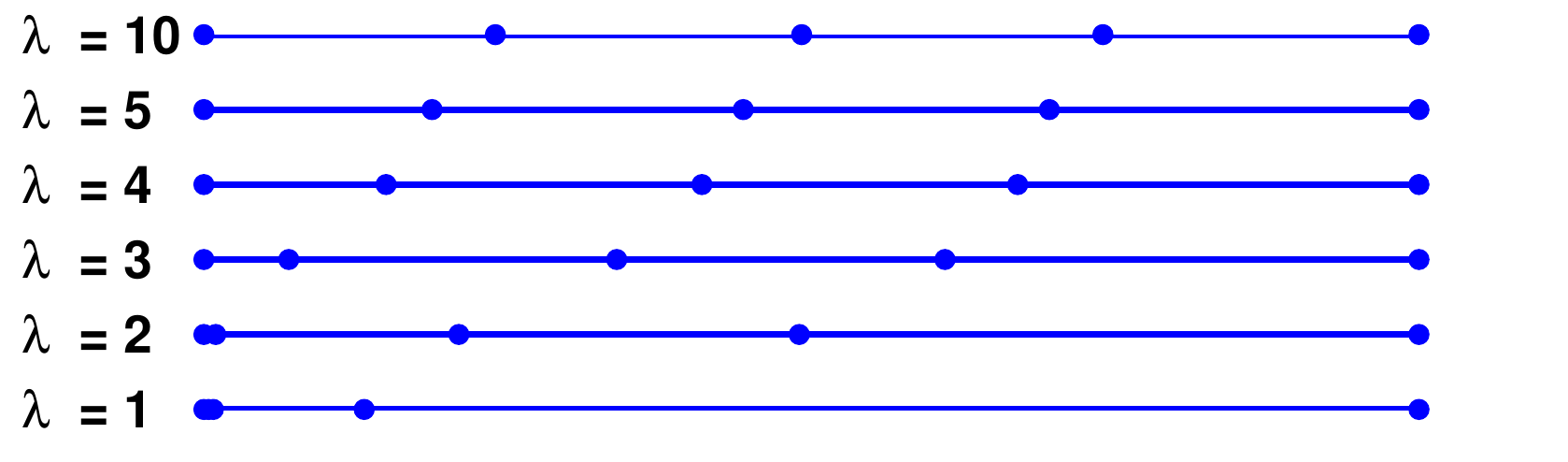}
\label{fig:3relays}}
\end{minipage} \hfill
\begin{minipage}[r]{0.32\linewidth}
\subfigure[$N=5$]{
\includegraphics[width=\linewidth]{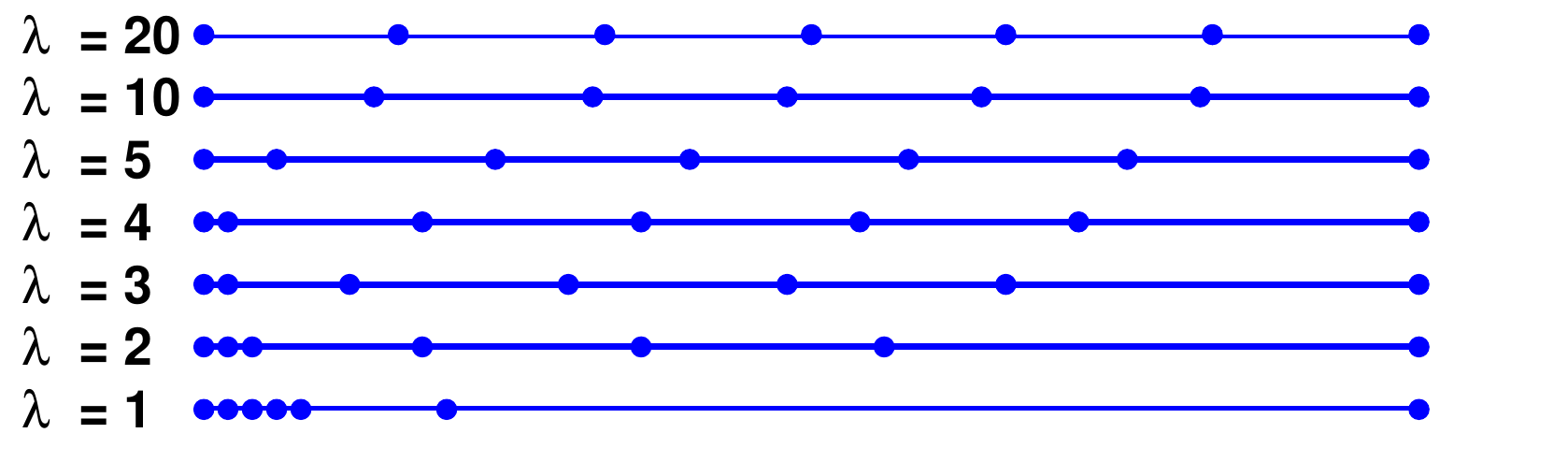}
\label{fig:5relays}}
\end{minipage}
\caption[]{$N$ relays, exponential path-loss model: depiction of the optimal relay positions for $N = 2, 3, 5$ for various values of $\lambda$.}
\label{fig:multirelay-optimal-location}
\end{figure*}


\subsection{Optimal Placement of a Single Relay Node}

\begin{thm}\label{theorem:single_relay_total_power}
 For the single relay node placement problem with sum power constraint and exponential path-loss model, the normalized optimum relay 
location $\frac{y_{1}^{*}}{L×}$, power allocation and optimized achievable rate are given as follows : 
 
\begin{enumerate}[label=(\roman{*})]
 
\item {\em For $\lambda \leq \log 3$}, $\frac{y_{1}^{*}}{L×}=0$, $P_{0,1}=\frac{2P_{T}}{e^{\lambda}+1×}$, 
  $P_{0,2}=P_{1,2}=\frac{e^{\lambda}-1}{e^{\lambda}+1×}\frac{P_{T}}{2×}$ and $R^{*}=C \left(\frac{2P_{T}}{(e^{\lambda}+1)\sigma^{2}×}\right)$.

\item {\em For $\lambda \geq \log 3$}, $\frac{y_{1}^{*}}{L×}=\frac{1}{\lambda×} \log \left(\sqrt{e^{\lambda}+1}-1\right)$, 
$P_{0,1}=\frac{P_{T}}{2×}$, $P_{0,2}=\frac{1}{\sqrt{e^{\lambda}+1}×}\frac{P_{T}}{2×}$, 
$P_{1,2}=\frac{\sqrt{e^{\lambda}+1}-1}{\sqrt{e^{\lambda}+1}×}\frac{P_{T}}{2×}$ and 
$R^{*}=C \left( \frac{1}{\sqrt{e^{\lambda}+1}-1×}\frac{P_{T}}{2 \sigma^{2}×} \right)$
\end{enumerate}
\end{thm}
\begin{proof}
 See Appendix \ref{appendix:total_power_constraint}.
\end{proof}

{\em Remarks and Discussion:}
\begin{itemize}
 \item  It is easy to check that $R^{*}$ obtained in Theorem \ref{theorem:single_relay_total_power} is strictly greater than the 
AWGN capacity $C \left(\frac{P_{T}}{\sigma^{2}×}e^{-\lambda}\right)$ for all $\lambda>0$. This happens because the source and 
relay transmit coherently to the sink. $R^{*}$ becomes equal to the AWGN capacity only at $\lambda=0$. At $\lambda=0$, 
we do not use the relay since the sink can decode any message that the relay is able to decode.

\item The variation of $\frac{y_{1}^{*}}{L×}$ and $\frac{P_{0,1}}{P_{T}×}$ with $\lambda$ has been shown in 
Figure~\ref{fig:single_relay_total_power}.
We observe that (from Figure~\ref{fig:single_relay_total_power} and Theorem~\ref{theorem:single_relay_total_power}) 
$\lim_{\lambda \rightarrow \infty} \frac{y_{1}^{*}}{L}=\frac{1}{2×}$, $\lim_{\lambda \rightarrow \infty}P_{0,2}=0$
 and $\lim_{\lambda \rightarrow 0}P_{0,1}=P_{T}$. For large values of $\lambda$, source and relay cooperation provides negligible 
benefit since source to sink attenuation is very high. So it is optimal to 
place the relay at a distance $\frac{L}{2×}$. The relay 
works as a repeater which forwards data received from the source to the sink. 
\end{itemize}

\vspace{-1mm}
\subsection{Optimal Relay Placement for a Multi-Relay Channel}\label{subsection:properties_of_gain}
\vspace{-1mm}
As we discussed earlier, we need to place $N$ relay nodes such that 
$\frac{1}{g_{0,1}×}+\sum_{k=2}^{N+1} \frac{(g_{0,k-1}-g_{0,k})}{g_{0,k}g_{0,k-1}\sum_{l=0}^{k-1}\frac{1}{g_{0,l}×}×}$ is minimized.
Here $g_{0,k}=e^{-\rho y_{k}}$. We have the constraint $0 \leq y_{1} \leq y_{2} \leq \cdots \leq y_{N} \leq y_{N+1}=L$. Now, writing 
$z_{k}=e^{\rho y_{k}}$, and defining $z_{0}:=1$, we arrive at the following optimization problem:
\begin{eqnarray}
& & \min \bigg\{ z_{1}+\sum_{k=2}^{N+1} \frac{z_{k}-z_{k-1}}{\sum_{l=0}^{k-1} z_{l}}\bigg\}\nonumber\\
& s.t & \,\, 1 \leq z_{1} \leq \cdots \leq z_{N} \leq z_{N+1} = e^{\rho L} \label{eqn:multirelay_optimization}
\end{eqnarray}
One special feature of the objective function is 
that it is convex in each of the variables $z_{1}, z_{2},\cdots, z_{N}$. The objective function is sum of linear fractionals, and
the constraints are linear. 

{\em Remark:} From optimization problem (\ref{eqn:multirelay_optimization}) we observe that optimum 
$z_{1},z_{2},\cdots,z_{N}$ depend only on $\lambda:=\rho L$. Since 
$z_{k}=e^{\lambda \frac{y_{k}}{L×}}$, we find that normalized optimal distance of relays from the source depend only on $\lambda$ and 
$N$.

\begin{thm}\label{theorem:capacity_increasing_with_N}
 For fixed $\rho$, $L$ and $\sigma^{2}$, the optimized achievable rate $R^{*}$ for a sum power constraint 
\textit{strictly} increases with the 
number of relay nodes.
\end{thm}
\begin{proof}
 See Appendix \ref{appendix:total_power_constraint}.
\end{proof}

Let us denote the relaying gain $G$ in (\ref{eqn:gain_definition}) by $G(N,\lambda)$.
\begin{thm}\label{theorem:G_increasing_in_lambda}
 For any fixed number of relays $N \geq 1$, $G(N,\lambda)$ is increasing in $\lambda$.
\end{thm}
\begin{proof}
 See Appendix \ref{appendix:total_power_constraint}.
\end{proof}

\vspace{-2mm}
\subsection{Numerical Work}
\vspace{-1mm}
We discretize the interval $[0,L]$ and run a search program to 
find normalized optimal relay locations for different values of $\lambda$ and $N$. 
 The results 
are summarized in Figure~\ref{fig:multirelay-optimal-location}. We observe that at 
low attenuation (small $\lambda$), relay nodes are clustered near the source node and are often at the source node, whereas 
at high attenuation (large $\lambda$) they are almost uniformly 
placed along the line. For large $\lambda$, the effect of long distance 
between any two adjacent nodes dominates the gain obtained by 
coherent relaying. Hence, it is beneficial to minimize the maximum distance between any two adjacent nodes and thus 
multihopping is a better strategy in this case. On the other hand, if attenuation is 
low, the gain obtained by coherent transmission is dominant.
 In order to allow this, relays should be able to receive sufficient information from their previous 
nodes. Thus, they tend to be clustered 
near the source.\footnote{At low attenuation, one or more than one relay nodes are placed very close to the source. We 
believe that some of them 
will indeed be placed at the source, but it is not showing up in Figure \ref{fig:multirelay-optimal-location}
 because we discretized the interval $[0,L]$ to run a 
search program to find optimum relay locations.}

\begin{figure}[t!]
\centering
\includegraphics[scale=0.3]{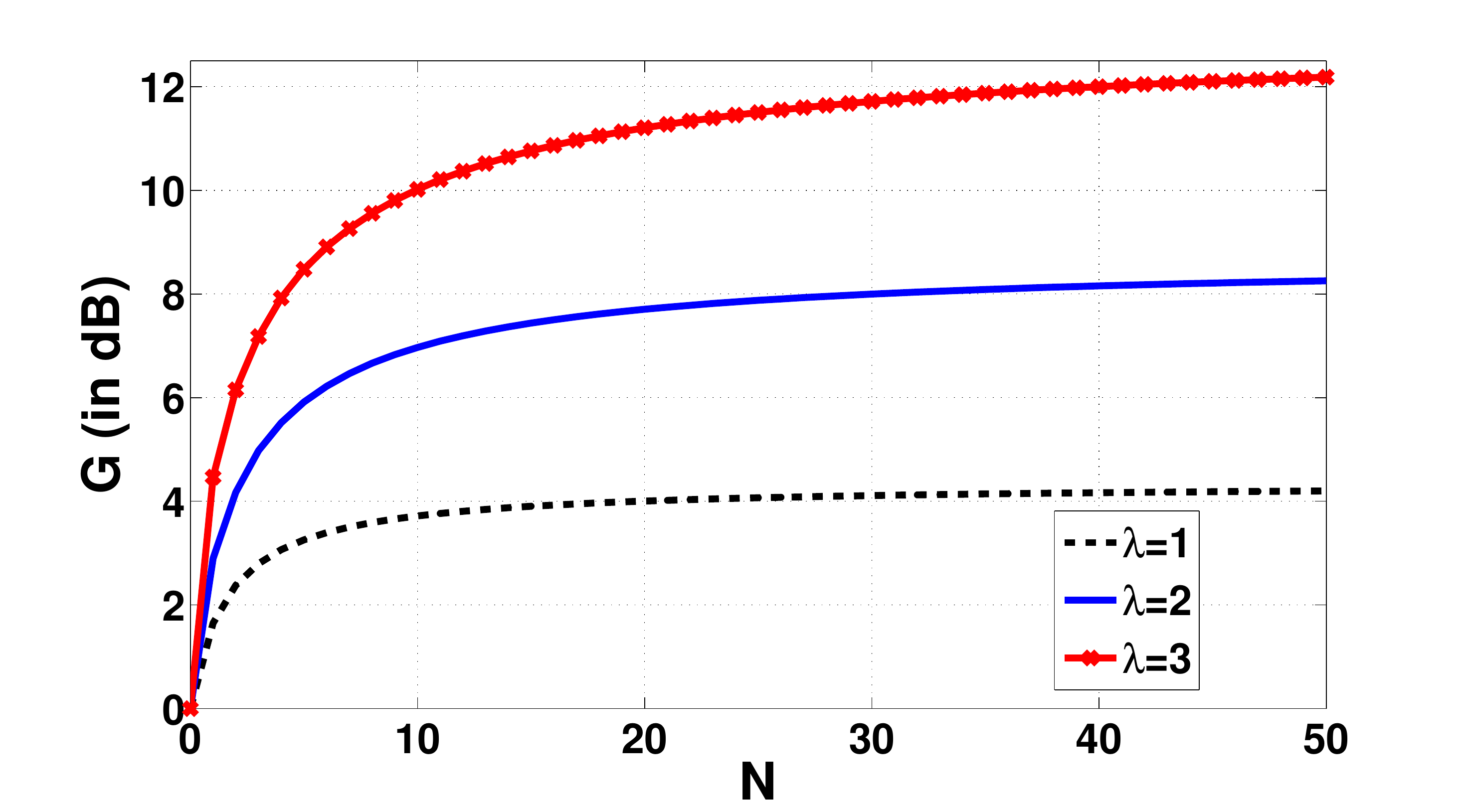}
\caption{$G$ vs $N$ for total power constraint.}
\label{fig:G_vs_N}
\end{figure}

In Figure~\ref{fig:G_vs_N} we plot the relaying gain $G(N,\lambda)$ in dB  vs. 
the number of relays $N$, for various values of $\lambda$. As proved in Theorem~\ref{theorem:capacity_increasing_with_N}, 
we see that the $G(N,\lambda)$ increases with $N$ for fixed $\lambda$. 
On the other hand, $G(N,\lambda)$ increases with $\lambda$ for fixed $N$, as proved in Theorem~\ref{theorem:G_increasing_in_lambda}.

\subsection{Uniformly Placed Relays, Large $N$}
\vspace{-1mm}
If the relays are uniformly placed, the behaviour of $R^{opt}_{P_T}(y_1,y_2,\cdots,y_N)$ (called $R_{N}$ in the next theorem) for large number of relays 
is captured by the following Theorem.

\begin{thm}\label{theorem:large_nodes_uniform}
For the exponential path-loss model with total power constraint, if $N$ relay nodes are placed uniformly between the source and sink, 
resulting in $R_{N}$ achievable rate, then $\liminf_N R_{N} \geq C\left(\frac{P_{T}/\sigma^{2}}{2-e^{-\lambda}×}\right)$ and 
$\limsup_N R_{N} \leq C \left(\frac{P_{T}}{\sigma^{2}×}\right)$.
\end{thm}
\begin{proof}
 See Appendix \ref{appendix:total_power_constraint}.
\end{proof}

 {\em Remark:} From Theorem \ref{theorem:large_nodes_uniform}, it is clear that 
we can ensure at least 
$C(\frac{P_{T}}{2 \sigma ^{2}×})$ data rate (roughly) by placing a large enough number of relay nodes.

\section{Optimal Sequential Placement of Relay Nodes on a Line of Random Length}\label{sec:mdp_total_power}
In this section, we consider the sequential placement of relay nodes along a straight line having unknown random length, subject 
to a total power constraint on the source node and the relay nodes. For analytical tractability, we model the length of the line as an 
exponential random variable with mean $\overline{L}=\frac{1}{\beta×}$.\footnote{One justification 
for the use of the exponential distribution, given the prior
knowledge of the mean length $\overline{L}$, is that it is the maximum entropy continuous probability density
function with the given mean. Thus, by using the exponential distribution, we
are leaving the length of the line as uncertain as we can, given the prior
knowledge of its mean.
} We are motivated by the scenario where, starting from the data source, 
a person 
is walking along the line, and places relays at appropriate places in order to maximize the end-to-end achievable data rate 
(to a sink at the end of the line) subject to a total power constraint and a constraint on the mean number of relays placed. 
We formulate the problem as a total cost Markov decision process (MDP).

\subsection{Formulation as an MDP}\label{subsection:mdp_formulation}
We recall from (\ref{eqn:capacity_multirelay}) that for a fixed length $L$ of the line and a fixed $N$,
$e^{\rho y_1}+\sum_{k=2}^{N+1}\frac{e^{\rho y_k}-e^{\rho y_{k-1}}}{1+e^{\rho y_1}+\cdots+e^{\rho y_{k-1}}×}$ has to be minimized in 
order to maximize $R^{opt}_{P_T}(y_1,y_2,\cdots,y_N)$. 
$e^{\rho y_1}+\sum_{k=2}^{N+1}\frac{e^{\rho y_k}-e^{\rho y_{k-1}}}{1+e^{\rho y_1}+\cdots+e^{\rho y_{k-1}}×}$ is basically a scaling 
factor which captures the effect of attenuation and relaying on the maximum possible SNR $\frac{P_T}{\sigma^2}$. 

In the ``as-you-go" placement problem, the person carries a number of nodes and places them as he walks, 
under the control of a placement policy. A deployment policy $\pi$ is a 
sequence of mappings $\mu_k, k \geq 1$, where at the $k$-th decision instant 
 $\mu_k$ provides the distance at which the next relay should be placed (provided that the 
line does not end before that point).  The decisions are made based on the locations of the relays placed earlier. 
The first decision instant is the start of the line, and the subsequent decision instants are the placement points of the relays. 
Let $\Pi$ denote the set of all deployment policies. Let $\mathbb{E}_{\pi}$ denote the 
expectation under policy $\pi$. Let $\xi>0$ be the cost of placing a relay. 
We are interested in solving the following problem:
\begin{eqnarray}
 \inf_{\pi\in \Pi } \mathbb{E}_{\pi} \left(e^{\rho y_1}+
\sum_{k=2}^{N+1}\frac{e^{\rho y_k}-e^{\rho y_{k-1}}}{1+e^{\rho y_1}+
\cdots+e^{\rho y_{k-1}}×}+\xi N \right)\label{eqn:unconstrained_mdp}
\end{eqnarray}
Note that, due to the randomness of the length of the line, the $y_k, k \geq 1,$ are also random variables.

We now formulate the above ``as-you-go" relay placement problem as a total cost Markov decision process. 
Let us define $s_0:=1$, $s_k:=\frac{e^{\rho y_k}}{1+e^{\rho y_1}+\cdots+e^{\rho y_k}×}\, \forall \, k \geq 1$. Also, recall that 
$r_{k+1}=y_{k+1}-y_k$. Thus, we can rewrite (\ref{eqn:unconstrained_mdp}) as follows:
\begin{eqnarray}
 \inf_{\pi\in \Pi } \mathbb{E}_{\pi} \left(1+\sum_{k=0}^{N}s_k(e^{\rho r_{k+1}}-1)+\xi N \right) \label{eqn:unconstrained_mdp_with_state}
\end{eqnarray}

When the person starts walking from the source along the line, the state 
of the system is set to $s_0:=1$. At this instant the 
placement policy provides the location at which the first relay should be placed. The person walks towards the 
prescribed placement point. If the line ends before reaching this point, the sink 
is placed; if not, then the first relay is placed at the placement point. In general, the state after placing the 
$k$-th relay is denoted by $s_k$, for $k=1,2,\cdots$. At state $s_k$, the action is the distance $r_{k+1}$ where the next relay has 
to be placed. If the line ends before this distance, the sink node has to be placed at the end. 
{\em The randomness is coming from the random residual length of the line.} Let $l_k$ 
denote the residual length of the line at the $k$-th instant.

With this notation, the state of the system evolves as follows:

\begin{eqnarray}
 s_{k+1}=
\begin{cases}
\frac{s_k e^{\rho r_{k+1}}}{1+s_k e^{\rho r_{k+1}}×}, \text{ if $l_k > r_{k+1}$} , \\
 \mathbf{e}, \text{  else}. \label{eqn:state-evolution}
\end{cases}
\end{eqnarray}

Here $\mathbf{e}$ denotes the end of the line, i.e., the termination state. The single stage cost incurred when the state is $s$, 
the action is $a$ and the residual length of the line is $l$, is given by:

\begin{eqnarray}
 c(s,a,l)=
\begin{cases}
\xi + s(e^{\rho a}-1), \text{ if $l>a$},\\
s(e^{\rho l}-1), \text{ if $l \leq a$}. \label{eqn:single-stage-cost}
\end{cases}
\end{eqnarray}
Also, $c(\mathbf{e},a,\cdot)=0$ for all $a$. 

Now, from (\ref{eqn:state-evolution}), it is clear that the next state $s_{k+1}$ 
depends on the current state $s_k$, the current action $r_{k+1}$ 
and the residual length of the line. Since the line is exponentially distributed the residual length need 
not be retained in the state description; from any placement point, the residual line length is exponentially 
distributed, and independent of the history of the process. Also, the cost incurred at the $k$-th decision instant 
is given by (\ref{eqn:single-stage-cost}), which depends on $s_k$, $r_{k+1}$ 
and $l_k$. Hence, our problem (\ref{eqn:unconstrained_mdp}) is an MDP problem with state space $\mathcal{S}:=(0,1]\cup \{\mathbf{e}\}$ 
and action space $\mathcal{A}\cup \{\infty \}$ where $\mathcal{A}:=[0,\infty)$. Action $\infty$ means that no further relay 
will be placed. 

Solving the problem in (\ref{eqn:unconstrained_mdp}) also helps in solving the following constrained problem:
\begin{eqnarray}
& & \inf_{\pi\in \Pi } \mathbb{E}_{\pi} \left(e^{\rho y_1}+
\sum_{k=2}^{N+1}\frac{e^{\rho y_k}-e^{\rho y_{k-1}}}{1+e^{\rho y_1}+\cdots+e^{\rho y_{k-1}}×}\right)\nonumber\\
&\textit{s.t., }&\mathbb{E}_{\pi}(N) \leq M \label{eqn:constrained_mdp}
\end{eqnarray}
where $M>0$ is the constraint on the average number of relays. The following standard result tells us how to choose the optimal 
$\xi$:

\begin{lem}
 Let $\pi_{\xi}^* \in \Pi$ be an optimal policy for the unconstrained problem (\ref{eqn:unconstrained_mdp}) such that 
$\mathbb{E}_{\pi_{\xi}^*}N=M$. Then $\pi_{\xi}^*$ is also optimal for the constrained problem (\ref{eqn:constrained_mdp}).
\end{lem}
The constraint on the mean number of relays can
be justified if we consider the relay deployment problem for
multiple source-sink pairs over several different lines of mean
length $\overline{L}$, given a large pool of relays, and we are only
interested in keeping small the total number of relays over
all these deployments.

{\em Remark:} The optimal policy for the problem (\ref{eqn:unconstrained_mdp}) will be used to place relay nodes along a line 
whose length is a sample from an exponential distribution with mean $\frac{1}{\beta×}$. After the deployment is over, the 
power $P_T$ will be shared optimally among the source and the deployed relay nodes (according to 
the formulas in Theorem~\ref{theorem:multirelay_capacity}), in order to implement the coding scheme (as described in 
Section~\ref{subsec:achievable_rate_multirelay_xie_and_kumar}) in the network.

\subsection{Analysis of the MDP}\label{subsection:analysis_of_mdp}
Suppose $s_k=s$ for some $k \geq 0$. Then, the optimal value function (cost-to-go) at state $s$ is defined by:
\begin{equation*}
 J_{\xi}(s)=\inf_{\pi \in \Pi} \mathbb{E} \left(\sum_{n=k}^{\infty} c(s_n, a_n, l_n)|s_k=s \right)
\end{equation*}
 
If we 
decide to place the next relay at a distance $a<\infty$ and follow the optimal policy thereafter, the expected cost-to-go 
at a state $s \in (0,1]$ becomes:
\begin{eqnarray}
 & & \int_{0}^{a}\beta e^{-\beta z}s(e^{\rho z}-1)dz \nonumber\\
& +&e^{-\beta a}\bigg(s(e^{\rho a}-1)+\xi+J_{\xi}\left(\frac{se^{\rho a}}{1+se^{\rho a}×}\right)\bigg)\label{eqn:cost-to-go}
\end{eqnarray}
The first term in (\ref{eqn:cost-to-go}) corresponds to the case in which the line ends 
at a distance less than $a$ and we are forced to place the sink node. 
The second term corresponds to the case where the residual length of the line is greater than $a$ and a relay is placed 
at a distance $a$. 

Note that our MDP has an uncountable state space $\mathcal{S}=(0,1] \cup \{\mathbf{e}\}$ and a non-compact action space 
$\mathcal{A}=[0,\infty) \cup \{\infty \}$. Several technical issues arise in this kind of problems, such as 
the existence of optimal or $\epsilon$-optimal policies, measurability of the policies, etc. 
We, therefore, invoke the results 
provided by Sch\"{a}l \cite{schal75conditions-optimality}, which deal with such issues. 
Our problem is one 
of minimizing total, undiscounted, non-negative costs over an infinite 
horizon. Equivalently, in the context of \cite{schal75conditions-optimality}, we have a problem of  
total reward maximization where the rewards are the negative of the costs. Thus, our problem specifically fits into 
the negative dynamic programming setting 
of \cite{schal75conditions-optimality} (i.e., the $\mathsf{N}$ case where single-stage rewards are non-positive).

Now, we observe that the state $\mathbf{e}$ is absorbing. Also no action is taken at 
this state and the cost at this state is $0$. Hence, we can think of this state as state $0$ in order to 
make our state space a Borel subset of the real line. 
We present the following theorem from \cite{schal75conditions-optimality} in our present context. 

\begin{thm}\label{thm:schal_bellman_eqn}
[\cite{schal75conditions-optimality}, Equation (3.6)]      
The optimal value function $J_{\xi}(\cdot)$ satisfies the Bellman equation.
\end{thm}

Thus, $J_\xi(\cdot)$ satisfies the following Bellman equation for each $s \in (0,1]$:
\begin{eqnarray}
 J_{\xi}(s) &=& \min \bigg \{\inf_{a \geq 0} \bigg[\int_{0}^{a}\beta e^{-\beta z}s(e^{\rho z}-1)dz \nonumber\\
& +&e^{-\beta a}\bigg(s(e^{\rho a}-1)+\xi+J_{\xi}\left(\frac{se^{\rho a}}{1+se^{\rho a}×}\right)\bigg)\bigg],\nonumber\\
& & \int_{0}^{\infty}\beta e^{-\beta z}s(e^{\rho z}-1)dz \bigg \} \label{eqn:bellman_unbroken} 
\end{eqnarray}
where the second term inside $\min\{\cdot, \cdot \}$ is the cost of not placing any relay (i.e., $a=\infty$).

Recall that the line has a length having exponential distribution with mean $\frac{1}{\beta×}$ and that the path-loss exponent 
is $\rho$ for the exponential path-loss model. We analyze the MDP problem 
for the two different cases: $\beta>\rho$ and $\beta \leq \rho$.

\subsubsection{Case I ($\beta>\rho$)}

We observe that the cost of not placing any relay (i.e., $a=\infty$) at state $s \in (0,1]$ is given by:
\begin{eqnarray*}
 \int_{0}^{\infty} \beta e^{-\beta z}s(e^{\rho z}-1)dz=\theta s
\end{eqnarray*}
where $\theta:=\frac{\rho}{\beta-\rho×}$ (using the fact that $\beta>\rho$). 
Since not placing a relay (i.e., $a = \infty$) is a possible action for every $s$, it follows that $J_{\xi}(s)\leq \theta s $.

The cost in (\ref{eqn:cost-to-go}), upon simplification, can be written as:
\begin{eqnarray}
 \theta s + e^{-\beta a}\bigg(-\theta s e^{\rho a}+\xi+
J_{\xi}\left(\frac{se^{\rho a}}{1+se^{\rho a}×}\right)\bigg)\label{eqn:cost_to_go_beta_greater_than_rho}
\end{eqnarray}
Since $J_{\xi}(s) \leq \theta$ for all $s \in (0,1]$, the expression in (\ref{eqn:cost_to_go_beta_greater_than_rho}) 
is strictly less that 
$\theta s$ for large enough $a<\infty$. Hence, according to (\ref{eqn:bellman_unbroken}), it is not optimal to not place any relay.

Thus, in this case, the Bellman equation (\ref{eqn:bellman_unbroken}) can be rewritten as:
\begin{eqnarray}
 J_{\xi}(s)= \theta s + \inf_{a \geq 0} e^{-\beta a}\bigg(-\theta s e^{\rho a}+
\xi+J_{\xi}\left(\frac{se^{\rho a}}{1+se^{\rho a}×}\right)\bigg)\label{eqn:bellman_equation_simplified_in_a}
\end{eqnarray}

Observe that $\theta s + e^{-\beta a}\bigg(-\theta s e^{\rho a}+
\xi+J_{\xi}\left(\frac{se^{\rho a}}{1+se^{\rho a}×}\right)\bigg) \rightarrow \theta s$ 
as $a \rightarrow \infty$.

\subsubsection{Case II ($\beta \leq \rho$)}
Here the cost in (\ref{eqn:cost-to-go}) is $\infty$ if we do not place a relay (i.e., if $a=\infty$). 
Let us consider a policy $\pi_1$ where we place the next relay at a fixed distance $0 <a <\infty$ from the current relay, 
irrespective of the current state. If the residual length of the line is $z$ at any state $s$, we will place less than $\frac{z}{a×}$ additional 
 relays, and 
for each relay a cost less than $(\xi+(e^{\rho a}-1))$ is incurred (since $s \leq 1$). At the last step when we place the 
sink, a cost less than $(e^{\rho a}-1)$ is incurred.  Thus, the value function of this policy is upper bounded by:
\begin{eqnarray}
& & \int_{0}^{\infty} \beta e^{-\beta z} \frac{z}{a×}(\xi+(e^{\rho a}-1)) dz+(e^{\rho a}-1)\nonumber\\
& =& \frac{1}{\beta a ×}\left(\xi+(e^{\rho a}-1)  \right)+(e^{\rho a}-1)\label{eqn:upper_bound_on_cost_beta_leq_rho}
\end{eqnarray}
Hence, $J_{\xi}(s) \leq \frac{1}{\beta a ×}\left(\xi+(e^{\rho a}-1)  \right)+(e^{\rho a}-1) < \infty$. 
Thus, by the same argument as in 
the case $\beta > \rho$, the minimizer in the Bellman equation lies in $[0,\infty)$, i.e., 
the optimal placement distance lies in $[0,\infty)$. Hence, the Bellman Equation 
(\ref{eqn:bellman_unbroken}) can be rewritten as:

\footnotesize
\begin{eqnarray}
 J_{\xi}(s) &=& \inf_{a \geq 0} \bigg\{\int_{0}^{a}\beta e^{-\beta z}s(e^{\rho z}-1)dz \nonumber\\
& +& e^{-\beta a}\bigg(s(e^{\rho a}-1)+\xi+J_{\xi}\left(\frac{se^{\rho a}}{1+se^{\rho a}×}\right)\bigg)\bigg\}
\label{eqn:bellman_beta_leq_rho} 
\end{eqnarray} 
\normalsize

\subsection{Upper Bound on the Optimal Value Function}

\begin{prop}\label{prop:upper_bound_on_cost_beta_geq_rho}
If $\beta>\rho$, then $J_{\xi}(s) < \theta s$ for all $s \in (0,1]$.
\end{prop}
\begin{proof}
 We know that $J_{\xi}(s) \leq \theta s \leq \theta$. Now, let us consider 
the Bellman equation (\ref{eqn:bellman_equation_simplified_in_a}). 
It is easy to see that $\bigg(-\theta s e^{\rho a}+\xi+J_{\xi}\left(\frac{se^{\rho a}}{1+se^{\rho a}×}\right)\bigg)$ 
is strictly negative for sufficiently large $a$. 
Hence, for a given $s$, the R.H.S of (\ref{eqn:bellman_equation_simplified_in_a}) is strictly less than $\theta s$. 
\end{proof}

\begin{cor}
 For $\beta>\rho$, $J_{\xi}(s) \rightarrow 0$ as $s \rightarrow 0$ for any $\xi>0$.
\end{cor}
\begin{proof}
 Follows from Proposition \ref{prop:upper_bound_on_cost_beta_geq_rho}.
\end{proof}

\begin{prop}\label{prop:upper_bound_on_cost}
If $\beta>0$ and $\rho>0$ and $0<a<\infty$, then $J_{\xi}(s) <\frac{1}{\beta a ×}\left(\xi+(e^{\rho a}-1)  \right)+(e^{\rho a}-1)$ for all $s \in (0,1]$.
\end{prop}
\begin{proof}
 Follows from (\ref{eqn:upper_bound_on_cost_beta_leq_rho}), since the analysis is valid even for $\beta>\rho$.
\end{proof}

\vspace{-0.5cm}
\subsection{Convergence of the Value Iteration}

The value iteration for our sequential relay placement problem is given by:

\footnotesize
\begin{eqnarray}
 J_{\xi}^{(k+1)}(s) &=& \inf_{a \geq 0} \bigg\{\int_{0}^{a}\beta e^{-\beta z}s(e^{\rho z}-1)dz \nonumber\\
& +&e^{-\beta a}\bigg(s(e^{\rho a}-1)+\xi+J_{\xi}^{(k)}\left(\frac{se^{\rho a}}{1+se^{\rho a}×}\right)\bigg)\bigg\}\label{eqn:value_iteration}
\end{eqnarray} 
\normalsize

Here $J_{\xi}^{(k)}(s)$ is the $k$-th iterate of the value iteration. 
Let us start with $J_{\xi}^{(0)}(s):=0$ for all $s \in (0,1]$. We set $J_{\xi}^{(k)}(\mathbf{e})=0$ for all $k \geq 0$. 
$J_{\xi}^{(k)}(s)$ is the optimal value function 
for a problem with the same single-stage cost and the same transition structure, but with the horizon 
length being $k$ (instead of infinite horizon as in our original problem) and $0$ terminal cost. Here, by horizon length $k$, 
we mean that there are $k$ number of relays available for deployment.

Let $\Gamma_k(s)$ be the set of minimizers of (\ref{eqn:value_iteration}) at the 
$k$-th iteration at state $s$, if the infimum is achieved at some $a<\infty$. 
Let $\Gamma_{\infty}(s):=\{a \in \mathcal{A}:a$ be an 
accumulation point of some sequence $\{a_k\}$ where each $a_k \in \Gamma_{k}(s)\}$. Let $\Gamma^*(s)$ be the set of minimizers in 
(\ref{eqn:bellman_beta_leq_rho}). In Appendix \ref{appendix:sequential_placement_total_power}, we show 
that $\Gamma_k(s)$ for each $k \geq 1$, $\Gamma_{\infty}(s)$ and $\Gamma^*(s)$ are nonempty.

\begin{thm}\label{theorem:convergence_of_value_iteration}
 The value iteration given by (\ref{eqn:value_iteration}) has the following properties: 
\begin{enumerate}[label=(\roman{*})]
 \item $J_{\xi}^{(k)}(s)\rightarrow J_{\xi}(s)$ for all $s \in (0,1]$, i.e., the value iteration converges 
       to the optimal value function.
\item $\Gamma_{\infty}(s) \subset \Gamma^*(s)$.
\item There is a stationary optimal policy $f^{\infty}$ where $f:(0,1] \rightarrow \mathcal{A}$ and $f(s) \in \Gamma_{\infty}(s)$ 
for all $s \in (0,1]$.
\end{enumerate}
\end{thm}

\begin{proof}
The proof is given in Appendix \ref{appendix:sequential_placement_total_power}, 
Section \ref{appendix_subsection-convergence-value-iteration-proof}. It uses some advanced mathematical tools, which have been 
discussed in Appendix \ref{appendix:sequential_placement_total_power}, Section \ref{appendix_subsection_schal-discussion}.
\end{proof}

\vspace{-0 cm}
\subsection{Properties of the Value Function $J_{\xi}(s)$}

\begin{prop}\label{prop:increasing_concave_in_s}
 $J_{\xi}(s)$ is increasing and concave over $s \in (0,1]$.
\end{prop}

\begin{prop}\label{prop:increasing_concave_in_lambda}
 $J_{\xi}(s)$ is increasing and concave in $\xi$ for all $s \in (0,1]$.
\end{prop}

\begin{prop}\label{prop:continuity_of_cost}
 $J_{\xi}(s)$ is continuous in $s$ over $(0,1]$ and continuous in $\xi$ over $(0,\infty)$.
\end{prop}

Proofs of the Propositions \ref{prop:increasing_concave_in_s}-\ref{prop:continuity_of_cost} have been given in Appendix 
\ref{appendix:sequential_placement_total_power}, Section \ref{appendix_subsection_policy-structure}. The proofs show that the 
increasing, concave nature of the value function is preserved at every stage of the 
value iteration (\ref{eqn:value_iteration}), and use the fact that 
$J_{\xi}^{(k)}(s)\rightarrow J_{\xi}(s)$.

\subsection{A Useful Normalization}

Note that in our sequential placement problem we have assumed $L$ to be exponentially distributed with mean $\frac{1}{\beta×}$. Hence, $\beta L$ 
is exponentially distributed with mean $1$. Now, defining $\Lambda:=\frac{\rho}{\beta×}$ and $\tilde{z}_k:=\beta y_k$, $k=1,2,\cdots,(N+1)$, 
we can rewrite (\ref{eqn:unconstrained_mdp}) as follows:

\begin{eqnarray}
 \inf_{\pi\in \Pi } \mathbb{E}_{\pi} \left(e^{\Lambda \tilde{z}_1}+
\sum_{k=2}^{N+1}\frac{e^{\Lambda \tilde{z}_k}-e^{\Lambda  \tilde{z}_{k-1}}}{1+e^{\Lambda \tilde{z}_1}+
\cdots+e^{\Lambda \tilde{z}_{k-1}}×}+\xi N \right)\label{eqn:unconstrained_mdp_beta_one}
\end{eqnarray}

Note that, $\Lambda$ plays the same role as $\lambda$ played in the known $L$ case 
(see Section \ref{sec:total_power_constraint}). Since $\frac{1}{\beta×}$ is the mean length of the line, $\Lambda$ can 
be considered as a measure of attenuation in the network. 
We can think of the new problem (\ref{eqn:unconstrained_mdp_beta_one}) in the same way as (\ref{eqn:unconstrained_mdp}), but with 
the length of the line being exponentially distributed with mean $1$ ($\beta'=1$) and the path-loss exponent being changed to 
$\rho'=\Lambda=\frac{\rho}{\beta×}$. The relay locations are also normalized ($\tilde{z}_k=\beta y_k$). One can solve the new problem 
(\ref{eqn:unconstrained_mdp_beta_one}) and obtain the optimal policy. Then the solution to (\ref{eqn:unconstrained_mdp}) can be 
obtained by multiplying each control distance (from the optimal policy of (\ref{eqn:unconstrained_mdp_beta_one})) 
with the constant $\frac{1}{\beta×}$. Hence, it suffices to work with $\beta=1$.

\vspace{2mm}
\subsection{A Numerical Study of the Optimal Policy}

\begin{figure}[t!]
\centering
\includegraphics[scale=0.25]{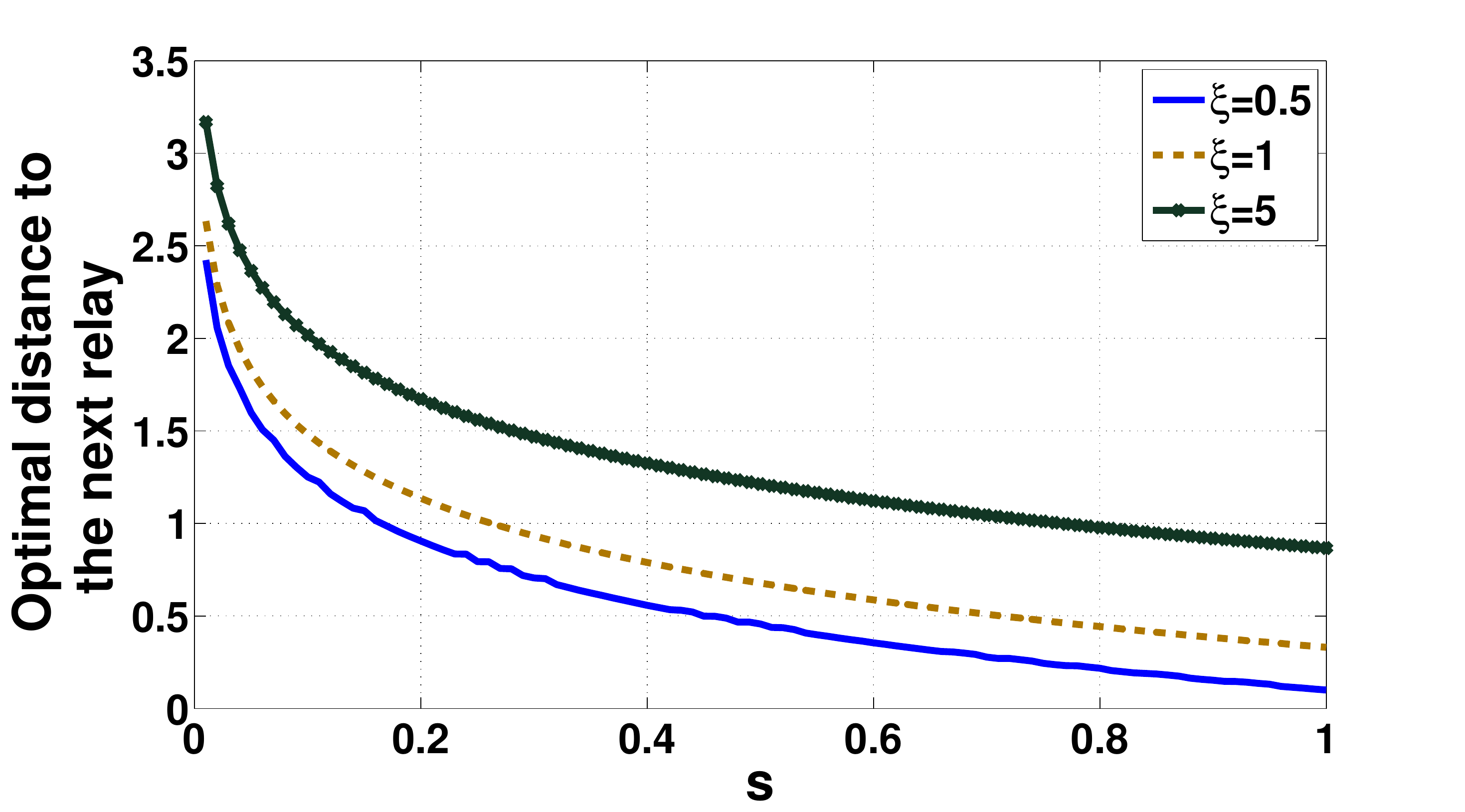}
\vspace{-0.3cm}
\caption{$\Lambda:=\frac{\rho}{\beta×}=2$; normalized optimal placement distance $a^*$ vs. $s$.}
\label{fig:control_vs_state}
\end{figure}

\begin{figure}[t!]
\centering
\includegraphics[scale=0.25]{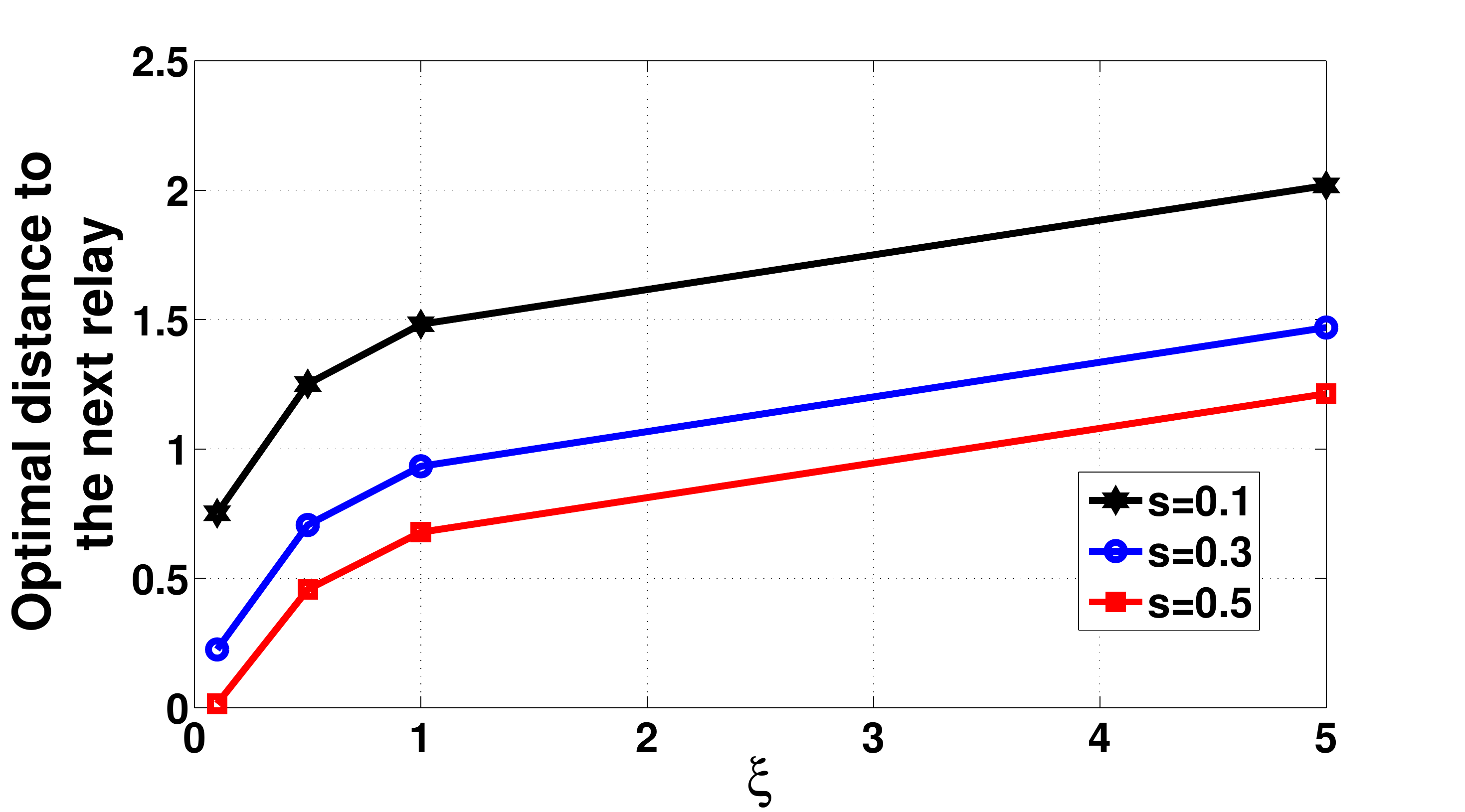}
\vspace{-0.3cm}
\caption{$\Lambda:=\frac{\rho}{\beta×}=2$; normalized optimal placement distance $a^*$ vs. $\xi$.}
\label{fig:control_vs_xi}
\end{figure}

Let us recall that the state of the system after placing the $k$-th relay is given by 
$s_k=\frac{e^{\Lambda \tilde{z}_k}}{\sum_{i=0}^{k}e^{\Lambda \tilde{z}_i}}$. The action is the normalized 
distance of the next relay to be placed from the 
current location. The single stage cost function for our total cost minimization problem is given by (\ref{eqn:single-stage-cost}). 

In our numerical work, we discretized the state space $(0,1]$ into $100$ steps, so that the state space becomes 
$\{0.01,0.02,\cdots,0.99,1\}$. We discretized the action space into steps of size $0.001$, i.e., the action space becomes 
$\{0,0.001,0.002,\cdots\}$.

We performed numerical experiments to study the structure of the optimal policy obtained through value 
iteration for $\beta'=1$ and some values of 
 $\Lambda:=\frac{\rho}{\beta×}$. The value iteration in this experiment converged and we obtained a unique stationary optimal policy, 
though Theorem~\ref{theorem:convergence_of_value_iteration} does not guarantee the uniqueness of the stationary optimal policy. 
Figure~\ref{fig:control_vs_state}   
shows that the normalized optimal placement  distance  $a^*$ 
is decreasing with the state $s \in (0,1]$. This can be understood as follows. The state $s$ (at a placement point) is small 
only if a sufficiently large number of 
relays have already been placed.\footnote{$\frac{e^{\Lambda \tilde{z}_k}}{\sum_{i=0}^{k}e^{\Lambda \tilde{z}_i}} 
\geq \frac{e^{\Lambda \tilde{z}_k}}{(k+1)e^{\Lambda \tilde{z}_k}}=\frac{1}{k+1}$; 
hence, if $s_k$ is 
small, $k$ must be large enough.} Hence, if several relays have already been placed and $\sum_{i=0}^{k}e^{\Lambda \tilde{z}_i}$ is 
sufficiently large compared to $e^{\Lambda \tilde{z}_k}$ (i.e., $s_k$ is small), the $(k+1)$-st relay will 
be able to receive sufficient amount of power 
from the previous nodes, and hence does not need to be placed close to the $k$-th relay. 
A value of $s_k$ close to $1$ indicates that there is a 
large gap between relay $k$ and relay $k-1$, the power 
received at the next relay from the previous relays is small and hence the next relay must be placed closer to the last one.

On the other hand, $a^*$ is increasing with $\xi$ (see also Figure~\ref{fig:control_vs_xi}, previous page). 
Recall that $\xi$ can be viewed as the price 
of placing a relay. This figure confirms the intuition that  if the relay price is high, then the relays should be placed less frequently.

\begin{figure}[t!]
\centering
\includegraphics[scale=0.25]{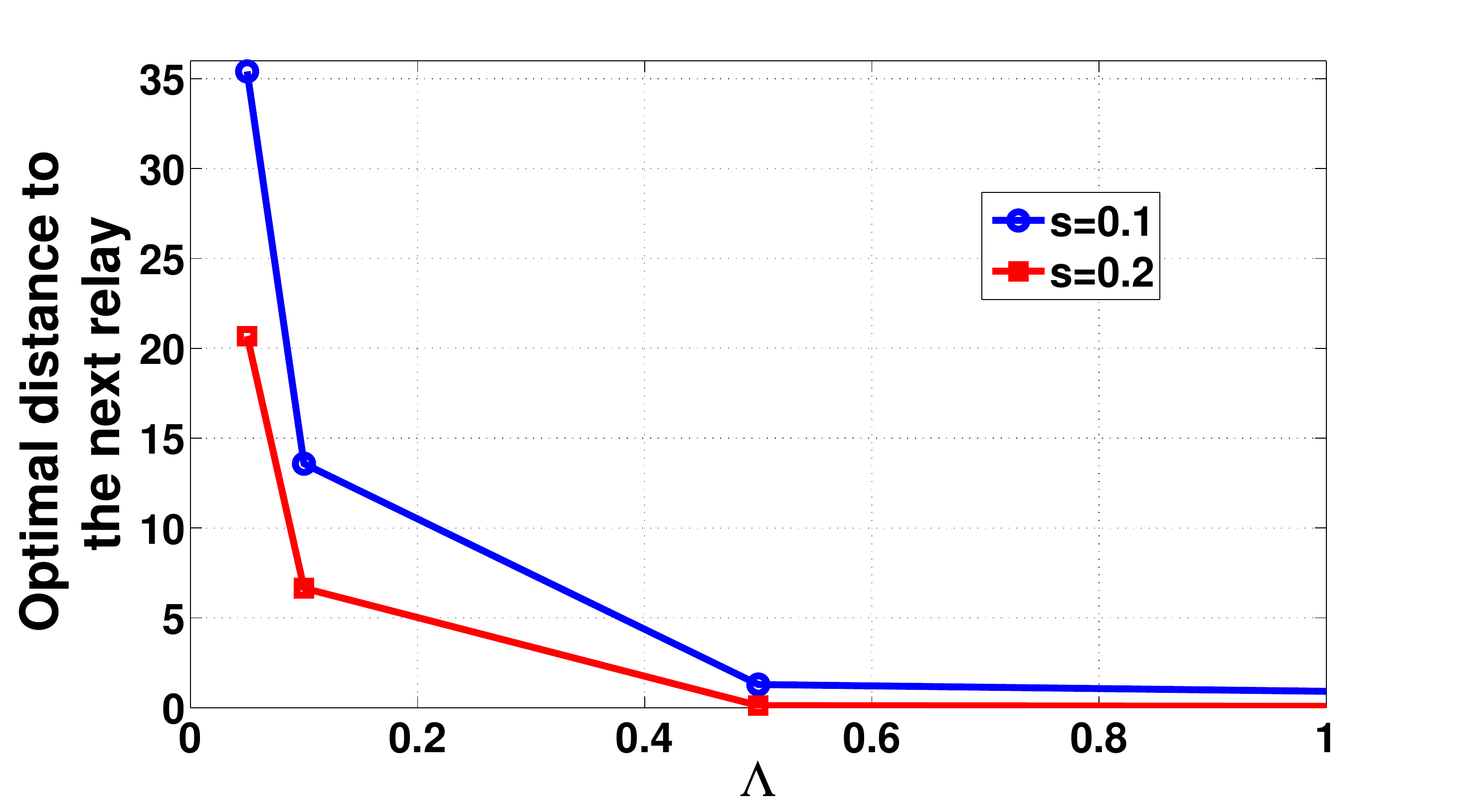}
\vspace{-0.3cm}
\caption{$\xi=0.01$; normalized optimal placement distance $a^*$ vs. $\Lambda$.}
\label{fig:control_vs_rho}
\end{figure}

\begin{table}[t!]
\footnotesize
\centering
\begin{tabular}{|c |c |}
\hline
$\mathbf{\Lambda}$  & Optimal normalized distance \\
 & of the nodes from the source \\ \hline
$0.01$ & 0,        0,   8.4180,   10.0000  \\ \hline 
$0.1$ &  0,         0,         0,         0,         0,         0,         0,    0.2950,    0.5950,    0.9810,    1.3670,         \\  
       & 1.7530,    2.1390,  $\cdots,$   9.0870,    9.4730,     9.8590,   10.0000   \\ \hline
$0.5$ &  0,  0,  0,  0,  0,  0,  0,  0,  0,  0, 0.0200,  0.1220,  0.2240, 0.3260,  0.4280,      \\ 
      &0.5300, $\cdots$, 9.8120,    9.9140,   10.0000   \\ \hline
$2$ & 0, 0,  0,  0,  0,  0,  0,  0,  0,  0, 0.0050, 0.0250,  0.0450, 0.0650, \\
 &  $\cdots$,  9.9650,  9.9850, 10.0000  \\ \hline 
$5$ &  0,  0,  0,  0,  0,  0,  0,  0,  0,  0,  0.0020, 0.0080,  0.0140,   \\ 
 &  0.0200, $\cdots$, 9.9860, 9.9920,  9.9980, 10.0000\\ \hline 
\end{tabular}
\caption{Optimal placement of nodes on a line of length $10$ for various values of $\Lambda$,
  using the corresponding optimal policies for $\xi=0.001$.}
\label{table:effect_of_rho_on_placement_1}
\normalsize
\end{table}

\begin{table}[t!]
\footnotesize
\centering
\begin{tabular}{|c |c |}
\hline
$\mathbf{\Lambda}$  & Evolution of state in the process of sequential placement \\ \hline
$0.01$ & 1,    0.5,    0.34,    0.27 \\ \hline 
$0.1$ &    1,    0.5,    0.34,    0.26, 0.21,  0.18, 0.16, 0.14, 0.13, 0.12,  0.12,   $\cdots$      \\  \hline
$0.5$ &   1,    0.5,    0.34,    0.26, 0.21,  0.18, 0.16, 0.14, 0.13,  \\
       &0.12,    0.11, 0.1,   0.1, 0.1, $\cdots$    \\ \hline
$2$ & 1,    0.5,    0.34,    0.26, 0.21,  0.18, 0.16, 0.14, 0.13,  0.12,    \\
     &0.11, 0.1,   0.1, 0.1, $\cdots$ \\  \hline
$5$ & 1,    0.5,    0.34,    0.26, 0.21,  0.18, 0.16, 0.14, 0.13,  0.12,    \\
     &0.11, 0.1,   0.1, 0.1, $\cdots$    \\ \hline
\end{tabular}
\caption{Evolution of state in the process of sequential placement on a line of length $10$ for various values of $\Lambda$,
  using the corresponding optimal policies for $\xi=0.001$.}
\label{table:state-evolution_1}
\normalsize
\end{table}

\begin{table}[t!]
\footnotesize
\centering
\begin{tabular}{|c |c |}
\hline
$\mathbf{\Lambda}$  & Optimal normalized distance \\
 & of the nodes from the source \\ \hline
$0.01$ & 10  \\ \hline 
$0.1$ &  0,         0,    1.3430,    4.6280,    7.9130,   10.0000     \\  \hline
$0.5$ &  0,  0,  0,  0.0390,  0.1270,  0.2550,  0.4820,  0.8300,  1.1780,  1.5260,     \\
   &   1.8740, 2.2220, $\cdots$, 9.5300, 9.8780, 10.0000   \\ \hline
$2$ &  0,  0,  0,  0, 0.0080, 0.0220, 0.0440, 0.0670,  0.0950, 0.1270,      \\
   &  0.1590, 0.1910,  0.2230, $\cdots$, 9.9510,  9.9830, 10.0000 \\ \hline 
$5$ & 0,  0,  0,   0,  0,  0.0030,  0.0080, 0.0150,  0.0240,  0.0350,  0.0460,  \\
    &  0.0570, $\cdots$, 9.9790,  9.9900, 10.0000 \\ \hline 
\end{tabular}
\caption{Optimal placement of nodes on a line of length $10$ for various values of $\Lambda$,
  using the corresponding optimal policies for $\xi=0.01$.}
\label{table:effect_of_rho_on_placement_2}
\normalsize
\end{table}

\begin{table}[t!]
\footnotesize
\centering
\begin{tabular}{|c |c |}
\hline
$\mathbf{\Lambda}$  & Evolution of state in the process of sequential placement \\ \hline
$0.01$ & 1\\ \hline 
$0.1$ & 1,    0.5,    0.34,   0.28,    0.28,    0.28 \\  \hline
$0.5$ & 1,    0.5,    0.34,    0.26, 0.21, 0.18,  0.17,  0.16,  0.16,  0.16, $\cdots$  \\ \hline
$2$ & 1,    0.5,    0.34,    0.26, 0.21, 0.18,    0.16, 0.15, 0.14,  0.13,  0.13,  $\cdots$  \\ \hline
$5$ &  1,    0.5,    0.34,    0.26, 0.21, 0.18,    0.16, 0.15, 0.14,  0.13,  0.13,  $\cdots$ \\  \hline

\end{tabular}
\caption{Evolution of state in the process of sequential placement on a line of length $10$ for various values of $\Lambda$,
  using the corresponding optimal policies for $\xi=0.01$.}
\label{table:state-evolution_2}
\normalsize
\end{table}

%

\begin{table}[t!]
\footnotesize
\centering
\begin{tabular}{|c |c |}
\hline
$\mathbf{\Lambda}$  & Optimal normalized distance \\
 & of the nodes from the source \\ \hline
$0.01$ &  10 \\ \hline 
$0.1$ &  5.3060,   10.0000     \\  \hline
$0.5$ &  0,  0.6580,  1.7130, 2.7680, 3.8230,  4.8780,  5.9330,  6.9880,    \\ 
       & 8.0430,  9.0980, 10.0000 \\ \hline
$2$ & 0, 0.0150, 0.1540, 0.3230,  0.5080, 0.6930, 0.8780,  1.0630,   \\ 
    &  1.2480,  1.4330, $\cdots$, 9.7580,  9.9430, 10.0000  \\ \hline 
$5$ & 0,  0.0050, 0.0510,  0.1220,  0.1930, 0.2640,$\cdots$, 9.9910, 10.0000  \\ \hline 
\end{tabular}
\caption{Optimal placement of nodes on a line of length $10$ for various values of $\Lambda$,
  using the corresponding optimal policies for $\xi=0.1$.}
\label{table:effect_of_rho_on_placement_3}
\normalsize
\end{table}

\begin{table}[t!]
\footnotesize
\centering
\begin{tabular}{|c |c |}
\hline
$\mathbf{\Lambda}$  & Evolution of state in the process of sequential placement \\ \hline
$0.01$ & 1\\ \hline 
$0.1$ & 1,    0.63\\  \hline
$0.5$ & 1,  0.5, 0.41,  0.41,  0.41,  $\cdots$ \\ \hline
$2$ &  1, 0.5, 0.35,  0.32,  0.31,  0.31,  0.31,   $\cdots$ \\ \hline
$5$ &  1, 0.5,  0.34, 0.3, 0.3, $\cdots$  \\  \hline

\end{tabular}
\caption{Evolution of state in the process of sequential placement on a line of length $10$ for various values of $\Lambda$,
  using the corresponding optimal policies for $\xi=0.1$.}
\label{table:state-evolution_3}
\normalsize
\end{table}

Figure \ref{fig:control_vs_rho} shows that $a^*$ is decreasing with $\Lambda$, 
for fixed values of $\xi$ and $s$. This is intuitive since increased attenuation in the network will require the relays 
to be placed closer to each other.

Tables \ref{table:effect_of_rho_on_placement_1}, \ref{table:effect_of_rho_on_placement_2} 
and \ref{table:effect_of_rho_on_placement_3} 
illustrate the impromptu placement of relay nodes using our policy, 
for different values of $\Lambda$. In this numerical example we 
applied the policy for $\beta'=1$  and three different values of 
$\xi$, to a line of length $10$. Thus, the policy that we use corresponds to a line having 
exponentially distributed length with mean $1$. If the 
line actually ends at some point before 10, the process would end there with the corresponding number 
of relays (as can be obtained from tables \ref{table:effect_of_rho_on_placement_1}, \ref{table:effect_of_rho_on_placement_2} 
and \ref{table:effect_of_rho_on_placement_3} ) before the end-point having been placed. 

We observe that as $\Lambda$ increases, more relays need to be placed since 
the optimal control decreases with $\Lambda$ for each $s$ 
(see Figure~\ref{fig:control_vs_rho}). On the other hand, the number of relays 
decreases with $\xi$; this can be explained from Figure~\ref{fig:control_vs_xi}. Note that, initially 
one or more relays are placed at or near the source if $a^*(s=1)$ is $0$ or small. But, 
after some relays have been placed, the relays are placed 
equally spaced apart. We see that this happens because,  after a few relays have been placed, the state, $s$, 
does not change, hence, resulting in the relays being subsequently placed equally spaced apart. 
This phenomenon is evident in  
Tables~\ref{table:state-evolution_1}, \ref{table:state-evolution_2} and \ref{table:state-evolution_3}. 
In our present numerical example, the state $s$ will remain unchanged after a relay placement if 
$s=\lceil{\frac{se^{\Lambda a^*(s)}}{0.01(1+se^{\Lambda a^*(s)})×}}\rceil \times 0.01$, 
since we have discretized the state space. 
After some relays are placed, the state becomes equal to one fixed point $s'$ of the function 
$\lceil{\frac{se^{\Lambda a^*(s)}}{0.01(1+se^{\Lambda a^*(s)})×}}\rceil \times 0.01$. Note that 
the deployment starts from $s_{0}:=1$, but for any value of $s_{0}$ (even with $s_{0}$ smaller than $s'$), 
we numerically observe the same phenomenon. Hence, 
$s'$ is an absorbing state. 


\begin{table}[t!]
\footnotesize
\centering
\begin{tabular}{|c |c |c |c |c |c |}
\hline
 & & Average  & Mean& Number of  & Maximum  \\
$\xi$ &   $\Lambda$ & percentage  & number of & cases where & percentage \\ 
 & & error & relays used& no relay & error \\
 & & & &was used & \\ \hline

 0.001   & 0.01     & 0.0068       & 2.0002       & 0       & 0.7698         \\ \hline
 0.001   & 0.1      & 0.3996      &  9.4849      &  0      & 6.8947         \\ \hline
0.001   & 0.5       & 2.0497       & 20.0407       & 0       &  19.5255        \\ \hline
0.001   & 2         &  10.1959      &  60.3740      &  0      &  45.9528        \\ \hline
0.01 & 0.01         &  0      &    0    &   10000     &    0      \\ \hline
 0.01  & 0.1        &  0.3517      & 2.2723       &   0     &   4.6618        \\ \hline
0.01 & 0.5          &  1.5661      & 7.7572       &   0     &   4.7789       \\ \hline
0.01   & 2          &  10.7631      & 36.8230       &   0     &   49.9450        \\ \hline
0.1 & 0.01          &  0      &  0     &   10000     &   0         \\ \hline
0.1 & 0.1           &   0.1259     &   0.0056     &   9944     &    25.9098        \\ \hline
0.1   & 0.5         &   2.9869     &    1.8252    &    0    &    12.5907        \\ \hline
0.1   & 2           &  4.7023      &   7.1530     &    0    &    9.0211        \\ \hline
\end{tabular}
\caption{Comparison of the performance (in terms of $H$; see text) 
of optimal sequential placement over a line of random length, with the optimal placement if the length was known.}
\vspace{-0.5cm}
\label{table:comparison_optimal_mdp}
\normalsize
\end{table}

The numerical experiments reported in Table~\ref{table:comparison_optimal_mdp} 
are a result of asking the following question: how does the cost of an impromptu deployment over a 
line of exponentially distributed length compare with the cost of placing the same number of relays optimally 
over the line, once the length of the line has been revealed? For several combinations of $\Lambda$ and $\xi$, we generated $10000$ 
random numbers independently from an exponential distribution with parameter $\beta'=1$. Each of these numbers was considered as 
a possible realization of the length of the line. Then we computed the placement of the relay nodes 
for each realization by optimal sequential placement policy, which gave us 
$H=\frac{1}{g_{0,1}×}+\sum_{k=2}^{N+1}\frac{(g_{0,k-1}-g_{0,k})}{g_{0,k}g_{0,k-1}\sum_{l=0}^{k-1}\frac{1}{g_{0,l}×}×}$, 
a quantity that we use to evaluate the quality of the relay placement. The significance of $H$ can be 
recalled from (\ref{eqn:capacity_multirelay}) where we found that the capacity $C(1+\frac{P_T/\sigma^2}{H×})$ 
can be achieved if total 
power $P_T$ is available to distribute among the source and the relays; i.e., $H$ can be interpreted as the 
net effective attenuation after power has been allocated optimally over the nodes. 
Also, for each realization, we computed $H$ 
for one-shot optimal relay placement, assuming that the length of the line is known and the number of relays available 
is the {\em same} as the number of relays used by the corresponding sequential placement policy. For a given combination of $\Lambda$ 
and $\xi$, for the $k$-th realization of the length of the line, let us denote the two $H$ values by $H_{\mathsf{sequential}}^{(k)}$ 
and 
$H_{\mathsf{optimal}}^{(k)}$. Then the percentage error for the $k$-th realization is given by:
\begin{equation*}
 e_k:= \frac{|H_{\mathsf{optimal}}^{(k)}-H_{\mathsf{sequential}}^{(k)}|}{H_{\mathsf{optimal}}^{(k)}×} \times 100
\end{equation*}
The average percentage error in Table~\ref{table:comparison_optimal_mdp} is  
the quantity $\frac{\sum_{k=1}^{10000}e_k}{10000×}$. 
The maximum percentage error is the quantity $\max_{k \in \{1,2,\cdots,10000\}}e_k$. 
The following is a discussion of the numerical values in Table~\ref{table:comparison_optimal_mdp}:

\begin{enumerate}[label=(\roman{*})]
\item We observe that, for small enough 
$\xi$, some relays will be placed at the source itself. For example, for $\Lambda=0.01$ and $\xi=0.001$, we will place two relays 
at the source itself (See Table~\ref{table:effect_of_rho_on_placement_1}. 
After placing the first relay, the next state will become $s=0.5$, and 
$a^*(s=0.5)=0$). The state after placing the second relay becomes $s=0.34$, for which $a^*(s=0.34)=8.41$ 
(see the placement in Table~\ref{table:effect_of_rho_on_placement_1}). Now, 
the line having an exponentially distributed length with mean $1$ will end before $a^*(s=0.34)=8.41$ distance with high probability, 
and the probability of placing the third relay will be very small. As a result, the mean number of relays will be $2.0002$. 
In case only $2$ relays are placed by the sequential deployment policy and we seek to place 
$2$ relays optimally 
for the same length of the line (with the length known), the optimal locations for both relays are close to 
the source location 
if the length of the line is small (i.e., if the attenuation $\lambda$ is small, recall the definition of $\lambda$ from 
Section~\ref{subsec:network_propagation_model}). If the line is long (which has a very small probability), 
the optimal placement 
will be significantly different from the sequential placement. Altogether, the error will be small in this case.

\item For $\Lambda=0.01$ and $\xi=0.01$, $a^*(s=1)$ is so large that with probability 
close to $1$ the line will end in a distance less than 
$a^*(s=1)$ and no relay will be placed. In this case the optimal placement and the sequential placement match exactly, and the 
error is $0$. 

\item On an average, the error is large for higher attenuation. At low attenuation, with high probability 
the number of deployed relays is small, and they will be placed at or near the 
source, in which case the optimal placement is close to the sequential placement. 
On the other hand, for high attenuation sequential 
placement uses a considerable number of relays, and the optimal placement may differ significantly from sequential placement. 
Also, we observe (see the last two rows of Table~\ref{table:comparison_optimal_mdp}) 
that even for high attenuation, if $\xi$ is too large, then we will not place any relay, 
which will result in $0$ error.

\end{enumerate}


\vspace{-2mm}
\section{Conclusion}
\label{conclusion}
\vspace{-0.5mm}

We have studied the problem of placing relay nodes along a line, in order to connect a sink at the end of the line to a source 
at the start of the line, in order to  maximize the end-to-end achievable data 
rate. We obtained explicit formulas for the optimal location of a single relay node with individual power constraints on the 
source and the relay under power-law and 
exponential path-loss models. For the multi-relay channel with exponential 
path-loss and sum power constraint, we derived an expression 
for the achievable rate and formulated the optimization problem to 
maximize the end-to-end data rate. We solved explicitly the single
 relay placement problem with exponential path-loss and sum power constraint. 
The asymptotic performance of uniform relay placement 
was also studied for the sum power constraint. Finally, 
a sequential relay placement problem along a line having unknown random length 
was formulated as an MDP whose value function was characterized analytically and the policy structure was investigated 
numerically. 

Our results in this paper are based on information theoretic capacity results, which make several assumptions: 
(i) all the nodes should have global knowledge of the codebooks, channel gains, 
and transmit powers, (ii) relay nodes need to be full-duplex, 
(iii) global synchronization at block level is required. In our ongoing work, 
we are also exploring non-information theoretic, packet forwarding models 
for optimal relay placement, with the aim of obtaining placement algorithms that can be more easily reduced to practice.

\renewcommand{\thesubsection}{\Alph{subsection}}

\appendices

\section{System Model and Notation}\label{appendix:system_model_and_notation}
\subsection{Proof of Theorem \ref{theorem:why_on_a_line}}
Let the source and the sink be located at $(0,0)$ and $(L,0)$. There are $N$ number of relays $1,2,...,N$, where the $i$-th 
relay is placed at $(a_{i},b_{i})$. Now let us consider a new placement of the relays where for each $i \in \{1,2,...,N\}$, the $i$-th 
relay is placed at $(a_{i},0)$. Let us define $a_{0}:=0$ for the source and $a_{N+1}:=L$ for the sink.
 This new placement results in a reduced (not necessarily strictly) distance between any two nodes in the multi-relay channel 
since now there is no variation along $y$-direction. Let us also assume that the relays are enumerated in the 
increasing order of their $x$-coordinate values. Now let $\mathcal{A}:=\{i:1 \leq i \leq N, a_{i}>L\}$, 
$\mathcal{B}:=\{i:1 \leq i \leq N, a_{i}<0\}$ and $\mathcal{C}:=\{i:1 \leq i \leq N, a_{i} \in [0,L]\}$. Let 
$a:=\min_{i,j \in \{0,1,2,...,N,N+1\}, i \neq j} |a_{i}-a_{j}| $. Clearly, $a \leq L $ since 
source and sink are placed $L$ distance apart from each other. If $\mathcal{A}$ is nonempty, place all relays belonging to $\mathcal{A}$ 
uniformly (i.e.,with equal distance between two successive nodes) in the interval $\left[ L-\frac{a}{2×},L \right]$. Similarly, if $\mathcal{B}$ is nonempty, place all relay nodes belonging to 
$\mathcal{B}$ uniformly in the interval $\left[0,\frac{a}{2×}\right]$. Thus we obtain a new set of relay locations where each relay 
is placed on the line segment joining the source and the sink, and the distance between any two relay nodes has been reduced. 
From (\ref{eqn:achievable_rate_multirelay}), it is clear that $R$ increases in the channel gain between each pair of nodes,
for any given power allocation scheme. 
Hence, the optimum relay locations will be such that no relay will be placed outside the line segment joining the source and the 
sink.

\section{Single Relay Node Placement: Node Power Constraints}\label{appendix:exponential_individual}

\subsection{Proof of Theorem \ref{theorem:exponential_individual_power_constraint}}
Observe that in (\ref{eqn:optimize_fd_relay_df}) $g_{0,2}+g_{1,2}+2 \sqrt{(1-\alpha)g_{0,2}g_{1,2}}$ is 
decreasing in $\alpha$
and $\alpha g_{0,1}$ is increasing in $\alpha$. Also,
 $g_{0,2}+g_{1,2}+2 \sqrt{(1-\alpha)g_{0,2}g_{1,2}}|_{\alpha=1}=g_{0,2}+g_{1,2}$,
 which is greater than 
or equal to $\alpha g_{0,1}|_{\alpha=1}$ if $g_{0,1} \leq g_{0,2}+g_{1,2}$.
This motivates us to consider two different cases: $g_{0,1} \leq g_{0,2}+g_{1,2}$ and $g_{0,1} \geq g_{0,2}+g_{1,2}$. We find that 
since $g_{0,1}$ is decreasing in $r$ and 
$g_{1,2}$ is increasing in $r$,
$g_{0,1} \leq g_{0,2}+g_{1,2}$ occurs if and only if $r \geq r_{\mathsf{th}}$, where 
$r_{\mathsf{th}}:=-\frac{1}{\rho×} \log \left(\frac{e^{-\rho L}+\sqrt{e^{-2 \rho L}+4e^{-\rho L}}}{2×}\right)$.

\textbf{Case I: $g_{0,1} \leq g_{0,2}+g_{1,2}$}.
 In this case $\alpha=1$ is the maximizer of $\min \{g_{0,2}+g_{1,2}+2 \sqrt{(1-\alpha)g_{0,2}g_{1,2}} ,\, \alpha g_{0,1} \}$.
So the optimum $R$ in this case becomes $C \bigg(\frac{g_{0,1}P}{\sigma^{2}×}\bigg)$. Further, since $g_{0,1}$ is decreasing in 
$r$, in the $r \geq r_{\mathsf{th}}$ region 
optimum relay location is $\max \{r_{\mathsf{th}},0\}$. 

\textbf{Case II: $g_{0,1} \geq g_{0,2}+g_{1,2}$}.
This holds for $r \leq r_{\mathsf{th}}$. Now as $\alpha$ increases from $0$ to $1$, $\alpha g_{0,1}$ increases from 
$0$ to $g_{0,1}$
and $(g_{0,2}+g_{1,2}+2 \sqrt{(1-\alpha)g_{0,2}g_{1,2}})$ decreases from 
$(g_{0,2}+g_{1,2}+2 \sqrt{g_{0,2}g_{1,2}})$ to $(g_{0,2}+g_{1,2})$. 
Moreover, both $\alpha g_{0,1}$
and $(g_{0,2}+g_{1,2}+2 \sqrt{(1-\alpha)g_{0,2}g_{1,2}})$ are continuous functions of 
$\alpha$. Hence, by 
Intermediate Value Theorem, there exists $\alpha_{\mathsf{th}} \in [0,1]$ such that 
$g_{0,2}+g_{1,2}+2 \sqrt{(1-\alpha)g_{0,2}g_{1,2}}=\alpha g_{0,1}$ at
 $\alpha=\alpha_{\mathsf{th}}$. 
Moreover, $\min \{g_{0,2}+g_{1,2}+2 \sqrt{(1-\alpha)g_{0,2}g_{1,2}} ,\, \alpha g_{0,1} \}$ is maximized 
at $\alpha =\alpha_{\mathsf{th}} \in [0,1]$.

From the condition $g_{0,2}+g_{1,2}+2 \sqrt{(1-\alpha)g_{0,2}g_{1,2}}=\alpha g_{0,1}$ we obtain
\footnotesize
{
\begin{equation}
\alpha_{\mathsf{th}}=\frac{1}{g_{0,1}×}\bigg(\sqrt{g_{0,2}\left(1-\frac{g_{1,2}}{g_{0,1}×}\right)}+\sqrt{g_{1,2}\left(1-\frac{g_{0,2}}{g_{0,1}×}\right)}\bigg)^{2}\label{eqn:alpha_th_expression}
\end{equation}
}
\normalsize

So for $r \leq r_{\mathsf{th}}$ we have 
\begin{equation}
\max_{\alpha \in [0,1]} \min \bigg\{g_{0,2}+g_{1,2}+2 \sqrt{(1-\alpha)
g_{0,2}g_{1,2}} ,\, \alpha g_{0,1} \bigg\} = \alpha_{\mathsf{th}}g_{0,1} \nonumber
\end{equation}

Thus our optimization problem becomes 
\begin{equation}
 \max_{r \in [0,r_{\mathsf{th}}]} \, \sqrt{g_{0,2}\bigg(1-\frac{g_{1,2}}{g_{0,1}×}\bigg)}+\sqrt{g_{1,2}\bigg(1-\frac{g_{0,2}}{g_{0,1}×}\bigg)} \label{eqn:problem_below_threshold}
\end{equation}

Let $k:=e^{-\rho L}$ and $u:=e^{-\rho r}$. Define $u_{\mathsf{th}}=e^{-\rho r_{\mathsf{th}}}$.
Since $u$ strictly decreases with $r$, our problem becomes

\begin{equation}
 \max_{u \in [u_{\mathsf{th}},1]} f(u)\nonumber
\end{equation}
  where
\begin{equation}
 f(u)=\sqrt{k\bigg(1-\frac{k}{u^{2}×}\bigg)}+\sqrt{\frac{k}{u×}\bigg(1-\frac{k}{u×}\bigg)}
\end{equation}
Clearly, $f(u)$ is differentiable w.r.t $u$ over $[u_{\mathsf{th}},1]$. $f'(u)=\frac{2k^{2}/u^{3}}{2 \sqrt{k-k^{2}/u^{2}}×}+
\frac{-k/u^{2}+2k^{2}/u^{3}}{2 \sqrt{k/u-k^{2}/u^{2}}×}$.
 Hence,
$f(\cdot)$ is increasing in $u$ over $[u_{\mathsf{th}},1]$ if and only if its derivative is nonnegative.
 Setting $f'(u)$ to be nonnegative yields the inequality:
\begin{equation}
 \sqrt{\bigg(\frac{1/u-k/u^{2}}{1-k/u^{2}×}  \bigg)} \geq -1+\frac{u}{2k×} \label{eqn:decreasing_condition}
\end{equation}

Now let us consider the case $-1+\frac{u}{2k×} \leq 0$, i.e $u \leq 2k$. In this case $f(\cdot)$ is increasing in $u$. So for 
$r \geq r_{0}$, where $r_{0}=-\frac{1}{\rho ×} \log(u_{0})=-\frac{1}{\rho ×} \log(2k)$, $f(\cdot)$ is decreasing in $r$. 
Of course this 
condition matters only if $r_{0} \leq r_{\mathsf{th}}$,
since, for $r \geq r_{\mathsf{th}}$, $R$ always decreases with $r$. Thus if $r_{0} \leq r_{\mathsf{th}}$, then 
 $R$ decreases with $r$ for $r \in [r_{0}, r_{\mathsf{th}}]$.

Now for $u \geq 2k$, i.e, $r \leq r_{0}$, from (\ref{eqn:decreasing_condition}) we will have $f'(u) \geq 0$ if and only if 
\begin{equation}
 \frac{u^{2}}{4k^{2}×}-\frac{u}{k×}+(1-\frac{1}{4k×}) \leq 0 \nonumber\\
\end{equation}

 So $f(\cdot)$ is increasing in $u$ (decreasing in $r$) if and only if $u \in [u_{-},u_{+}]$ where $u_{-}$ and $u_{+}$ are the two roots 
of the equality version of the above inequality and are given by $u_{-}=2k-\sqrt{k}, u_{+}=2k+\sqrt{k}$. Now let us define
$r_{+}=-\frac{1}{\rho×} \log (u_{+})$.

$u_{-} \leq u_{\mathsf{th}}$ holds if and
only if $3 \sqrt{k}-\sqrt{k+4} \leq 2$. It is easy to check that this always holds since $k \in [0,1]$ in our case. But we are 
maximizing $f(\cdot)$ in the $[u_{\mathsf{th}},1]$ interval. 
So $f(\cdot)$ is increasing in $u$ (decreasing in $r$) if and only if $u \in [u_{\mathsf{th}},u_{+}]$.

Again, $u_{+} \geq u_{\mathsf{th}}$ holds if and only if  
$\sqrt{k+4}-3\sqrt{k}\leq 2$, which holds for any $k \in [0,1]$. So $u_{+} \geq u_{\mathsf{th}}$. Hence, 
$r_{+} \leq r_{\mathsf{th}}$. Thus $f(\cdot)$ decreases with $r$ for $r \in [r_{+},r_{\mathsf{th}}]$.

Also $u_{+}=2k+\sqrt{k}\geq 2k=u_{0}$. Hence, $r_{+} \leq r_{0}$. 

Thus $f(\cdot)$ increases with $r$ for $r \in [0,r_{+}]$ and decreases for $r \geq r_{+}$. If $r_{+} \leq 0$, 
$f(\cdot)$ decreases with $r$ in $[0,L]$.

\textbf{Summary:} 
\begin{enumerate}
 \item We found that $R$ is increasing in $r$ for $r \leq r_{+}$ and decreasing for $r \geq r_{+}$.
Let $r^{*}$ be the optimal position
of the relay node, $x^{*}:=\frac{r^{*}}{L×}$ and $x_{+}:=\frac{r_{+}}{L×}=-\frac{1}{\rho L×} \log(u_{+})$. Then $x^{*}=\max\{x_{+},0\}$, 
where $x_{+}=-\frac{1}{\lambda×} \log \bigg(2e^{-\lambda}+e^{-\frac{\lambda}{2}}\bigg)$. 
 
\item $r_{th} \geq 0$ if and only if $\frac{e^{-\lambda}+\sqrt{e^{-2\lambda}+4e^{-\lambda}}}{2×} \leq 1$, i.e, 
 $\lambda \geq \log 2$. Also $r_{+} \geq 0$ if and only if $2e^{-\lambda}+e^{-\frac{\lambda}{2×}} \leq 1$, i.e., $\lambda \geq \log 4$. 

\item For $\lambda \leq \log 2$, we have $r_{+} \leq 0$ and $r_{th} \leq 0$. So in this case we will place the relay at 
$r^{*}=0 \geq r_{th}$ and hence $\alpha^{*}=1$ and $R^{*}=C(\frac{P}{\sigma^{2}×})$.

\item For $\log 2 \leq \lambda \leq \log 4$, $r_{+} \leq 0$ and $r_{th} \geq 0$. So 
$r^{*}=0$ and hence we place the node at $r^{*}=0 \leq r_{th}$. So $\alpha^{*}=\alpha_{th}|_{r=0}=4e^{-\lambda}(1-e^{-\lambda})$ and hence 
$R^{*}=C \left(\frac{P}{\sigma^{2}×}4e^{-\lambda}(1-e^{-\lambda})\right)$.

\item For $\lambda \geq \log 4$,  
$r^{*}=r_{+} \geq 0$. So substituting the values of $r_{+}$ in (\ref{eqn:alpha_th_expression}), we obtain the expression of 
$\alpha^{*}$ (given in (\ref{eqn:optimum_alpha})) and thus of $R^{*}$ too 
(given in (\ref{eqn:theorem_capacity_high_attenuation})).  
\end{enumerate}

\qed

\begin{rem}
 Let $x_{\mathsf{th}}=\frac{r_{\mathsf{th}}}{L}$. Then $\lim_{\lambda \to \infty}{x_{\mathsf{th}}}=\frac{1}{2×}$ and $\lim_{\lambda \to \infty}{x_{+}}=\frac{1}{2×}$.
\end{rem}

\begin{rem}
 $\alpha_{\mathsf{th}}$ increases with $r$ in the region $r \in [0,r_{\mathsf{th}}]$.
\end{rem}
\begin{proof}
 In the region $r \leq r_{\mathsf{th}}$, for any fixed 
$r$, $\max_{\alpha \in [0,1]} \min \bigg\{g_{0,2}+g_{1,2}+2 \sqrt{(1-\alpha)g_{0,2}g_{1,2}} ,\, \alpha g_{0,1} \bigg\}$
is maximized when $g_{0,2}+g_{1,2}+2 \sqrt{(1-\alpha)g_{0,2}g_{1,2}}=\alpha g_{0,1}$.
 Let $q(\alpha,r)=g_{0,2}+g_{1,2}+2 \sqrt{(1-\alpha)g_{0,2}g_{1,2}}-\alpha g_{0,1}$.
Clearly, $q(\alpha,r)$ is strictly decreasing in $\alpha$ and strictly increasing in $r$, and also continuous in both arguments. 
Also in the region $r \leq r_{\mathsf{th}}$, $q(\alpha_{\mathsf{th}}(r),r)=0$.
Let us assume that the operating point is $\{r, \alpha_{\mathsf{th}}(r)\}$ for some $r < r_{\mathsf{th}}$. Now consider some $r'$ 
such that $r' \leq r_{\mathsf{th}},\, r' >r$.
Then $q(r', \alpha_{\mathsf{th}}(r))>0$. So we must increase $\alpha$ from $\alpha_{\mathsf{th}}(r)$ in order to bring 
$q(\alpha, r)$ back to $0$. Hence,
 $\alpha_{\mathsf{th}}(r')>\alpha_{\mathsf{th}}(r)$.  
\end{proof}

\subsection{Proof of Theorem \ref{theorem:power_law} }
As we observed in the exponential path-loss model (see proof of Theorem \ref{theorem:exponential_individual_power_constraint}), 
$\min \left\{g_{0,2}+g_{1,2}+2 \sqrt{(1-\alpha)g_{0,2}g_{1,2}} ,\, \alpha g_{0,1} \right\}$ 
is maximized for 
$\alpha=1$ if $g_{0,1} \leq g_{0,2}+g_{1,2}$. This occurs, assuming the power law path-loss model, if 

\begin{equation}
r^{-\eta} \leq L^{-\eta}+(L-r)^{-\eta}\label{eqn:single_relay_node_power_r_th}
\end{equation}
 i.e, if $r \geq r_{\mathsf{th}}$. Note that $r_{\mathsf{th}}<\frac{L}{2×}$. 

In the region $r \geq r_{\mathsf{th}}$, $g_{0,1}$
decreases with $r$, and hence the optimum placement point in the interval $[r_{\mathsf{th}},L]$ is $r_{\mathsf{th}}$.

In the region $r \leq r_{\mathsf{th}}$, the problem is similar to the one of 
case II in the proof of Theorem \ref{theorem:exponential_individual_power_constraint}. 
Thus we arrive at the following problem:

\begin{equation}
 \max_{x \in (0, x_{\mathsf{th}}]} f(x)\nonumber\\
\end{equation}

where $x :=\frac{x}{L×}$, $x_{\mathsf{th}}:=\frac{r_{\mathsf{th}}}{L×}$ and 
\begin{equation}
 f(x)=\sqrt{1-\frac{(1-x)^{-\eta}}{x^{-\eta}×}}+\sqrt{(1-x)^{-\eta}(1-\frac{1}{x^{-\eta}×})}
\end{equation}

Now $f'(x)\geq 0$ if and only if:
\begin{equation}
 x^{-\eta+1}-1 \geq \sqrt{\frac{(1-x)^{-\eta}-(\frac{1}{x×}-1)^{-\eta}}{1-(\frac{1}{x×}-1)^{-\eta}×}}
\end{equation}

Since $x < \frac{1}{2×}$, $\eta >1$ and the square roots are well-defined, the condition above can be written as $f_{1}(x)\geq f_{2}(x)$ where 
$f_{1}(x)=(x^{-\eta+1}-1)^{2}(1-(\frac{1}{x×}-1)^{-\eta})$ and $f_{2}(x)=(1-x)^{-\eta}-(\frac{1}{x×}-1)^{-\eta}$. Now 
 it is easy to check that $f_{1}(x)$ 
decreases from $+ \infty$ to $0$ as $x$ increases from $0$ to $\frac{1}{2×}$. $f_{2}(x)$ increases from $0$ to $(2^{\eta}-1)$
as $x$ increases from $0$ to $\frac{1}{2×}$. So there is a unique $p \in [0,\frac{1}{2×}]$, such that $f (x)$ is increasing 
in $x$ for $x \in [0, p]$ and decreasing for $x \geq p$. 

Let us now compare $p$ and $x_{\mathsf{th}}$. Recall that $x_{\mathsf{th}} \leq \frac{1}{2×}$, 
and observe from (\ref{eqn:single_relay_node_power_r_th}) that $x_{\mathsf{th}}$ verifies:
$x_{\mathsf{th}}^{-\eta}=1+(1-x_{\mathsf{th}}^{-\eta})$,
i.e., $(\frac{1}{x_{\mathsf{th}}×})^{-\eta}+(\frac{1}{x_{\mathsf{th}}×}-1)^{-\eta}=1$. Hence, we have:
\begin{eqnarray}
 f_{1}(x_{\mathsf{th}})=\frac{(x_{\mathsf{th}}^{-\eta+1}-1)^{2}}{x_{\mathsf{th}}^{-\eta}×}\nonumber\\
f_{2}(x_{\mathsf{th}})=\frac{(x_{\mathsf{th}}^{-\eta}-1)^{2}}{x_{\mathsf{th}}^{-\eta}×}
\end{eqnarray}
Clearly, since $x_{\mathsf{th}}<1$ and $\eta>1$ we will always have $f_{1}(x_{\mathsf{th}})<f_{2}(x_{\mathsf{th}})$. Thus we must have 
$p<x_{\mathsf{th}}$.
This implies that $R$ increases with $r$ in $[0,pL]$, and decreases with $r$ in $[pL,L]$. So
$R$ is maximized by placing the relay node at $r^{*}=pL$, where $p$ is obtained by solving $f_{1}(x)=f_{2}(x)$.

For the ``modified power law path-loss'' model, if $2b <L$, then $g_{0,1}$ is constant for $r \in [0,b]$ but $g_{1,2}$ increases with $r$ for 
$r \in [0,b]$. So if $pL<b$, we can improve $R$ by placing the relay at $b$. Hence, $x^{*}=\max \{p,\frac{b}{L×}\}$ in this case.

\qed

 \vspace{-0.5cm}
 \section{Multiple Relay Placement : Sum Power Constraint}
 \label{appendix:total_power_constraint}

\subsection{Proof of Theorem \ref{theorem:multirelay_capacity}}
We want to maximize $R$ given in (\ref{eqn:achievable_rate_multirelay}) subject to the total power constraint, assuming fixed 
relay locations. 
Let us consider $C \bigg(\frac{1}{\sigma^{2}×} \sum_{j=1}^{k} ( \sum_{i=0}^{j-1} h_{i,k} \sqrt{P_{i,j}}  )^{2}\bigg)$, i.e., the $k$-th 
term in the argument of $\min \{\cdots\}$ in (\ref{eqn:achievable_rate_multirelay}). By the monotonicity of $C(\cdot)$, 
it is sufficient to consider $\sum_{j=1}^{k} ( \sum_{i=0}^{j-1} h_{i,k} \sqrt{P_{i,j}}  )^{2}$. Now since the channel gains are multiplicative, 
we have:
\begin{equation}
 \sum_{j=1}^{k} ( \sum_{i=0}^{j-1} h_{i,k} \sqrt{P_{i,j}}  )^{2}=g_{0,k}\sum_{j=1}^{k} \bigg( \sum_{i=0}^{j-1}\frac{\sqrt{P_{i,j}}}{h_{0,i}×}\bigg)^{2}\nonumber
\end{equation}
Thus our optimization problem becomes:

\begin{eqnarray}
 & & \max \, \min_{k \in \{1,2,\cdots,N+1\}} g_{0,k} \sum_{j=1}^{k} \bigg( \sum_{i=0}^{j-1}\frac{\sqrt{P_{i,j}}}{h_{0,i}×}\bigg)^{2}\nonumber\\
& \textit{s.t} & \, \sum_{j=1}^{N+1}\gamma_{j} \leq P_{T}\nonumber\\
& \textit{and}& \,\sum_{i=0}^{j-1}P_{i,j}=\gamma_{j} \, \forall \, j \in \{1,2,\cdots,N+1\}\label{eqn:optimization_problem}
\end{eqnarray}

Let us fix $\gamma_{1}, \gamma_{2},\cdots,\gamma_{N+1}$ such that their sum is equal to $P_{T}$. We observe that $P_{i,N+1}$ for $i \in \{0,1,\cdots,N\}$ appear 
in the objective function
only once: for $k=N+1$ through the term $\bigg( \sum_{i=0}^{N}\frac{\sqrt{P_{i,N+1}}}{h_{0,i}×}\bigg)^{2}$.
Since we have fixed $\gamma_{N+1}$, we need to maximize this term over $P_{i,N+1},\, i \in \{0,1,\cdots,N\}$. So we have the following optimization problem:
\begin{eqnarray}
& & \max \sum_{i=0}^{N} \frac{\sqrt{P_{i,N+1}}}{h_{0,i}×}\nonumber\\
& s.t & \, \sum_{i=0}^{N} P_{i,N+1}=\gamma_{N+1}\label{eqn:problem}
\end{eqnarray}
By Cauchy-Schwartz inequality, the objective function in this optimization problem is upper bounded by 
$\sqrt{(\sum_{i=0}^{N} P_{i,N+1})(\sum_{i=0}^{N}\frac{1}{g_{0,i}×})}=\sqrt{\gamma_{N+1}\sum_{i=0}^{N}\frac{1}{g_{0,i}}}$ (using the fact
 that $g_{0,i}=h_{0,i}^2$ $\forall i \in \{0,1,\cdots,N\}$). The upper bound is achieved if 
there exists some $c>0$ such that $\frac{\sqrt{P_{i,N+1}}}{\frac{1}{h_{0,i}×}×}=c$ $\forall i \in \{0,1,\cdots,N\}$. So we have
\begin{equation}
 P_{i,N+1}=\frac{c^2}{g_{0,i}×} \,\, \forall i \in \{0,1,\cdots,N\}\nonumber\\
\end{equation}
Using the fact that $\sum_{i=0}^{N}P_{i,N+1}=\gamma_{N+1}$, we obtain $c^2=\frac{\gamma_{N+1}}{\sum_{l=0}^{N}\frac{1}{g_{0,l}×}×}$.
 Thus, 

\begin{equation}
 P_{i,N+1}=\frac{\frac{1}{g_{0,i}×}}{\sum_{l=0}^{N}\frac{1}{g_{0,l}×}×}\gamma_{N+1}
\end{equation}
Here we have used the fact that $h_{0,0}=1$. Now $\{P_{i,N}: i=0,1,\cdots,(N-1)\}$ appear only through the sum $\sum_{i=0}^{N-1}\frac{\sqrt{P_{i,N}}}{h_{0,i}×}$,
 and it appears twice: for $k=N$ and $k=N+1$. We need to maximize this sum subject to the constraint $\sum_{i=0}^{N-1}P_{i,N}=\gamma_{N}$. This optimization 
can be solved in a similar way as before. Thus by repeatedly using this argument and solving optimization problems similar 
in nature to (\ref{eqn:problem}), we obtain:
\begin{equation}
 P_{i,j}=\frac{\frac{1}{g_{0,i}×}}{\sum_{l=0}^{j-1}\frac{1}{g_{0,l}×}×}\gamma_{j} \,\, \forall 0 \leq i < j \leq (N+1)
\end{equation}
Substituting for $P_{i,j},0 \leq i < j \leq (N+1)$ in (\ref{eqn:optimization_problem}), we obtain the following optimization problem:
\begin{eqnarray}
& & \max \min_{k \in \{1,2,\cdots,N+1\}} g_{0,k} \sum_{j=1}^{k}\bigg(\gamma_{j} \sum_{i=0}^{j-1}\frac{1}{g_{0,i}×} \bigg)\nonumber\\
& \textit{s.t.} & \,\, \sum_{j=1}^{N+1}\gamma_{j} \leq P_{T}
\end{eqnarray}

Let us define $b_{k}:=g_{0,k}$ and $a_{j}:=\sum_{i=0}^{j-1}\frac{1}{g_{0,i}×}$. Observe that $b_{k}$ is decreasing and $a_{k}$ is 
increasing with $k$. Let us define:
\begin{equation}
 \tilde{s}_{k}(\gamma_{1},\gamma_{2},\cdots,\gamma_{N+1}) := b_{k} \sum_{j=1}^{k} a_{j} \gamma_{j}
\end{equation}
With this notation, our optimization problem becomes: 
\begin{eqnarray}
& & \max \min_{k \in \{1,2,\cdots,N+1\}} \tilde{s}_{k}(\gamma_{1},\gamma_{2},\cdots,\gamma_{N+1})\nonumber\\
& \textit{s.t.} & \,\, \sum_{j=1}^{N+1}\gamma_{j} \leq P_{T} \label{eqn:modified_optimization_problem}
\end{eqnarray}

\begin{claim}
 Under optimal allocation of $\gamma_{1}, \gamma_{2},\cdots,\gamma_{N+1}$ for the optimization problem 
(\ref{eqn:modified_optimization_problem}),
  $\tilde{s}_{1}=\tilde{s}_{2}=\cdots=\tilde{s}_{N+1}$.
\end{claim}
\begin{proof}
 The optimization problem (\ref{eqn:modified_optimization_problem}) can be rewritten as:
\begin{eqnarray}
&& \max \zeta \nonumber\\
\textit{s.t} && \zeta \leq b_{k}\sum_{j=1}^{k}a_{j}\gamma_{j} \, \forall \, k \in \{1,2,\cdots,N+1\}, \nonumber\\
&& \sum_{j=1}^{N+1}\gamma_{j} \leq P_{T}, \nonumber\\
&&\gamma_{j} \geq 0 \,\, \forall \, j \in \{1,2,\cdots,N+1\}\label{eqn:equality_problem_primal}
\end{eqnarray}
The dual of this linear program is given by:
\begin{eqnarray}
 &&\min P_{T}\theta \nonumber\\
\textit{s.t} && \sum_{k=1}^{N+1}\mu_{k}=1,\nonumber\\
&& a_{l}\sum_{k=l}^{N+1}b_{k}\mu_{k}+\nu_{l}=\theta \, \forall \, l \in \{1,2,\cdots,N+1\},\nonumber\\
&& \theta \geq 0, \nonumber\\
&& \mu_{l} \geq 0, \nu_{l} \geq 0 \, \forall \, l \in \{1,2,\cdots,N+1\} \label{eqn:equality_problem_dual}
\end{eqnarray}
Now let us consider a primal feasible solution $(\{\gamma_j^*\}_{1 \leq j \leq N+1}, \zeta^*)$ which satisfies:
\begin{eqnarray}
&&  b_{k}\sum_{j=1}^{k}a_{j}\gamma_{j}^*=\zeta^* \, \forall \, k \in \{1,2,\cdots,N+1\}, \nonumber\\
&& \sum_{j=1}^{N+1}\gamma_{j}^* = P_{T}
\end{eqnarray}
Thus we have, $b_1 a_1 \gamma_1^*=\zeta^*$, i.e., $\gamma_1^*=\frac{\zeta^*}{b_1 a_1×}$. Again, $b_2(a_1 \gamma_1^*+a_2 \gamma_2^*)=\zeta^*$, 
which implies:
\begin{equation}
 \frac{b_2}{b_1×}\zeta^*+b_2 a_2 \gamma_2^*=\zeta^*\nonumber\\
\end{equation}
Thus we obtain $\gamma_2^*=\frac{\zeta^*}{a_2×}(\frac{1}{b_2×}-\frac{1}{b_1×})$. In general, we can write:
\begin{equation}
 \gamma_k^*=\frac{\zeta^*}{a_k×}\left(\frac{1}{b_k×}-\frac{1}{b_{k-1}×}\right) \, \forall k \in \{1,2,\cdots,N+1\}\nonumber\\
\end{equation}
with $\frac{1}{b_{0}×}:=0$. Now from the condition $\sum_{k=1}^{N+1}\gamma_k^*=P_T$, we obtain:
\footnotesize
\begin{eqnarray}
  \zeta^*&=&\frac{P_{T}}{\sum_{k=1}^{N+1}\frac{1}{a_{k}×}(\frac{1}{b_k×}-\frac{1}{b_{k-1}×})×}\nonumber\\
 \gamma_j^*&=&\frac{ \frac{1}{a_j×} \left(\frac{1}{b_j×}-\frac{1}{b_{j-1}×}\right) }
{\sum_{k=1}^{N+1}\frac{1}{a_{k}×}(\frac{1}{b_k×}-\frac{1}{b_{k-1}×})×}P_{T}, \, j \in \{1,2,\cdots,N+1\}
\label{eqn:primal_optimal}
\end{eqnarray}
\normalsize
It should be noted that since $b_{k}$ is nonincreasing in $k$, the primal variables above are nonnegative and satisfies feasibility 
conditions.
Again, let us consider a dual feasible solution $(\{\mu_j^*,\nu_j^*\}_{1 \leq j \leq N+1}, \theta^*)$ which satisfies:
\begin{eqnarray}
&& \nu_l^*=0 \, \forall \, l \in \{1,2,\cdots,N+1\}\nonumber\\
&& \sum_{k=1}^{N+1}\mu_{k}^*=1\nonumber\\
&& a_{l}\sum_{k=l}^{N+1}b_{k}\mu_{k}^*+\nu_{l}^*=\theta^* \, \forall \, l \in \{1,2,\cdots,N+1\}
\end{eqnarray}
Solving these equations, we obtain:
\footnotesize
\begin{eqnarray}
 \theta^*&=&\frac{1}{\sum_{k=1}^{N+1}\frac{1}{b_k×}\left(\frac{1}{a_k×}-\frac{1}{a_{k+1}×}\right)×}\nonumber\\
 \mu_j^*&=&\frac{\frac{1}{b_j×}(\frac{1}{a_j×}-\frac{1}{a_{j+1}×})}
{\sum_{k=1}^{N+1}\frac{1}{b_k×}\left(\frac{1}{a_k×}-\frac{1}{a_{k+1}×}\right)×},\, j \in \{1,2,\cdots,N+1\}
\label{eqn:dual_optimal}
\end{eqnarray}
\normalsize
where $\frac{1}{a_{N+2}×}:=0$. Since $a_k$ is increasing in $k$, 
all dual variables are feasible. It is easy to check that $\zeta^*=P_T \theta^*$, which means that there is no duality gap for the chosen 
primal and dual variables. Since the primal is a linear program, the solution $(\gamma_1^*, \gamma_2^*,\cdots,\gamma_{N+1}^*, \zeta^*)$ 
is primal optimal. Thus we have established the claim, since the primal optimal solution satisfies it. 
\end{proof}

So let us obtain $\gamma_{1},\gamma_{2},\cdots,\gamma_{N+1}$ for which $\tilde{s}_{1}=\tilde{s}_{2}=\cdots=\tilde{s}_{N+1}$. 
Putting $\tilde{s}_{k}=\tilde{s}_{k-1}$, we obtain : 

\begin{equation}
 b_{k} \sum_{j=1}^{k} a_{j} \gamma_{j}=b_{k-1} \sum_{j=1}^{k-1} a_{j} \gamma_{j}
\end{equation}

Thus we obtain:

\begin{equation}
 \gamma_{k}=\frac{(b_{k-1}-b_{k})}{b_{k}×} \frac{1}{a_{k}×} \sum_{j=1}^{k-1}a_{j}\gamma_{j}
\end{equation}

Let $d_{k}:=\frac{(b_{k-1}-b_{k})}{b_{k}×} \frac{1}{a_{k}×}$. Hence,

\begin{equation}
 \gamma_{k}=d_{k} \sum_{j=1}^{k-1}a_{j}\gamma_{j}
\end{equation}

From this recursive equation we have:

\begin{equation}
 \gamma_{2}=d_{2}a_{1}\gamma_{1}
\end{equation}
\begin{equation}
 \gamma_{3}=d_{3}(a_{1}\gamma_{1}+a_{2}\gamma_{2})=d_{3}a_{1}(1+a_{2}d_{2})\gamma_{1}
\end{equation}
and in general for $k \geq 3$,
\begin{equation}
 \gamma_{k}=d_{k}a_{1}\Pi_{j=2}^{k-1}(1+a_{j}d_{j})\gamma_{1}\label{eqn:gamma_k_gamma_1}
\end{equation}

Using the fact that $\gamma_{1}+\gamma_{2}+\cdots+\gamma_{N+1}=P_{T}$, we obtain: 

\begin{equation}
 \gamma_{1}=\frac{P_{T}}{1+d_{2}a_{1}+ \sum_{k=3}^{N+1}d_{k}a_{1} \Pi_{j=2}^{k-1}(1+a_{j}d_{j}) ×} \label{eqn:gamma_1}
\end{equation}

\footnotesize
\begin{figure*}[!t]
 \begin{eqnarray}\label{eqn:capacity_increasing_with_N}
 & & z_{1}^{*}+\frac{z_{2}^{*}-z_{1}^{*}}{1+z_{1}^{*}×}+\cdots+\frac{e^{y}-z_{i}^{*}}{1+z_{1}^{*}+\cdots+z_{i}^{*}×}+\frac{z_{i+1}^{*}-e^{y}}{1+z_{1}^{*}+\cdots+z_{i}^{*}+e^{y}×}+\cdots.+\frac{e^{\rho L}-z_{N}^{*}}{1+z_{1}^{*}+\cdots+z_{i}^{*}+e^{y}+z_{i+1}^{*}+\cdots+z_{N}^{*}×}\nonumber\\
&  & < z_{1}^{*}+\frac{z_{2}^{*}-z_{1}^{*}}{1+z_{1}^{*}×}+\cdots+\frac{e^{\rho L}-z_{N}^{*}}{1+z_{1}^{*}+\cdots+z_{i}^{*}+z_{i+1}^{*}+\cdots+z_{N}^{*}×}
\end{eqnarray}
\hrule
\end{figure*}
\normalsize

Thus if $\tilde{s}_{1}=\tilde{s}_{2}=\cdots=\tilde{s}_{N+1}$, there is a unique 
allocation $\gamma_{1},\gamma_{2},\cdots,\gamma_{N+1}$. So this must be the one 
maximizing $R$. Hence, optimum $\gamma_{1}$ is obtained by (\ref{eqn:gamma_1}). Then, 
substituting the values of $\{a_{k}:k=0,1,\cdots,N\}$ 
and $d_{k}:k=1,2,\cdots,N+1$ in (\ref{eqn:gamma_k_gamma_1}) and (\ref{eqn:gamma_1}), 
we obtain the values of $\gamma_{1},\gamma_{2},\cdots,\gamma_{N+1}$ 
as shown in Theorem~\ref{theorem:multirelay_capacity}.

Now under these optimal values of $\gamma_{1},\gamma_{2},\cdots,\gamma_{N+1}$, all terms in the argument of $\min \{\cdots\}$ 
in (\ref{eqn:achievable_rate_multirelay}) are equal. So we can consider the first term alone.
 Thus we obtain the expression for $R$ optimized over power allocation among all the nodes 
for fixed relay locations as : $R_{P_T}^{opt}(y_1,y_2,\cdots,y_N)=C \left(\frac{g_{0,1}P_{0,1}}{\sigma^{2}×}\right)=C \left(\frac{g_{0,1}\gamma_{1}}{\sigma^{2}×}\right)$. 
Substituting the expression for $\gamma_{1}$ from (\ref{eqn:gamma_one}), we obtain the achievable rate 
formula (\ref{eqn:capacity_multirelay}).

\qed

\subsection{Proof of Theorem \ref{theorem:single_relay_total_power}}

Here we want to place the relay node at a distance $r_{1}$ from the source 
 to minimize $\bigg\{\frac{1}{g_{0,1}×}+\frac{g_{0,1}-g_{0,2}}{g_{0,2}(1+g_{0,1})×}\bigg\}$ 
(see \ref{eqn:capacity_multirelay}). Hence, our optimization problem becomes :
\begin{equation}
\min_{r_{1} \in [0,L]} \bigg\{e^{\rho r_{1}}+\frac{e^{-\rho r_{1}}-e^{-\rho L}}{e^{-\rho L}(1+e^{-\rho r_{1}})×}\bigg\}\nonumber\\
\end{equation}
Writing $z_{1}=e^{\rho r_{1}}$, the problem becomes :
\begin{equation}
 \min_{z_{1} \in [1,e^{\rho L}]}  \bigg\{z_{1}-1+\frac{e^{\rho L}+1}{z_{1}+1×}\bigg\}\nonumber\\
\end{equation}
This is a convex optimization problem. Equating the derivative of the objective function to zero, we obtain
 $1-\frac{e^{\rho L}+1}{(z_{1}+1)^{2}×}=0$. Thus the derivative becomes zero at $z_{1}'=\sqrt{1+e^{\rho L}}-1>0$. 
Hence, the objective function is decreasing in $z_{1}$ for $z_{1} \leq z_{1}'$ and increasing in $z_{1} \geq z_{1}'$. 
So the minimizer is $z_{1}^{*}=\max \{z_{1}',1 \}$. 
So the optimum distance of 
the relay node from the source is $y_{1}^{*}=r_{1}^{*}=\max \{0,r_{1}' \}$, where $r_{1}'=\frac{1}{\rho×} \log (\sqrt{1+e^{\rho L}}-1)$. 
Hence, $\frac{y_{1}^{*}}{L×}=\max \{\frac{1}{\lambda×} \log \left(\sqrt{e^{\lambda}+1}-1 \right),0\}$. Now 
$r_{1}' \geq 0$ if and only if $\lambda \geq \log 3$. Hence, $\frac{y_{1}^{*}}{L×}=0$ for $\lambda \leq \log 3$ and 
$\frac{y_{1}^{*}}{L×}=\frac{1}{\lambda×} \log \left(\sqrt{e^{\lambda}+1}-1 \right)$ for $\lambda \geq \log 3$.

{\em For $\lambda \leq \log 3$}, the relay is placed at the source. Then $g_{0,1}=1$ and $g_{0,2}=g_{1,2}=e^{-\lambda}$. 
Then $P_{0,1}=\gamma_{1}=\frac{2P_{T}}{e^{\lambda}+1×}$ (by Theorem \ref{theorem:multirelay_capacity}) 
and $R^{*}=C \left(\frac{2P_{T}}{(e^{\lambda}+1)\sigma^{2}×}\right)$. Also $\gamma_{2}=\frac{e^{\lambda}-1}{e^{\lambda}+1×}P_{T}$. 
Hence, $P_{0,2}=P_{1,2}=\frac{e^{\lambda}-1}{e^{\lambda}+1×}\frac{P_{T}}{2×}$.

{\em for $\lambda \geq \log 3$}, the relay is placed at $r_{1}'$. Substituting the value of $r_{1}'$ into (\ref{eqn:power_gamma_relation}), 
we obtain $P_{0,1}=\gamma_{1}=\frac{P_{T}}{2×}$, $P_{0,2}=\frac{1}{\sqrt{e^{\lambda}+1}×}\frac{P_{T}}{2×}$, 
$P_{1,2}=\frac{\sqrt{e^{\lambda}+1}-1}{\sqrt{e^{\lambda}+1}×}\frac{P_{T}}{2×}$. So in this case
$R^{*}=C \left(\frac{g_{0,1}P_{0,1}}{\sigma^{2}×} \right)$. Since $P_{0,1}=\frac{P_{T}}{2×}$, we have 
$R^{*}=C \left( \frac{1}{\sqrt{e^{\lambda}+1}-1×}\frac{P_{T}}{2 \sigma^{2}×} \right)$.
\qed

\subsection{Proof of Theorem \ref{theorem:capacity_increasing_with_N}}
\vspace{-2 mm}

For the $N$-relay problem, let the minimizer in (\ref{eqn:multirelay_optimization}) be $z_{1}^{*}, z_{2}^{*},\cdots,z_{N}^{*}$ and let 
$y_{k}^{*}=\frac{1}{\rho×} \log z_{k}^{*}$. Clearly, 
there exists $i \in \{0,1,\cdots,N\}$ such that $y_{i+1}^{*}>y_{i}^{*}$. Let us insert a new relay at a distance $y$ from the source such 
that $y_{i}^{*}<y<y_{i+1}^{*}$. Now we find that we can easily reach (\ref{eqn:capacity_increasing_with_N}) just by simple comparison. For example, 
\begin{eqnarray*}
& & \frac{e^{y}-z_{i}^{*}}{1+z_{1}^{*}+\cdots+z_{i}^{*}×}+\frac{z_{i+1}^{*}-e^{y}}{1+z_{1}^{*}+\cdots+z_{i}^{*}+e^{y}×} \\
& < & \frac{e^{y}-z_{i}^{*}}{1+z_{1}^{*}+\cdots+z_{i}^{*}×}+\frac{z_{i+1}^{*}-e^{y}}{1+z_{1}^{*}+\cdots+z_{i}^{*}×}\\
& =& \frac{z_{i+1}^{*}-z_{i}^{*}}{1+z_{1}^{*}+\cdots+z_{i}^{*}×}
\end{eqnarray*}
First $i$ terms in the summations of L.H.S (left hand side) and R.H.S (right hand side) 
of (\ref{eqn:capacity_increasing_with_N}) are identical. Also sum of the remaining terms in L.H.S 
is smaller than that of the R.H.S since there is an additional $e^{y}$ in the denominator of each fraction for the L.H.S. Hence, 
we can justify (\ref{eqn:capacity_increasing_with_N}). Now 
R.H.S is precisely the optimum objective function for the $N$-relay placement problem (\ref{eqn:multirelay_optimization}). On the 
other hand, L.H.S is a particular value of the objective in (\ref{eqn:multirelay_optimization}), for $(N+1)$-relay placement problem.
This clearly implies that by adding one additional relay we can strictly 
improve from $R^{*}$ of the $N$ relay channel. Hence, $R^{*}(N+1)>R^{*}(N)$.
\qed

\subsection{Proof of Theorem \ref{theorem:G_increasing_in_lambda}}
\vspace{-2 mm}
Consider the optimization problem (\ref{eqn:multirelay_optimization}). Let us consider $\lambda_1$, $\lambda_2$, with 
$\lambda_1<\lambda_2$, the respective 
minimizers being $(z_{1}^{*},\cdots,z_{N}^{*})$ and $(z_{1}',\cdots,z_{N}')$. Clearly,
\begin{eqnarray}
 G(N,\lambda_1)=\frac{e^{\lambda_1}}{ z_{1}^{*}+\sum_{k=2}^{N+1} \frac{z_{k}^{*}-z_{k-1}^{*}}{\sum_{l=0}^{k-1} z_{l}^{*}}×}
\end{eqnarray}
with $z_{N+1}^{*}=e^{\lambda_1}$ and $z_{0}^{*}=1$. 
With $N \geq 1 $, note that 
$z_{1}^{*}-\frac{z_{N}^{*}}{1+z_{1}^{*}+\cdots+z_{N}^{*}×} \geq 0$, since $z_{1}^{*} \geq 1$ and  
$\frac{z_{N}^{*}}{1+z_{1}^{*}+\cdots+z_{N}^{*}×} \leq 1$. Hence, it is easy to see that 
$\frac{e^{\lambda}}{ z_{1}^{*}-\frac{z_{N}^{*}}{1+z_{1}^{*}+\cdots+z_{N}^{*}×}+\sum_{k=2}^{N} \frac{z_{k}^{*}-z_{k-1}^{*}}{\sum_{l=0}^{k-1} z_{l}^{*}}+\frac{e^{\lambda}}{1+z_{1}^{*}+\cdots+z_{N}^{*}×}×}$ is increasing in 
$\lambda$ where $(z_{1}^{*},\cdots,z_{N}^{*})$ is the optimal solution of (\ref{eqn:multirelay_optimization}) with 
$\lambda=\lambda_1$. Hence,
\begin{eqnarray}
 G(N,\lambda_1)&=&\frac{e^{\lambda_1}}{ z_{1}^{*}+\sum_{k=2}^{N} \frac{z_{k}^{*}-z_{k-1}^{*}}{\sum_{l=0}^{k-1} z_{l}^{*}}+\frac{e^{\lambda_1}-z_{N}^{*}}{\sum_{l=0}^{N} z_{l}^{*}}×}\nonumber\\
&\leq & \frac{e^{\lambda_2}}{ z_{1}^{*}+\sum_{k=2}^{N} \frac{z_{k}^{*}-z_{k-1}^{*}}{\sum_{l=0}^{k-1} z_{l}^{*}}+\frac{e^{\lambda_2}-z_{N}^{*}}{1+z_{1}^{*}+\cdots+z_{N}^{*}×}×}\nonumber\\
&\leq & \frac{e^{\lambda_2}}{ z_{1}^{'}+\sum_{k=2}^{N} \frac{z_{k}^{'}-z_{k-1}^{'}}{\sum_{l=0}^{k-1} z_{l}^{'}}+\frac{e^{\lambda_2}-z_{N}^{'}}{1+z_{1}^{'}+\cdots+z_{N}^{'}×}×}\nonumber\\
&=& G(N,\lambda_2)
\end{eqnarray}
The second inequality follows from the fact that $(z_{1}',\cdots,z_{N}')$ minimizes 
$z_{1}+\sum_{k=2}^{N+1} \frac{z_{k}-z_{k-1}}{\sum_{l=0}^{k-1} z_{l}}$ subject to the constraint 
$1 \leq z_1 \leq z_2 \leq \cdots \leq z_N \leq z_{N+1} =e^{\lambda_2}$.

Hence, $G(N,\lambda)$ is increasing in $\lambda$ for fixed $N$.
\qed

\subsection{Proof of Theorem \ref{theorem:large_nodes_uniform}}

When $N$ relay nodes are uniformly placed along a line, we will have $y_{k}=\frac{kL}{N+1×}$. Then our formula for achievable rate $R_{P_T}^{opt}(y_1,y_2,\cdots,y_N)$
for total power constraint becomes :
$R_{N}=C(\frac{P_{T}}{\sigma^{2}×}\frac{1}{f(N)×})$ where $f(N)=a+\sum_{k=2}^{N+1}\frac{a^{k}-a^{k-1}}{1+a+\cdots+a^{k-1}×}$ with 
$a=e^{ \frac{\rho L}{N+1×}}$.
\\ Since $a>1$, we have : 

\begin{eqnarray*}
f(N)&=& a+\sum_{k=1}^{N}\frac{a^{k+1}-a^{k}}{1+a+\cdots+a^{k}×}\\
&=& a+ (a-1)^{2}\sum_{k=1}^{N} \frac{a^{k}}{a^{k+1}-1×}\\
& \geq & a+ (a-1)^{2} \sum_{k=1}^{N} \frac{a^{k}}{a^{k+1}×}\\
& =& a+ (a-1)^{2} \frac{N}{a×}
\end{eqnarray*}

Thus we obtain:

\begin{eqnarray*}
& & \liminf_N f(N)\\
& \geq & \lim_{N \rightarrow \infty} \left(e^{ \frac{\rho L}{N+1×}}+\frac{N}{e^{ \frac{\rho L}{N+1×}}×}(e^{ \frac{\rho L}{N+1×}}-1)^{2}\right)\\
&= & 1 
\end{eqnarray*}

On the other hand, noting that $a>1$, we can write:

\begin{eqnarray*}
f(N)&=& a+ (a-1)^{2}\sum_{k=1}^{N} \frac{a^{k}}{a^{k+1}-1×}\\
& \leq & a+ (a-1)^{2} \sum_{k=1}^{N} \frac{a^{k+1}}{a^{k+1}-1×}\\
& = & a+ (a-1)^{2} \sum_{k=1}^{N} \left(1+\frac{1}{a^{k+1}-1×}\right)\\
& \leq & a+ (a-1)^{2} \sum_{k=1}^{N} \left(1+\frac{1}{a^{k+1}-a^{k}×}\right)\\
& =& a+N (a-1)^{2}+(a-1)\sum_{k=1}^{N} a^{-k}\\
&=& a+N (a-1)^{2}+(a-1)a^{-1}\frac{1-a^{-N}}{1-a^{-1}×}
\end{eqnarray*}

So 

\begin{eqnarray*}
 & & \limsup_N f(N)\\
& \leq & \lim_{N \rightarrow \infty} \left(e^{ \frac{\rho L}{N+1×}}+N(e^{ \frac{\rho L}{N+1×}}-1)^{2}+ (1-e^{ \frac{-\rho NL}{N+1×}})\right)\\
&=& 2-e^{-\rho L}
\end{eqnarray*}

Thus we have $\limsup_N R_{N} \leq C \left(\frac{P_{T}}{\sigma^{2}×}\right)$ and 
$\liminf_N R_{N} \geq C\left(\frac{P_{T}/\sigma^{2}}{2-e^{-\rho L}×}\right)$. Hence, the theorem is proved.
\qed

\section{Optimal Sequential Placement of Relay Nodes on a Line of Random Length}\label{appendix:sequential_placement_total_power}

As we have seen in Section \ref{sec:mdp_total_power}, our problem is a negative dynamic programming problem 
(i.e., the $\mathsf{N}$ case of \cite{schal75conditions-optimality}, where single-stage rewards are non-positive). It is to be noted 
that Sch\"{a}l \cite{schal75conditions-optimality} discusses two other kind of problems as well: 
the $\mathsf{P}$ case (single-stage rewards are positive) 
and the $\mathsf{D}$ case (the reward at stage $k$ is discounted by a factor $\alpha^k$, where $0<\alpha<1$). In this appendix, 
we first state a general-purpose theorem for the value iteration, prove it by some results of \cite{schal75conditions-optimality}, 
and then we use this theorem to prove Theorem \ref{theorem:convergence_of_value_iteration}.

%

\subsection{A General Result (Derived from \cite{schal75conditions-optimality})}\label{appendix_subsection_schal-discussion}

Consider an infinite horizon total cost MDP whose state 
space $\mathcal{S}$ is an interval in $\mathbb{R}$ and the action space 
$\mathcal{A}$ is $[0,\infty)$. Let the set of possible actions at state $s$ be denoted by $\mathcal{A}(s)$. 
Let the single-stage cost be $c(s,a,w) \geq 0$ where $s$, $a$ and $w$ are the state, 
the action and the disturbance, respectively. Let us denote the optimal expected cost-to-go at state $s$ by $V^*(s)$. Let the state of 
the system evolve as $s_{k+1}=h(s_k,a_k,w_k)$, where $s_k$, $a_k$ and $w_k$ are the state, the action and 
the disturbance at the $k$-th 
instant, respectively. Let $s^{*} \in \mathcal{S}$ be an absorbing state with $c(s^{*},a, w)=0$ for all $a$, $w$. 
Let us consider the value iteration for all $s \in \mathcal{S}$:

\footnotesize
\begin{eqnarray}
V^{(0)}(s)&=& 0 \nonumber\\
 V^{(k+1)}(s)&=&\inf_{a \in [0,\infty)} \mathbb{E}_{w} \bigg( c(s,a,w)+V^{(k)}(h(s,a,w)) \bigg), s \neq s^{*}\nonumber\\
V^{(k+1)}(s^{*})&=& 0\label{eqn:value_iteration_general}
\end{eqnarray}
\normalsize  

We provide some results and concepts from \cite{schal75conditions-optimality}, which will be used later to prove 
Theorem \ref{theorem:convergence_of_value_iteration}.

\begin{thm}\label{thm:schal_convergence_value_iteration}
       [{\em Theorem 4.2 (\cite{schal75conditions-optimality})}] $V^{(k)}(s)\rightarrow V^{(\infty)}(s)$ for all $s \in \mathcal{S}$, 
i.e., the value iteration (\ref{eqn:value_iteration_general}) converges.
\end{thm}

Let us recall that $\Gamma_k(s)$ is the set of minimizers of (\ref{eqn:value_iteration_general}) at the 
$k$-th iteration at state $s$, if the infimum is achieved at some $a<\infty$. 
$\Gamma_{\infty}(s):=\{a \in \mathcal{A}:a$ is an 
accumulation point of some sequence $\{a_k\}$ where each $a_k \in \Gamma_{k}(s)\}$. $\Gamma^*(s)$ is the set of minimizers in 
the Bellman Equation.

\begin{thm}\label{thm:value_iteration_general}

 If the value iteration (\ref{eqn:value_iteration_general}) satisfies the following conditions:
\begin{enumerate}[label=(\alph*)]
\item For each $k$, $\mathbb{E}_{w}\bigg(c(s,a,w)+V^{(k)}(h(s,a,w))\bigg)$ is jointly continuous in $a$ and $s$ for $s \neq s^{*}$.
\item The infimum in (\ref{eqn:value_iteration_general}) is achieved in $[0,\infty)$ for all $s \neq s^{*}$.
\item For each $s \in \mathcal{S}$, there exists $a(s)<\infty$ such that $a(s)$ is continuous in $s$ for $s \neq s^{*}$, 
and no minimizer of (\ref{eqn:value_iteration_general}) lies in $(a(s), \infty)$ for each $k \geq 0$.

 \end{enumerate}

Then the following hold:
\begin{enumerate}[label=(\roman{*})]
 \item The value iteration converges, i.e., $V^{(k)}(s) \rightarrow V^{(\infty)}(s)$ for all $s \neq s^{*}$.
\item $V^{(\infty)}(s)=V^*(s)$ for all $s \neq s^{*}$.
\item $\Gamma_{\infty}(s) \subset \Gamma^*(s)$ for all $s \neq s^{*}$.
\item There is a stationary optimal policy $f^{\infty}$ where the mapping $f:\mathcal{S} \setminus \{s^{*}\} \rightarrow \mathcal{A}$ and $f(s) \in \Gamma_{\infty}(s)$ 
for all $s \neq s^{*}$.
\end{enumerate}

\end{thm}
\qed

Let $\mathcal{C}(\mathcal{A})$ be the set of nonempty compact subsets of $\mathcal{A}$. 
 The Hausdorff metric $d$ on $\mathcal{C}(\mathcal{A})$ is defined as follows:
\begin{equation*}
 d(C_1,C_2)=\max \{ \sup_{c \in C_1}\rho(c,C_2), \, \sup_{c \in C_2}\rho(c,C_1) \}
\end{equation*}
where $\rho(c,C)$ is the minimum distance between the point $c$ and the compact set $C$. 

\begin{prop}\label{prop:separable_hausdorff}
     [Proposition 9.1(\cite{schal75conditions-optimality})]
$(\mathcal{C}(\mathcal{A}),d)$ is a separable metric space. 
 \end{prop}

A mapping $\phi:\mathcal{S}\rightarrow \mathcal{C}(\mathcal{A})$ is called measurable if it is measurable with respect to
the $\sigma$-algebra of Borel subsets of $(\mathcal{C}(\mathcal{A}),d)$.

 $\hat{\mathcal{F}}(\mathcal{S} \times \mathcal{A})$ is the set of all measurable functions 
$v:\mathcal{S} \times \mathcal{A} \rightarrow \mathbb{R}$ which are bounded below and where every such $v(\cdot)$ is the limit 
of a non-decreasing sequence of measurable, bounded functions $v_k:\mathcal{S} \times \mathcal{A} \rightarrow \mathbb{R}$.

We will next present a condition followed by a theorem. The condition, if satisfied, implies 
the convergence of value iteration (\ref{eqn:value_iteration_general}) 
to the optimal value function (according to the theorem).

\begin{condition} \label{condition_A}
[{\em Derived from Condition A in \cite{schal75conditions-optimality}}]
\begin{enumerate}[label=(\roman{*})]
 \item $\mathcal{A}(s)\in \mathcal{C}(\mathcal{A})$ for all $s \in \mathcal{S}$ and 
$\mathcal{A}:\mathcal{S}\rightarrow \mathcal{C}(\mathcal{A})$ is measurable.
\item $\mathbb{E}_{w}(c(s,a,w)+V^{(k)}(h(s,a,w))) $ is in $\hat{\mathcal{F}}(\mathcal{S}\times \mathcal{A})$ for all $k \geq 0$.
\end{enumerate}    
\end{condition}
\qed

\begin{thm}\label{thm:schal_main_theorem}
 [{\em Theorem $13.3$, \cite{schal75conditions-optimality}}] If $c(s,a,w) \geq 0$ for all $s,a,w$ and Condition \ref{condition_A} 
holds:
\begin{enumerate}[label=(\roman{*})]
 \item $V^{(\infty)}(s)=V^*(s)$, $s \in \mathcal{S}$.
\item $\Gamma_{\infty}(s) \subset \Gamma^*(s)$.
\item There is a stationary optimal policy $f^{\infty}$ where $f:\mathcal{S} \rightarrow \mathcal{A}$ and $f(s) \in \Gamma_{\infty}(s)$ 
for all $s \in \mathcal{S}$.
\end{enumerate}
 \end{thm}

\qed

But the requirement of compact action sets in Theorem \ref{thm:schal_main_theorem} is often too restrictive. The next condition 
and theorem deal with the situation where the action space is noncompact.

\begin{condition} \label{condition_B}
[{\em Condition B (\cite{schal75conditions-optimality})}] There is a measurable mapping 
$\underline{\mathcal{A}}:\mathcal{S} \rightarrow \mathcal{C}(\mathcal{A})$ 
 such that:
\begin{enumerate}[label=(\roman{*})]
 \item $\underline{{\mathcal{A}}}(s)\subset {\mathcal{A}}(s)$ for all $s \in \mathcal{S}$.
\item \small{$\inf_{a \in \mathcal{A}(s)-\underline{\mathcal{A}}(s)} \mathbb{E}_{w}\bigg(c(s,a,w)+V^{(k)}(h(s,a,w))\bigg)> \inf_{a \in \mathcal{A}(s)}\mathbb{E}_{w}\bigg(c(s,a,w)+V^{(k)}(h(s,a,w))\bigg)$} 
for all $k \geq 0$.
\end{enumerate}
 \end{condition}
\qed

This condition requires that for each state $s$, there is a compact set $\underline{\mathcal{A}}(s)$ of actions such that no optimizer 
 of the value iteration lies outside the set $\underline{\mathcal{A}}(s)$ at any stage $k \geq 0$.

\begin{thm}\label{thm:schal_compact_action}
       [{\em Theorem $17.1$, \cite{schal75conditions-optimality}}] If Condition \ref{condition_B} is satisfied 
and if the three statements in Theorem \ref{thm:schal_main_theorem} are valid for the modified problem having 
admissible set of actions $\underline{\mathcal{A}}(s)$ for each state $s \in \mathcal{S}$, 
then those statements are valid for the original problem as well.
      \end{thm}
\qed

\textbf{\em Proof of Theorem \ref{thm:value_iteration_general}:}

By Theorem \ref{thm:schal_convergence_value_iteration}, the value iteration converges, i.e., 
$V^{(k)}(s)\rightarrow V^{(\infty)}(s)$. Moreover, $V^{(k)}(s)$ 
is the optimal cost for a $k$-stage problem with zero terminal cost, 
and the cost at each stage is positive. Hence, $V^{(k)}(s)$ 
increases in $k$ for every $s \in \mathcal{S}$. Thus, for all $s \in \mathcal{S}$, 
\begin{equation}
 V^{(k)}(s) \uparrow V^{(\infty)}(s)
\end{equation}


Now, Condition \ref{condition_B} and Theorem \ref{thm:schal_compact_action} 
say that if no optimizer of the value iteration in each stage $k$ lies outside a compact subset $\underline{\mathcal{A}}(s)$ 
of $\mathcal{A}(s) \subset \mathcal{A}$, 
then we can deal with the modified problem having a new action space $\underline{\mathcal{A}}(s)$. 
If the value iteration converges to the optimal 
value in this modified problem, then it will converge to the optimal 
value in the original problem as well, provided that the mapping $\underline{\mathcal{A}}:\mathcal{S}\rightarrow \mathcal{C}(\mathcal{A})$ 
is measurable. Let us choose $\underline{\mathcal{A}}(s):=[0,a(s)]$, where $a(s)$ satisfies hypothesis (c) of 
Theorem \ref{thm:value_iteration_general}. Since $a(s)$ is continuous at $s \neq s^{*}$, 
for any $\epsilon>0$ 
we can find a $\delta_{s,\epsilon}>0$ such that $|a(s)-a(s')|<\epsilon$ whenever $|s-s'|<\delta_{s,\epsilon}$, $s \neq s^{*}$, 
$s' \neq s^{*}$. Now, when  
$|a(s)-a(s')|<\epsilon$, we have $d([0,a(s)],[0,a(s')])<\epsilon$. Hence, the mapping 
$\underline{\mathcal{A}}:\mathcal{S}\rightarrow \mathcal{C}(\mathcal{A})$ is 
continuous at all $s \neq s^{*}$, and thereby measurable in this case. Hence, 
the value iteration (\ref{eqn:value_iteration_general}) satisfies Condition \ref{condition_B}.

Thus, the value iteration for $s \neq s^{*}$ can be modified as follows:
\begin{equation}
  V^{(k+1)}(s)=\inf_{a \in [0,a(s)]} \mathbb{E}_{w} \bigg( c(s,a,w)+V^{(k)}(h(s,a,w)) \bigg)\label{eqn:modified_value_iteration_general}
\end{equation}

Now, $\mathbb{E}_{w} \bigg( c(s,a,w)+V^{(k)}(h(s,a,w)) \bigg)$ is continuous (can be discontinuous at $s=s^{*}$, since this quantity 
is $0$ at $s=s^*$) 
on $\mathcal{S} \times \mathcal{A}$ (by our hypothesis). Hence, $\mathbb{E}_{w} \bigg( c(s,a,w)+V^{(k)}(h(s,a,w)) \bigg)$ 
is measurable on $\mathcal{S} \times \mathcal{A}$. Also, it is bounded below by $0$. 
Hence, it can be approximated by an increasing sequence of bounded measurable functions $\{v_{n,k}\}_{n \geq 1}$ given by 
$v_{n,k}(s,a)=\min \bigg \{\mathbb{E}_{w} \bigg( c(s,a,w)+V^{(k)}(h(s,a,w)) \bigg),\,n \bigg \}$. Hence, 
$\mathbb{E}_{w} \bigg( c(s,a,w)+V^{(k)}(h(s,a,w)) \bigg)$ is in $\hat{\mathcal{F}}(\mathcal{S} \times \mathcal{A})$.


Thus, Condition \ref{condition_A} is satisfied for the modified problem and therefore, by Theorem \ref{thm:schal_main_theorem}, 
the modified value iteration in (\ref{eqn:modified_value_iteration_general}) converges to the optimal value function. 
Now, by Theorem 
\ref{thm:schal_compact_action}, we can argue that the value iteration (\ref{eqn:value_iteration_general}) converges 
to the optimal value function in the original problem and hence $V^{(\infty)}(s)=V^*(s)$ for all $s \in \mathcal{S} \setminus s^*$. 
Also, $\Gamma_{\infty}(s) \subset \Gamma^*(s)$ for all $s \in \mathcal{S} \setminus s^*$ and there exists a stationary 
optimal policy $f^{\infty}$ where $f(s) \in \Gamma_{\infty}(s)$ for all $s \in \mathcal{S} \setminus s^*$ 
(by Theorem \ref{thm:schal_main_theorem}).

\qed

\subsection{Proof of Theorem \ref{theorem:convergence_of_value_iteration}}\label{appendix_subsection-convergence-value-iteration-proof}

This proof uses the results of Theorem \ref{thm:value_iteration_general} provided in this appendix. Remember that the state 
$\mathbf{e}$ is absorbing and $c(\mathbf{e}, a, w)=0$ for all $a$, $w$. We can think of it as state $0$ so that our state space 
becomes $[0,1]$ which is a Borel set. We will see that the state $0$ plays the role of the state $s^*$ 
as mentioned in Theorem \ref{thm:value_iteration_general}.

We need to check whether the conditions (a), (b), and (c) in Theorem \ref{thm:value_iteration_general} 
are satisfied for the value iteration (\ref{eqn:value_iteration}). Of course, 
$J_{\xi}^{(0)}(s)=0$ is concave, increasing in $s \in (0,1]$. Suppose that 
$J_{\xi}^{(k)}(s)$ is concave, increasing in $s$ for some $k \geq 0$. 
Also, for any fixed $a \geq 0$, $\frac{se^{\rho a}}{1+se^{\rho a}×}$ 
is concave and increasing in $s$. Thus, by the composition rule for the composition of a concave increasing 
function $J_{\xi}^{(k)}(\cdot)$ and a concave increasing function $\frac{se^{\rho a}}{1+se^{\rho a}×}$, 
for any $a \geq 0$ the term $J_{\xi}^{(k)}\left(\frac{se^{\rho a}}{1+se^{\rho a}×}\right)$ 
is concave, increasing over $s \in (0,1]$. Hence, 
$\int_{0}^{a}\beta e^{-\beta z}s(e^{\rho z}-1)dz + e^{-\beta a}\bigg(s(e^{\rho a}-1)+\xi+J_{\xi}^{(k)}\left(\frac{se^{\rho a}}{1+se^{\rho a}×}\right)\bigg)$ 
 (in (\ref{eqn:value_iteration})) is concave increasing over $s \in (0,1]$. Since the 
infimization over $a$ preserves concavity, we conclude that 
$J_{\xi}^{(k+1)}(s)$ is concave, increasing over $s \in (0,1]$. 
Hence, for each $k$, $J_{\xi}^{(k)}(s)$ is continuous in $s$ over $(0,1)$, since otherwise concavity w.r.t. $s$ will be violated. 
Now, we must have $J_{\xi}^{(k)}(1) \leq \lim_{s \uparrow 1} J_{\xi}^{(k)}(s)$, since otherwise 
the concavity of $J_{\xi}^{(k)}(s)$ will be violated. 
But since $J_{\xi}^{(k)}(s)$ 
is increasing in $s$, $J_{\xi}^{(k)}(1) \geq \lim_{s \uparrow 1} J_{\xi}^{(k)}(s)$. Hence, $J_{\xi}^{(k)}(1) = \lim_{s \uparrow 1}  J_{\xi}^{(k)}(s)$. 
Thus, $J_{\xi}^{(k)}(s)$ is continuous in $s$ over $(0,1]$ for each $k$.

Hence, $\int_{0}^{a}\beta e^{-\beta z}s(e^{\rho z}-1)dz + e^{-\beta a}\bigg(s(e^{\rho a}-1)+\xi+J_{\xi}^{(k)}(\frac{se^{\rho a}}{1+se^{\rho a}×}) \bigg)$ 
is continuous in $s,a$ for $s \neq 0$. Hence, condition (a) in Theorem \ref{thm:value_iteration_general} 
is satisfied. 

Now, we will check condition (c) in Theorem \ref{thm:value_iteration_general}.

By Theorem \ref{thm:schal_convergence_value_iteration}, the value iteration converges, i.e., 
$J_{\xi}^{(k)}(s)\rightarrow J_{\xi}^{(\infty)}(s)$. Also, $J_{\xi}^{(\infty)}(s)$ is concave, increasing 
in $s \in (0,1]$ and hence continuous. Moreover, $J_{\xi}^{(k)}(s)$ 
is the optimal cost for a $k$-stage problem with zero terminal cost, 
and the cost at each stage is positive. Hence, $J_{\xi}^{(k)}(s)$ 
increases in $k$ for every $s \in (0,1]$. Thus, for all $s \in (0,1]$, 
\begin{equation}
 J_{\xi}^{(k)}(s) \uparrow J_{\xi}^{(\infty)}(s)
\end{equation}
Again, $J_{\xi}^{(k)}(s)$ is the optimal cost for a $k$-stage problem with zero terminal cost. Hence, it is less than or equal to 
the optimal cost for the infinite horizon problem with the same transition law and cost structure. Hence, 
$J_{\xi}^{(k)}(s)\leq J_{\xi}(s)$ for all $k \geq 1$. Since $J_{\xi}^{(k)}(s)\uparrow J_{\xi}^{(\infty)}(s)$, 
we have $J_{\xi}^{(\infty)}(s) \leq J_{\xi}(s)$.

Now, consider the following two cases:

\subsubsection{$\beta>\rho$}

Let us define a function $\psi:(0,1]\rightarrow \mathbb{R}$ by 
$\psi(s)=\frac{J_{\xi}^{(\infty)}(s)+\theta s}{2×}$. By Proposition \ref{prop:upper_bound_on_cost_beta_geq_rho}, 
$J_{\xi}(s)<\theta s$ for all $s \in (0,1]$. Hence, $J_{\xi}^{(\infty)}(s)<\psi(s)<\theta s$ and $\psi(s)$ 
is continuous over $s \in (0,1]$.
Since $\beta>\rho$ and $J_{\xi}^{(k)}(s) \in [0, \theta]$ for any 
$s$ in $(0,1]$, the expression $\theta s + e^{-\beta a}\bigg(-\theta s e^{\rho a}+\xi+J_{\xi}^{(k)}\left(\frac{se^{\rho a}}{1+se^{\rho a}×}\right)\bigg)$ 
obtained from the R.H.S of (\ref{eqn:value_iteration}) converges to $\theta s$ as $a \rightarrow \infty$. 
A lower bound to this expression is $\theta s + e^{-\beta a}(-\theta s e^{\rho a})$.
With $\beta > \rho$, for each $s$, there exists 
$a(s)<\infty$ such that $\theta s + e^{-\beta a}(-\theta s e^{\rho a})> \psi(s)$ for 
all $a > a(s)$. 
But $\theta s + \inf_{a \geq 0} e^{-\beta a}\bigg(-\theta s e^{\rho a}+\xi+J_{\xi}^{(k)}\left(\frac{se^{\rho a}}{1+se^{\rho a}×}\right)\bigg)$ 
is equal to $J_{\xi}^{(k+1)}(s)<\psi(s)$. 
Hence, for any $s \in (0,1]$, the minimizers for (\ref{eqn:bellman_equation_simplified_in_a}) always 
lie in the compact interval $[0,a(s)]$ for all $k \geq 1$.
 Since $\psi(s)$ is continuous in $s$, 
we can choose $a(s)$ as a continuous function of $s$ on $(0,1]$.

\subsubsection{$\beta \leq \rho$}
Fix $A$, $0 <A <\infty$. Let $K:=\frac{1}{\beta A ×}\left(\xi+(e^{\rho A}-1)  \right)+(e^{\rho A}-1)$. 
Then, by Proposition \ref{prop:upper_bound_on_cost}, $J_{\xi}(s) \leq K$ for all $s \in (0,1]$. Now, we observe that the objective function 
(for minimization over $a$) in the R.H.S of (\ref{eqn:value_iteration}) is lower bounded by 
$\int_{0}^{a}\beta e^{-\beta z}s(e^{\rho z}-1)dz$, which is continuous in $s,a$ and goes to $\infty$ as $a \rightarrow \infty$ for 
each $s \in (0,1]$. Hence, 
for each $s \in (0,1]$, there exists $0<a(s)<\infty$ such that 
$\int_{0}^{a}\beta s e^{-\beta z}(e^{\rho z}-1)dz >2K$ 
for all $a>a(s)$ and $a(s)$ is continuous over $s \in (0,1]$. But $J_{\xi}^{(k+1)}(s) \leq J_{\xi}(s) \leq K$ for 
all $k$. Hence, the minimizers in (\ref{eqn:value_iteration}) always lie in $[0,a(s)]$ 
where $a(s)$ is independent of $k$ and continuous over $s \in (0,1]$.

Let us set $a(0)=a(1)$.\footnote{Remember that at state $0$ (i.e., state $\mathbf{e}$), 
the single stage cost is $0$ irrespective of the action, 
and that this state is absorbing. Hence, any action at state $0$ can be optimal.} Then, the chosen function $a(s)$ is 
continuous over $s \in (0,1]$ and can be discontinuous only at $s=0$.
Thus, condition (c) of Theorem \ref{thm:value_iteration_general} has 
been verified for the value iteration (\ref{eqn:value_iteration}). Condition (b) of Theorem \ref{thm:value_iteration_general} is 
obviously satisfied since a continuous function over a compact set always has a minimizer.
\qed

{\em Remark:} Observe that in our value iteration (\ref{eqn:value_iteration}) it is always sufficient 
to deal with compact action spaces, and the objective 
functions to be minimized at each stage of the value iteration are continuous in $s$, $a$. Hence, $\Gamma_{k}(s)$ is nonempty 
for each $s \in (0,1]$, $k \geq 0$.
Also, since there exists $K>0$ such that $J_{\xi}^{(k)}(s) \leq K$ 
for all $k \geq 0$, $s \in (0,1]$, it is sufficient to 
restrict the action space in (\ref{eqn:value_iteration}) 
to a set $[0, a(s)]$ for any $s \in (0,1]$, $k \geq 0$. Hence, $\Gamma_k (s) \subset [0,a(s)]$ 
for all $s \in (0,1]$, $k \geq 0$. 
Now, for a fixed $s \in (0,1]$, any sequence $\{a_k\}_{k \geq 0}$ with $a_k \in \Gamma_k (s)$, in bounded. Hence, the sequence 
must have a limit point. Hence, $\Gamma_{\infty}(s)$ is nonempty for each $s \in (0,1]$. Since $\Gamma_{\infty}(s) \subset \Gamma^*(s)$,
 $\Gamma^*(s)$ is nonempty for each $s \in (0,1]$.

\subsection{Proofs of Propositions \ref{prop:increasing_concave_in_s}, \ref{prop:increasing_concave_in_lambda} and \ref{prop:continuity_of_cost}}\label{appendix_subsection_policy-structure}

\subsubsection{Proof of Proposition \ref{prop:increasing_concave_in_s}}
Fix $\xi$. Consider the value iteration (\ref{eqn:value_iteration}).
Let us start with $J_{\xi}^{(0)}(s):=0$ for all $s \in (0,1]$. Clearly, $J_{\xi}^{(1)}(s)$ is concave and increasing in $s$, 
since pointwise infimum of linear functions is concave. Now let us assume that $J_{\xi}^{(k)}(s)$ 
is concave and increasing in $s$. Then, by the 
composition rule, it is easy to show that $J_{\xi}^{(k)}(\frac{se^{\rho a}}{1+se^{\rho a}×})$ is concave and 
increasing in $s$ for any fixed $a\geq 0$. Hence, 
$J_{\xi}^{(k+1)}(s)$ is concave and increasing, since pointwise infimum of a set of concave and increasing functions is 
concave and increasing. 
By Theorem \ref{theorem:convergence_of_value_iteration}, $J_{\xi}^{(k)}(s)\rightarrow J_{\xi}(s)$. Hence, $J_{\xi}(s)$ is concave and increasing in $s$.
\qed

\subsubsection{Proof of Proposition \ref{prop:increasing_concave_in_lambda}}
Consider the value iteration (\ref{eqn:value_iteration}). Since $J_{\xi}^{(0)}(s):=0$ for all $s \in (0,1]$,  
$J_{\xi}^{(1)}(s)$ is obtained by taking infimum (over $a$) of a linear, increasing function of $\xi$. 
Hence, $J_{\xi}^{(1)}(s)$ is concave, increasing over $\xi \in (0, \infty)$. 
If we assume that $J_{\xi}^{(k)}(s)$ is concave and increasing 
in $\xi$, then $J_{\xi}^{(k)}(\frac{se^{\rho a}}{1+se^{\rho a}×})$ is also concave and increasing 
in $\xi$ for fixed $s$ and $a$. Thus, $J_{\xi}^{(k+1)}(s)$ 
is also concave and increasing in $\xi$. Now, $J_{\xi}^{(k)}(s) \rightarrow J_{\xi}(s)$ for all $s \in \mathcal{S}$, 
 and $J_{\xi}^{(k)}(s)$ is concave, increasing in $\xi$ for all $k \geq 0$, $s \in \mathcal{S}$. 
Hence, $J_{\xi}(s)$ is concave and increasing in $\xi$.
\qed

\subsubsection{Proof of Proposition \ref{prop:continuity_of_cost}}
Clearly, $J_{\xi}(s)$ is continuous in $s$ over $(0,1)$, since otherwise concavity w.r.t. $s$ will be violated. 
Now, since $J_{\xi}(s)$ is concave in $s$ over $(0,1]$, we must have $J_{\xi}(1) \leq \lim_{s \uparrow 1}J_{\xi}(s)$. 
But since $J_{\xi}(s)$ 
is increasing in $s$, $J_{\xi}(1) \geq \lim_{s \uparrow 1}J_{\xi}(s)$. Hence, 
$J_{\xi}(1) = \lim_{s \uparrow 1}J_{\xi}(s)$. 
Thus, $J_{\xi}(s)$ is continuous in $s$ over $(0,1]$.

Again, for a fixed $s \in (0,1]$, $J_{\xi}(s)$ is concave and increasing in $\xi$. Hence, $J_{\xi}(s)$ is continuous in 
$\xi$ over $\xi \in (0,c),\, \forall \, c>0$. Hence, $J_{\xi}(s)$ is continuous in $\xi$ over $(0,\infty)$.
\qed

%
%
%

\bibliographystyle{IEEEtran}
\bibliography{IEEEabrv,chattopadhyay-etal12capacity-max-relay-placement-line}

\remove{

}

\end{document}